\documentclass[12pt,a4paper]{article}

\usepackage[english]{babel}

\usepackage[a4paper,top=2cm,bottom=2cm,left=2.5cm,right=2.5cm,marginparwidth=1.75cm]{geometry}

\usepackage[style=authoryear-comp, backend=biber, sortcites=false]{biblatex}
\usepackage{amsmath}
\allowdisplaybreaks
\usepackage{amssymb}
\usepackage{amsthm}
\usepackage{graphicx}
\usepackage[colorlinks=true, allcolors=blue]{hyperref}
\usepackage[title]{appendix}
\usepackage{mathrsfs}
\usepackage{amsfonts}
\usepackage{caption}
\usepackage{authblk}
\usepackage{csquotes}
\usepackage[T1]{fontenc} 
\usepackage{lipsum}
\usepackage{float}

\usepackage{setspace}
\usepackage{titlesec}
\titleformat{\section} % Redefine section numbering format
  {\normalfont\Large\bfseries}{\thesection.}{1em}{}
  
\theoremstyle{plain}
\newtheorem{theorem}{Theorem}
\newtheorem{lemma}[theorem]{Lemma}
\newtheorem{proposition}[theorem]{Proposition}
\newtheorem{corollary}[theorem]{Corollary}
\newtheorem{claim}{Claim}
\newtheorem{assumption}{Assumption}

\theoremstyle{definition}
\newtheorem{definition}{Definition}
\newtheorem{example}{Example}

\DeclareMathOperator*{\argmax}{argmax} % Ou \DeclareMathOperator*{\argmax}{arg\,max}
\DeclareMathOperator*{\argmin}{argmin} % Ou \DeclareMathOperator*{\argmax}{arg\,max}

\title{Depth Projections and the Local Nature of the Cass Criterion}

\author[1,2]{Leandro Lyra Braga Dognini\thanks{I am grateful to Aloisio Araujo and the participants of the Mathematical Economics Seminar at the Institute for Pure and Applied Mathematics (IMPA) for all insightful comments on this manuscript. The paper has also benefited greatly from the comments of anonymous referees. All figures were drawn and the revision was done with the assistance of Claude, Anthropic's large language model.}}
\affil[1]{\small Department of Economics, Rio de Janeiro State University}
\affil[2]{\small Legislative Advisory, Federal Senate of Brazil}
\date{September 28, 2026}

\begin{document}
\maketitle
\begin{abstract}
\noindent This paper defines the depth projections of the indifference and offer hypersurfaces, and of the transformation frontier. These projections are functions conceived to measure curvature (thus are not numbers, as the Gaussian curvature is) and their second-order behavior is governed by the Hessians of the expenditure and profit functions. In particular, up to second order, the depth projection of the offer hypersurface is always exactly twice that of the indifference hypersurface. In an overlapping generations economy with production, these projections lead to the derivation of general statements for the necessity and sufficiency of the Cass criterion in the form $\sum^{\infty}_{t=1}1/(\Vert p_{t}\Vert\sum_{h\in G_{t}}\Vert c^{h}\Vert)=\infty$, which allows unbounded dynamics for the demography and for per capita endowments, and reconciles the Cass criterion with the Balasko--Cass--Shell relabeling algorithm. Furthermore, the assumptions behind these statements reveal the local nature of the Cass criterion: only the behavior of the indifference hypersurfaces and transformation frontiers in a neighborhood of the equilibrium allocation matters.
\end{abstract}

\begin{sloppypar}
\end{sloppypar}


%-------------------------------------------
% Paper Body
%-------------------------------------------
%--- Section ---%
\section{Introduction}\label{sec1}

\textsc{The canonical manifolds} of demand theory are the trade hyperplane, the indifference hypersurface, and the offer hypersurface. A natural and intuitive graphical description of the \textit{local} behavior of these manifolds is the following: the trade hyperplane is the supporting hyperplane of both the indifference and the offer hypersurfaces, which are locally the boundaries of strictly convex sets that are strictly included one within the other.

With this depiction in mind, this paper shows that the indifference and offer hypersurfaces can be seen as curved versions of the trade hyperplane and defines the \textit{depth projections} of these hypersurfaces in order to quantify how bent the trade hyperplane must be to reach them.

The main results of \hyperref[sec2]{Section 2} characterize these depth projections up to second order. First, both vanish, together with their gradients, at the reference relative prices, and their Hessians are given by the Hessian of the expenditure function (i.e., by the Slutsky matrix): the depth projection of the indifference hypersurface is, up to second order, a quadratic form in the Slutsky matrix, and the depth projection of the offer hypersurface is always exactly twice that of the indifference hypersurface, regardless of preferences and prices.

Second, when the depth projection of the indifference hypersurface is regarded as a function of the consumption bundles themselves (instead of relative prices), it is locally bounded from below by a quadratic function of the distance to the reference bundle, and from above by a quadratic function of the distance in a single coordinate when there are two goods. These bounds are the tools used later on to characterize efficient equilibria.

\hyperref[sec3]{Section 3} shows that the same construction applies to the transformation frontier of a strictly convex production set, with the iso-profit hyperplane and the profit function taking over the roles of the trade hyperplane and the expenditure function: the depth projection of the transformation frontier is, up to second order, a quadratic form in the Hessian of the profit function (i.e., in the supply substitution matrix), and it satisfies the same local quadratic lower bound.

The depth projections are related to, but distinct from, the classical measures of curvature used in economic theory, such as the Gaussian curvature (e.g., \textcite[pp. 612--613]{Debreu_1972} and \textcite[pp. 39--40]{Mas-Colell_1985}). The Gaussian curvature of a hypersurface at a point is a number, the determinant of the differential of the Gauss map, which associates each point of the hypersurface with its unit normal vector. The depth projections, on the other hand, are functions that gather curvature behavior in a neighborhood of a point. For instance, loosely speaking, the depth projection of the indifference hypersurface records radii of curvature rather than curvatures: the more curved (i.e., smaller radius) the indifference hypersurface is around the reference bundle, the smaller its depth projection is (see Example \ref{exCES}).

Two features of this distinction are essential for the applications in this paper. First, the depth projection contains directional information that the Gaussian curvature, being a single number, discards. Second, and more importantly, the depth projection is well-behaved on the ``enlarged'' indifference hypersurfaces produced by the relabeling algorithm of \textcite{BalaskoCassShell_1980}, whose Gaussian curvature is identically zero, as discussed in \hyperref[subsecBCS]{Section 4.1}.

The remainder of the paper applies these results to obtain a more general statement of the \textcite{Cass_1972} criterion. It is well known that in overlapping generations economies the First Welfare Theorem fails and equilibria may be inefficient (e.g., \textcite{Samuelson_1958, CassYaari_1966, Shell_1971}). Following \textcite{Cass_1972}, the literature has derived necessary and sufficient conditions to characterize efficient allocations (e.g., \textcite{Benveniste1976,Benveniste_1986,BalaskoShell_1980,OkunoZilcha_1980, GeanakoplosPolemarchakis_1991,RaoAiyagari_1992, Bloise_2008a}) that are mainly based on the asymptotic behavior of the supporting price sequence, in the form of
\begin{eqnarray*}
    \sum^{\infty}_{t=1}\frac{1}{\Vert p_{t}\Vert}=\infty.
\end{eqnarray*}
This statement of the Cass criterion assumes bounded dynamics for both the demography and per capita endowments. \textcite[pp. 1932--1936, Theorems 5a and 5b]{GeanakoplosPolemarchakis_1991} derive a more general version that allows for unbounded demographic dynamics, which is given by 
\begin{eqnarray*}
    \sum^{\infty}_{t=1}\frac{1}{H_{t}\Vert p_{t}\Vert}=\infty,
\end{eqnarray*}
with $H_{t}\in\mathbb{N}$ the number of households born at period $t\geq1$. This paper extends the Cass criterion in two directions. First, it is stated as
\begin{eqnarray}\label{eqGeneralCass}
    \sum^{\infty}_{t=1}\frac{1}{\Vert p_{t}\Vert \sum_{h\in G_{t}}\Vert c^{h}\Vert}=\infty,
\end{eqnarray}
where $G_{t}$ is the set of households born at period $t$ and $c^{h}$ is the equilibrium consumption bundle of household $h$, thus allowing for unbounded dynamics for both the demography and per capita endowments (the two previous statements are recovered when per capita consumption is bounded away from zero and infinity).

Second, the economy is an overlapping generations economy with production, for which an equilibrium existence theorem is provided along the lines of \textcite{Dognini_2026a}. Firms are owned by the generation whose lifetime coincides with their activity, and the general statement (\ref{eqGeneralCass}) is exactly the same for consumption-loan and production economies.

Reaching a more general statement for the Cass criterion is a relevant milestone for properly deploying overlapping generations models, which are used for a wide range of economic applications stemming from monetary theory (e.g., \textcite{Lucas_1972,Wallace_1980,Santos_1990, Dognini_2026a, Dognini_2026c}) to the literature on bubbles (e.g., \textcite{Tirole_1985,SantosWoodford_1997,HiratoToda_2025}) and social security design (e.g., \textcite{Samuelson_1975,Dognini_2026b}). 

The last contribution of the paper is to resolve an apparent inconsistency between the relabeling algorithm of \textcite{BalaskoCassShell_1980} and the usual proofs of the Cass criterion. The relabeling algorithm poses any overlapping generations economy with finitely-lived households as a two-period one, at the cost of enlarging the domain of the utility functions with commodities that the household does not desire. As explained in \hyperref[subsecBCS]{Section 4.1}, this enlargement sets the Gaussian curvature of every indifference hypersurface to zero, in the same way as a cylinder built over a circle has zero Gaussian curvature.

On the other hand, the classical characterizations of efficiency rely on uniform bounds on the Gaussian curvature. For instance, Proposition 5.6 in \textcite[pp. 296--297]{BalaskoShell_1980} requires the Gaussian curvature of each household's indifference surface to be uniformly bounded from above \textit{at every point} of the surface (their Property C, used for the necessity of the Cass criterion), and to be uniformly bounded away from zero on a whole region of the surface (their Property C$^{\prime}$, used for the sufficiency).

This paper shows that the Cass criterion does not need this global information: the assumptions that support the general statement (\ref{eqGeneralCass}) are presented ``on demand'' in a step-by-step proof that highlights either their \textit{pointwise} or \textit{local} nature. Pointwise assumptions depend only on the equilibrium price sequence and on the primitives of the economy, whereas local assumptions depend on the behavior of the indifference hypersurfaces (and of the transformation frontiers) in a neighborhood of each household's equilibrium consumption (and of each firm's equilibrium production plan), with the size of the neighborhood measured relative to the size of the bundle.

After this \hyperref[sec1]{Introduction}, the outline of the paper is as follows. \hyperref[sec2]{Section 2} defines and characterizes the local behavior of the depth projections of the indifference and the offer hypersurfaces. \hyperref[sec3]{Section 3} extends the definition to the transformation frontier of a production set, with the profit function taking over the role of the expenditure function. \hyperref[sec4]{Section 4} discusses the relabeling algorithm, defines the production overlapping generations economy, proves the existence of equilibrium, and derives the general statements for the necessity and sufficiency of the Cass criterion. Finally, a brief
discussion of the results is presented in the \hyperref[sec5]{Conclusion} 
and all omitted proofs are stated in the \hyperref[appx]{Appendix}.

%--- Section ---%

\section{The depth projections of the canonical manifolds}\label{sec2}
The consumption set\footnote{I follow \textcite{Mas-ColellWhinstonGreen_1995} for classical demand theory definitions and results.} is $X=\mathbb{R}^{L}_{+}$, $L\geq2$, and preferences are defined through a smooth\footnote{By \textit{smooth}, I mean continuous in $\mathbb{R}^{L}_{+}$ and with continuous derivatives of any order in $\mathbb{R}^{L}_{++}$.}, strictly increasing, and strictly quasiconcave\footnote{Utility functions for which the indifference hypersurface of the zero utility level is $\partial \mathbb{R}^{L}_{+}$ (e.g., Cobb--Douglas) require a caveat, since for two bundles on this boundary, any convex combination also has zero utility.} utility function $u:\mathbb{R}^{L}_{+}\rightarrow\mathbb{R}_{+}$, with $u(0)=0$, $\nabla u(c)\in\mathbb{R}^{L}_{++}$, $c\in\mathbb{R}^{L}_{++}$, and $u(\cdot)$ unbounded above. The budget set correspondence is $B(p,w)=\{c\in\mathbb{R}^{L}_{+}\mid p \cdot c\leq w\}$, $(p,w)\in\mathbb{R}^{L+1}_{++}$. The Walrasian demand function\footnote{It is a function, and not a correspondence, due to the strict quasiconcavity assumption on $u(\cdot)$ \parencite[p. 52, Proposition 3.D.2]{Mas-ColellWhinstonGreen_1995}.} $x:\mathbb{R}^{L+1}_{++}\rightarrow \mathbb{R}^{L}_{+}$ is defined through 
\begin{eqnarray}\label{UMP}
    x(p,w)=\argmax_{c\in\mathbb{R}^{L}_{+}} u(c) \textrm{ s.t. } c\in B(p,w),
\end{eqnarray}
and I assume $x(\cdot)$ is smooth\footnote{As a remark, it is well known (e.g., \textcite[pp. 94--95]{Mas-ColellWhinstonGreen_1995}) that the differentiability of $x(\cdot)$ is ensured by a nonzero determinant of the bordered Hessian of $u(\cdot)$. 
}and $x(p,w)\in\mathbb{R}^{L}_{++}$, $(p,w)\in\mathbb{R}^{L+1}_{++}$ (i.e., (\ref{UMP}) has no corner solutions\footnote{I prefer to state this assumption rather than impose conditions that ensure it. For instance, if the utility functions are strongly monotonic, satisfy a strong bordered Hessian condition, and the indifference hypersurfaces are closed in $\mathbb{R}^{L}$, then this assumption is satisfied (e.g., \textcite[p. 2]{YuHosoya_2022}).}).

\begin{definition}\label{defNormalizedWalrasian}
The \textit{normalized Walrasian demand} $x_{n}:\mathbb{R}^{L}_{++}\rightarrow \mathbb{R}^{L}_{++}$ is given by $x_{n}(p)=x(p,1)$.
\end{definition}
The normalized Walrasian demand represents the household demand when prices are stated as percentages of wealth. It gathers all information about household consumption behavior since $x_{n}(p/w)=x(p,w)$, $(p,w)\in\mathbb{R}^{L+1}_{++}$, and can be seen as representing the fact that, from a household's perspective, consumption decisions only depend on its purchasing power.

Let $\lambda(p,w)>0$ be the Lagrange multiplier of (\ref{UMP}) and the indirect utility function $v:\mathbb{R}^{L+1}_{++}\rightarrow \mathbb{R}_{++}$ be given by $v(p,w)=u(x(p,w))$. Following Definition \ref{defNormalizedWalrasian}, let $\lambda_{n}(p)=\lambda(p,1)$ and $v_{n}(p)=u(x_{n}(p))$, $p\in\mathbb{R}^{L}_{++}$. The Hicksian demand function $h:\mathbb{R}^{L+1}_{++}\rightarrow\mathbb{R}^{L}_{++}$ is
\begin{eqnarray*}
    h(p,u)=\argmin_{c\in\mathbb{R}^{L}_{+}} p\cdot c \textrm{\ s.t.\ } u(c)\geq u,
\end{eqnarray*}
and the expenditure function $e:\mathbb{R}^{L+1}_{++}\rightarrow\mathbb{R}_{++}$ is $e(p,u)=p \cdot h(p,u)$. To ease notation\footnote{Throughout the paper, vectors are row vectors, $y^{T}$ denotes the transpose of $y$, and $x\cdot y=xy^{T}$ is the usual inner product; for a smooth function $g(\cdot)$, $\nabla g(\cdot)$, $\mathbf{J}g(\cdot)$ and $\mathbf{H}g(\cdot)$ denote its gradient (a row vector), Jacobian, and Hessian matrices.}, I write $\nabla e(p,u)=(\nabla_{1}e(p,u),\nabla_{2}e(p,u))\in\mathbb{R}^{L}_{++}\times\mathbb{R}_{++}$, $\mathbf{H}_{11}e(p,u)=\mathbf{J}_{1}\nabla_{1}e(p,u)\in\mathbb{R}^{L\times L}$ and $\mathbf{H}_{21}e(p,u)=\partial \nabla_{1}e(p,u)/\partial u\in\mathbb{R}^{L}$. The following result gathers and extends Lemma 1 and Proposition 2 in \textcite{Dognini_2025b} (see it for a detailed proof and discussion of these definitions and their relations to the \textit{consumption}, \textit{normalized} and \textit{flat domains}).

\begin{lemma}\label{lemmaBasicIdentities}
Let $x_{n}(\cdot)$, $v_{n}(\cdot)$, $\lambda_{n}(\cdot)$, $h(\cdot)$, and $e(\cdot)$ be as defined above. Then, $h(\cdot)$ and $e(\cdot)$ are smooth, $x_{n}(\cdot)$ is a diffeomorphism onto $\mathbb{R}^{L}_{++}$, and, for $(p,u)\in\mathbb{R}^{L+1}_{++}$ and $c\in\mathbb{R}^{L}_{++}$:
\begin{itemize}
    \item[i.] $\nabla u(x_{n}(p))=\lambda_{n}(p)p$ and $x^{-1}_{n}(c)=\nabla u(c)/(c\cdot\nabla u(c))$;
    \item[ii.] $p \cdot x_{n}(p)=1$, $p\mathbf{J}x_{n}(p)=-x_{n}(p)$, and $\lambda_{n}(p)=x_{n}(p)\cdot \nabla u(x_{n}(p))$;
    \item[iii.] $\nabla v_{n}(p)=-\lambda_{n}(p)x_{n}(p)$ and $p\cdot \nabla v_{n}(p)=-\lambda_{n}(p)$;
    \item[iv.] $\nabla_{1}e(p,u)=h(p,u)$ and $p\mathbf{H}_{11}e(p,u)=0$;
    \item[v.] $x_{n}(p)=h(p,v_{n}(p))$, $h(p,u)=x_{n}(p/e(p,u))$, and $e(p,v_{n}(p))=1$;
    \item[vi.] $\lambda_{n}(p)\partial e(p,v_{n}(p))/\partial u=1$.
\end{itemize}
\end{lemma}

I proceed with the following definition from \textcite{Dognini_2025b}.

\begin{definition}\label{defCanonical}
     The \textit{canonical manifolds} through  $c\in\mathbb{R}^{L}_{++}$ are:
    \begin{itemize}
        \item[i.] The \textit{indifference hypersurface} given by $\mathcal{I}(c)=\{z\in\mathbb{R}^{L}_{++}\mid u(z)=u(c)\}$;
        \item[ii.] The \textit{offer hypersurface} given by $\mathcal{O}(c)=\{ z\in\mathbb{R}^{L}_{++}\mid \nabla u(z)\cdot (z-c)=0\}$;
        \item[iii.] The \textit{trade hyperplane} given by $\mathcal{H}(c)=\{z\in\mathbb{R}^{L}_{++}\mid \nabla u(c)\cdot(z-c)=0\}$.
    \end{itemize}
\end{definition}

The three canonical manifolds are connected $(L-1)$-dimensional manifolds with a single intersection at $c\in\mathbb{R}^{L}_{++}$. Furthermore, the indifference and the offer hypersurfaces can be written as $\mathcal{I}(c)=\{h((q,1),u(c))\mid q\in\mathbb{R}^{L-1}_{++}\}$ and $\mathcal{O}(c)=\{x_{n}((q,1)/((q,1)\cdot c))\mid q\in\mathbb{R}^{L-1}_{++}\}$ \parencite[Proposition 5]{Dognini_2025b}. Figure \ref{FigCanonical} depicts the canonical manifolds through $c=(1,1)$ for $u(c)=c_{1}c_{2}$.

\begin{figure}[H]
\centering
\includegraphics[width=0.5\textwidth]{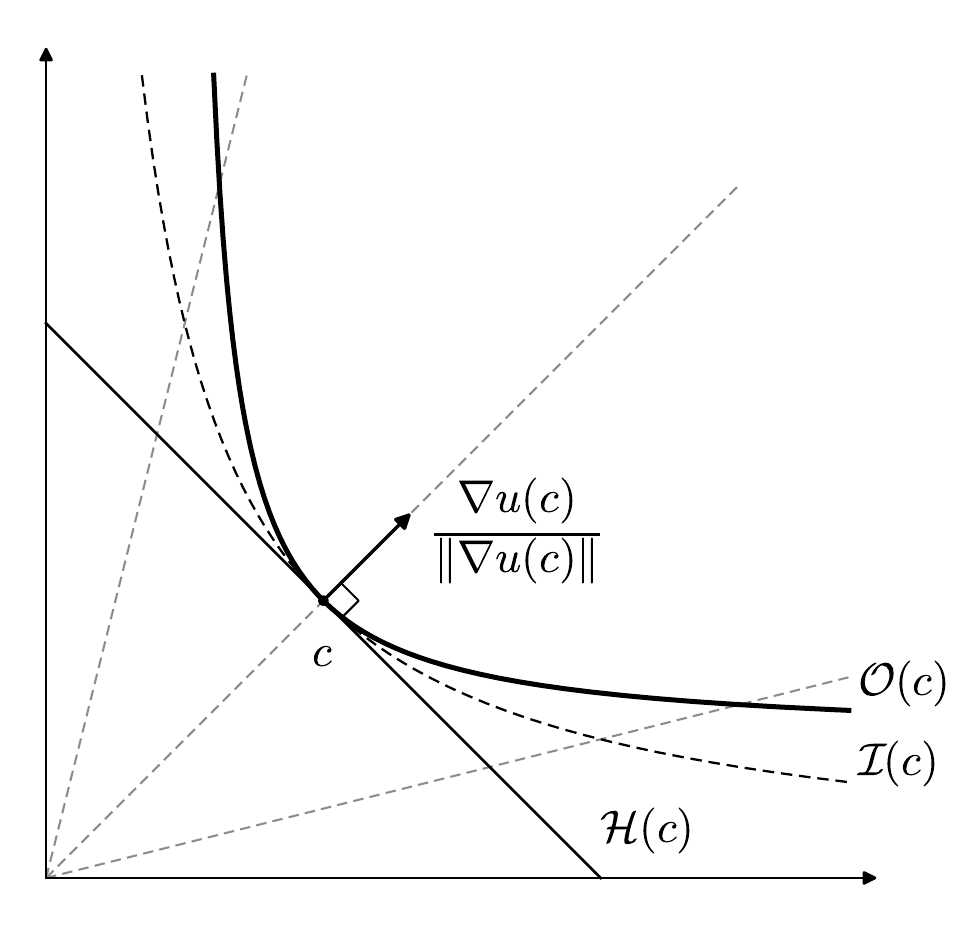}
\caption{Graphical depiction of the canonical manifolds through $c=(1,1)$ for $u(c)=c_{1}c_{2}$.}
\label{FigCanonical}
\end{figure}

The gray dashed rays in Figure \ref{FigCanonical} are wealth expansion paths \parencite[p. 25]{Mas-ColellWhinstonGreen_1995} and lead us to note that these paths, for any given relative prices, intersect each of the three canonical manifolds at a unique point. By associating these points, one can see the indifference hypersurface as a bent version of the trade hyperplane around the point $c=(1,1)$. 

Further bending the indifference hypersurface thus leads us to the offer hypersurface, and the degree to which the trade hyperplane is curved can be measured by how far the projection along $\nabla u(c)/\Vert \nabla u(c)\Vert$  ``runs deep'' into the interior of the convex sets. This motivates the following definition.

\begin{definition}\label{defDepthProjection}
    For $c\in\mathbb{R}^{L}_{++}$, the \textit{depth projections of the indifference hypersurface $\mathcal{I}(c)$ and the offer hypersurface $\mathcal{O}(c)$ at $c$}, $\pi_{i}:\mathbb{R}^{L-1}_{++}\rightarrow\mathbb{R}_{+}$ and $\pi_{o}:\mathbb{R}^{L-1}_{++}\rightarrow\mathbb{R}_{+}$, respectively, are given by
    \begin{eqnarray*}
        \pi_{i}(q)&=&\frac{\nabla u(c)}{\Vert \nabla u(c)\Vert}\cdot (h((q,1),u(c))-c)\\
        \pi_{o}(q)&=&\frac{\nabla u(c)}{\Vert \nabla u(c)\Vert}\cdot\biggr( x_{n}\biggr(\frac{(q,1)}{(q,1)\cdot c}\biggr)-c\biggr),
    \end{eqnarray*}
    for $q\in\mathbb{R}^{L-1}_{++}$.
\end{definition}

The depth projections in Definition \ref{defDepthProjection} have an intuitive economic interpretation through the lens of the expenditure minimization problem (EMP) and the utility maximization problem (UMP). For the indifference hypersurface, think of $\nabla {u}(c)/\Vert \nabla u(c)\Vert\in\mathbb{R}^{L}_{++}$ as a price vector and $u(c)>0$ as a given utility threshold. Then, the EMP tells us that $c\in\mathbb{R}^{L}_{++}$ is the point that attains this utility level at the lowest possible cost under prices $\nabla u(c)/\Vert \nabla u(c)\Vert$. Therefore, if we look at any point on this indifference hypersurface $h((q,1),u(c))$, $(q,1)\nparallel\nabla u(c)$, it will imply an additional cost of $\pi_{i}(q)=(\nabla u(c)/\Vert \nabla u(c)\Vert)\cdot(h((q,1),u(c))-c)>0$.

For the offer hypersurface, think of $\nabla {u}(c)/\Vert \nabla {u}(c)\Vert\in\mathbb{R}^{L}_{++}$ as an initial price vector and $c\in\mathbb{R}^{L}_{++}$ as an endowment. Suppose that the relative prices change to $(q,1)\nparallel \nabla u(c)$, so that the household will now trade to obtain $x_{n}((q,1)/((q,1)\cdot c))\neq c$ according to the UMP. This final basket was unaffordable at the initial prices, since the household would need an additional wealth of $\pi_{o}(q)=(\nabla u(c)/\Vert \nabla u(c)\Vert)\cdot(x_{n}((q,1)/((q,1)\cdot c))-c)>0$ to buy it\footnote{I opted to write Definition \ref{defDepthProjection} with the inner product norm since this leads to the natural geometric interpretation of an orthogonal projection. Nonetheless, if one looks for a strictly economic interpretation, adopting $\Vert \nabla u(c)\Vert_{1}=\sum^{L}_{i=1}\vert \partial u(c)/\partial c_{i}\vert$ in the denominator carries the meaning that the corresponding wealth variations are measured in real terms, according to the basket $(1,\ldots,1)\in\mathbb{R}^{L}_{++}$, since $ \Vert \nabla u(c)\Vert_{1}=\nabla u(c)\cdot (1,\ldots,1)$.}.

Geometrically, the depth projections $\pi_{i}(\cdot)$ and $\pi_{o}(\cdot)$ are simply the lengths of the orthogonal projections, in the direction $\nabla u(c)/\Vert \nabla u(c)\Vert$, of the segments joining $c$ to the corresponding points from the indifference and the offer hypersurfaces. This can be seen in Figure \ref{FigDepth}, where the steepest gray dashed ray through the origin is the wealth expansion path associated with the price vector $(q,1)$: it intersects the indifference hypersurface at $h((q,1),u(c))$ and the offer hypersurface at $x_{n}((q,1)/((q,1)\cdot c))$.

\begin{figure}[H]
\centering
\includegraphics[width=0.5\textwidth]{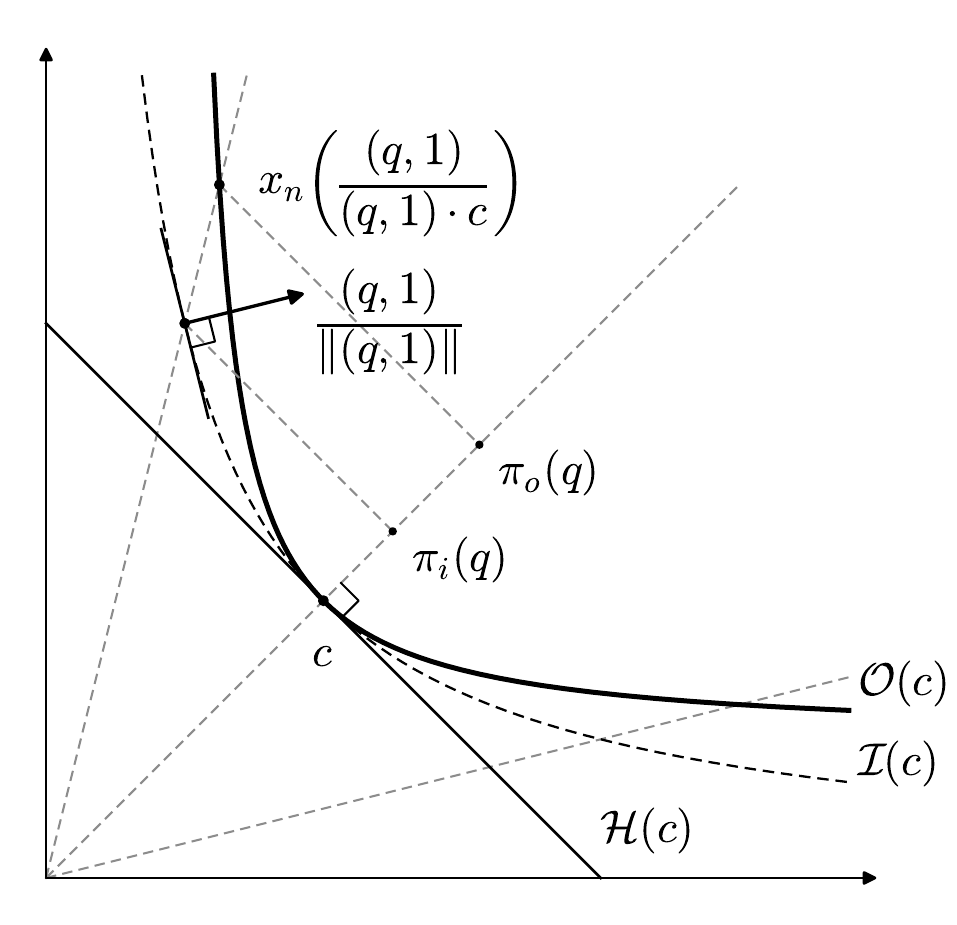}
\caption{The depth projections $\pi_{i}(q)$ and $\pi_{o}(q)$ for $c=(1,1)$, $u(c)=c_{1}c_{2}$.}
\label{FigDepth}
\end{figure}

Three remarks should still be made about Definition \ref{defDepthProjection}. First, the depth projections are related to the curvature of the hypersurfaces but are not numbers, as the Gaussian curvature is. As a reminder, the Gaussian curvature is the determinant of the differential of the Gauss map, which is the mapping from the hypersurface to the unit sphere that associates each point with its normal vector. For example, if one considers the indifference hypersurface $\mathcal{I}(c)$, $c\in\mathbb{R}^{L}_{++}$, then the map $p\rightarrow h(p,u(c))$, $p\in \mathbb{R}^{L}_{++}$, $\Vert p\Vert=1$, from the unit sphere to the indifference hypersurface is the inverse of the Gauss map. The following example illustrates this first remark.

\begin{example}\label{exCES}
Let $u_{\rho}:\mathbb{R}^{2}_{++}\rightarrow\mathbb{R}$ be given by $u_{\rho}(c)=(c_{1}^{\rho}/2+c_{2}^{\rho}/2)^{\frac{1}{\rho}}$, $\rho\in(0,1)$. Then, the Hicksian demand function is
\begin{eqnarray}\label{eqHicksianCES}
    h_{\rho}(p,u)=u\Biggr(p_{1}^{-\frac{1}{1-\rho}}\biggr(\frac{p_{1}^{-\frac{\rho}{1-\rho}}+p_{2}^{-\frac{\rho}{1-\rho}}}{2}\biggr)^{-\frac{1}{\rho}},\ p_{2}^{-\frac{1}{1-\rho}}\biggr(\frac{p_{1}^{-\frac{\rho}{1-\rho}}+p_{2}^{-\frac{\rho}{1-\rho}}}{2}\biggr)^{-\frac{1}{\rho}}\Biggr),
\end{eqnarray}
and the depth projection of the indifference hypersurface that passes through $c=(1,1)$ is
\begin{eqnarray*}
    \pi_{i}(q,\rho)=\frac{1}{\sqrt{2}}\biggr[\bigl(q^{-\frac{1}{1-\rho}}+1\bigr)\biggr(\frac{q^{-\frac{\rho}{1-\rho}}+1}{2}\biggr)^{-\frac{1}{\rho}}-2\biggr]
\end{eqnarray*}
for $q\in\mathbb{R}_{++}$, where the $\rho\in(0,1)$ enters as an argument of $\pi_{i}(\cdot)$ just to ease notation. Clearly, the $\rho\rightarrow1$ case approaches a flat indifference hypersurface, and the $\rho\rightarrow0$ case approaches a Cobb--Douglas indifference hypersurface. The Gaussian curvature of the indifference hypersurface at $c=(1,1)$ can be obtained from (\ref{eqHicksianCES}) in the following way. 

First, notice that the tangent space to the indifference hypersurface $\mathcal{I}_{\rho}(c)$ at $c=(1,1)$ is given by $\mathbf{span}\{\nabla u_{\rho}(c)\}^{\perp}=\{(\lambda,-\lambda)\in\mathbb{R}^{2}\mid \lambda\in\mathbb{R}\}$, for all $\rho\in(0,1)$. Then, differentiate (\ref{eqHicksianCES}) to obtain
\begin{eqnarray*}
    \mathbf{J}_{p}h_{\rho}(p,u)=\frac{2^{\frac{1}{\rho}}u}{(1-\rho)}\biggr(p_{1}^{-\frac{\rho}{1-\rho}}+p_{2}^{-\frac{\rho}{1-\rho}}\biggr)^{-\frac{1}{\rho}-1}\begin{bmatrix}
        -p_{1}^{-\frac{2-\rho}{1-\rho}}p_{2}^{-\frac{\rho}{1-\rho}} & (p_{1}p_{2})^{-\frac{1}{1-\rho}}\\
        (p_{1}p_{2})^{-\frac{1}{1-\rho}} & -p_{2}^{-\frac{2-\rho}{1-\rho}}p_{1}^{-\frac{\rho}{1-\rho}}
    \end{bmatrix}.
\end{eqnarray*}
Next, parametrize the strictly positive orthant of the unit circle by $p(\theta)=(\cos\theta,\sin\theta)$, $\theta\in(0,\pi/2)$. Notice that $u_{\rho}(1,1)=1$ and $h_{\rho}(p(\pi/4),1)=(1,1)$, $\rho\in(0,1)$. Then,
\begin{eqnarray*}
    \mathbf{J}_{p}h_{\rho}(p(\pi/4),1)=\frac{1}{\sqrt{2}(1-\rho)}\begin{bmatrix}
        -1 & 1\\
        1 & -1
    \end{bmatrix}.
\end{eqnarray*}
For $\gamma(\theta)=h_{\rho}(p(\theta),1)$, $\theta\in(0,\pi/2)$, we have
\begin{eqnarray*}
    \gamma^{\prime}\biggr(\frac{\pi}{4}\biggr)=\mathbf{J}_{p}h_{\rho}(p(\pi/4),1)p^{\prime}(\pi/4)^{T}=\biggr(\frac{1}{1-\rho},-\frac{1}{1-\rho}\biggr),
\end{eqnarray*}
since $p^{\prime}(\pi/4)=(-1/\sqrt{2},1/\sqrt{2})$. Therefore, the differential of the inverse of the Gauss map associates the vector $p^{\prime}(\pi/4)$ in the tangent space of the unit sphere with the vector $(1/(1-\rho),-1/(1-\rho))$ in the tangent space of the indifference hypersurface $\mathcal{I}_{\rho}(c)$ at $c=(1,1)$ (i.e., $\mathbf{span}\{\nabla u_{\rho}(c)\}^{\perp}$). Since we are dealing with differentials that are linear transformations between one-dimensional vector spaces, the determinant of the differential of the Gauss map (i.e., the Gaussian curvature) is straightforwardly given by
\begin{eqnarray*}
    K(c,\rho)= \frac{\Vert p^{\prime}(\pi/4)\Vert}{\Vert \gamma^{\prime}(\pi/4)\Vert}=\frac{1-\rho}{\sqrt{2}}.
\end{eqnarray*}
We conclude that the Gaussian curvature is strictly decreasing with $\rho\in(0,1)$ and $\lim_{\rho\rightarrow1}K(c,\rho)=0$, which is in accordance with the fact that the indifference curves are becoming flatter as $\rho\rightarrow 1$.

And how are the depth projections $\pi_{i}(\cdot,\rho)$ related to the curvature of these indifference hypersurfaces? After some calculations (see Corollary \ref{corDepthProj} for a related general result), we have $\pi_{i}(1,\rho)=\pi^{\prime}_{i}(1,\rho)=0$ and
\begin{eqnarray*}
    \pi^{\prime\prime}_{i}(1,\rho)=\frac{1}{2\sqrt{2}(1-\rho)}
\end{eqnarray*}
for $\rho\in(0,1)$. Therefore, for all $0<\rho_{1}<\rho_{2}<1$, there is an open set $O(\rho_{1},\rho_{2}) \ni 1$ such that $\pi_{i}(q,\rho_{1})\leq\pi_{i}(q,\rho_{2})$, $q\in O(\rho_{1},\rho_{2})$. This may seem counterintuitive, since the most curved indifference hypersurface (i.e., $\mathcal{I}_{\rho_{1}}(c)$) is the one that, locally, has the lowest depth projection. Shouldn't it be the other way around? 

There is, I believe, an intuitive way to understand this. Think of the parameter $q\approx1$ as measuring how twisted your hand is when you slide it through the hypersurface around $c=(1,1)$ (in this case, the vector $(q,1)$ is normal to your palm). Then, if the hypersurface is more pointy at $c=(1,1)$ (i.e., more curved), your hand will make a relatively big twist without making a big translation (i.e., your depth projection will be small relative to the change in the parameter $q\approx 1$). Therefore, for an \textit{equal twist}, your hand translates less the pointier the hypersurface is.

This discussion can be summarized as follows. If one is given two Gaussian curvatures $K(c,\rho_{1})$ and $K(c,\rho_{2})$, then (as expected) the larger the Gaussian curvature, the more curved the indifference hypersurface is around $c\in\mathbb{R}^{2}$. If one is given two depth projections $\pi_{i}(\cdot,\rho_{1})$ and $\pi_{i}(\cdot,\rho_{2})$, then the one that is locally smaller identifies the more curved indifference hypersurface around $c\in\mathbb{R}^{2}$. Loosely speaking, these depth projections of the indifference hypersurfaces record radii of curvature rather than curvatures: the more curved the hypersurface (i.e., smaller radius) is around the reference bundle, the smaller its depth projection is.
\end{example}

The other two remarks regarding Definition \ref{defDepthProjection} are as follows. Notice that the definition uses \textit{specific} global parametrizations defined in $\mathbb{R}^{L-1}_{++}$ to map the indifference and offer hypersurfaces. This is central to the results in this paper as they are stated. For instance, the global parametrization used to map the indifference hypersurface $\mathcal{I}(c)$ is directly related to the normal vector at each point of the hypersurface, since $\nabla u(h((q,1),u(c)))\parallel(q,1)$, $q\in\mathbb{R}^{L-1}_{++}$. In this sense, this global parametrization mimics the logic behind the inverse Gauss map; instead of mapping the unit sphere to the hypersurface, it maps a hyperplane.

On the other hand, the global parametrization used to map the offer hypersurface is \textit{not related to the normal vector} at each point of the hypersurface. In fact, since $\nabla u(x_{n}((q,1)/((q,1)\cdot c)))\parallel (q,1)$, $q\in\mathbb{R}^{L-1}_{++}$, I am actually parameterizing a point in the offer hypersurface by the normal vector to the indifference hypersurface that passes through it (i.e., by the marginal rates of substitution; observe the steeper wealth expansion path in Figure \ref{FigDepth} to see that $(q,1)$ is not the normal vector to the offer hypersurface at $x_{n}((q,1)/((q,1)\cdot c))$).

The third and last remark is as follows. Depth projections can also be defined for a well-behaved transformation frontier, following the same logic used for the indifference hypersurface in Definition \ref{defDepthProjection}. This is done in \hyperref[sec3]{Section 3} (see Definition \ref{defDepthProjectionProduction} and Theorem \ref{theoProduction}), where the role played by the expenditure function is taken over by the profit function.

For a given $p\in\mathbb{R}^{L}_{++}$, it is convenient to define the auxiliary functions $\eta:\mathbb{R}^{L}_{++}\rightarrow\mathbb{R}_{+}$ and $\zeta:\mathbb{R}^{L}_{++}\rightarrow\mathbb{R}_{+}$ as
\begin{eqnarray*}
    \eta(q)&=&p\cdot (h(q,v_{n}(p))-x_{n}(p))\\
    \zeta(q)&=&p\cdot\biggr( x_{n}\biggr(\frac{q}{q\cdot x_{n}(p)}\biggr)-x_{n}(p)\biggr),
\end{eqnarray*}
where the dependence of $\eta(\cdot)$ and $\zeta(\cdot)$ on $p$ is omitted to ease notation. Notice that, for $c=x_{n}(p)\in\mathbb{R}^{L}_{++}$, Lemma \ref{lemmaBasicIdentities} implies $\nabla u(c)/\Vert\nabla u(c)\Vert=p/\Vert p\Vert$ and, therefore, $\pi_{i}(q)=\eta(q,1)/\Vert p\Vert$ and $\pi_{o}(q)=\zeta(q,1)/\Vert p\Vert$, $q\in\mathbb{R}^{L-1}_{++}$. The next theorem characterizes the local behavior of these auxiliary functions.

\begin{theorem}\label{theoAwesome}
    Let $p\in\mathbb{R}^{L}_{++}$, $\eta(\cdot)$ and $\zeta(\cdot)$ be as defined above. Then, $\eta(p)=\zeta(p)=0$, $\nabla \eta(p)=\nabla \zeta(p)=0$, and $\mathbf{H}\zeta(p)=2\mathbf{H}\eta(p)=-2\mathbf{H}_{11}e(p,v_{n}(p))$.
\end{theorem}

The following example illustrates Theorem \ref{theoAwesome}.
\begin{example}\label{ex1}
Let $u:\mathbb{R}^{2}_{++}\rightarrow\mathbb{R}$ be given by $u(c)=c_{1}c_{2}$. The normalized Walrasian demand is $x_{n}(p)=(1/2p_{1},1/2p_{2})$, $\lambda_{n}(p)=1/2p_{1}p_{2}$, and the normalized indirect utility function is $v_{n}(p)=1/4p_{1}p_{2}$, $p\in\mathbb{R}^{2}_{++}$. The Hicksian demand function is $h(p,u)=(\sqrt{up_{2}/p_{1}},\sqrt{up_{1}/p_{2}})$ and the expenditure function is $e(p,u)=2\sqrt{p_{1}p_{2}u}$, $(p,u)\in\mathbb{R}^{3}_{++}$. Notice that
\begin{eqnarray*}
    \mathbf{H}_{11}e(p,v_{n}(p))=\frac{1}{4p_{1}p_{2}}\begin{bmatrix}
        -\dfrac{p_{2}}{p_{1}} & 1\\
        1 & -\dfrac{p_{1}}{p_{2}}
    \end{bmatrix},
\end{eqnarray*}
which is negative semi-definite and has eigenvalues $\mu_{1}=0$ and  $\mu_{2}=-\Vert p\Vert^{2}/4p_{1}^{2}p_{2}^{2}$, $p\in\mathbb{R}^{2}_{++}$. Next, we have
\begin{eqnarray*}
    \eta(q)=\frac{1}{2}\sqrt{\frac{p_{1}q_{2}}{p_{2}q_{1}}}+\frac{1}{2}\sqrt{\frac{p_{2}q_{1}}{p_{1}q_{2}}}-1,
\end{eqnarray*}
so that 
\begin{eqnarray*}
    \nabla \eta(q)=\biggr(-\frac{1}{4}\sqrt{\frac{p_{1}q_{2}}{p_{2}q_{1}^{3}}}+\frac{1}{4}\sqrt{\frac{p_{2}}{p_{1}q_{1}q_{2}}},\frac{1}{4}\sqrt{\frac{p_{1}}{p_{2}q_{1}q_{2}}}-\frac{1}{4}\sqrt{\frac{p_{2}q_{1}}{p_{1}q_{2}^{3}}}\biggr)
\end{eqnarray*}
and
\begin{eqnarray*}
    \mathbf{H}\eta(q)=\begin{bmatrix}
        \dfrac{3}{8}\sqrt{\dfrac{p_{1}q_{2}}{p_{2}q_{1}^{5}}}-\dfrac{1}{8}\sqrt{\dfrac{p_{2}}{p_{1}q_{1}^{3}q_{2}}} & 
        -\dfrac{1}{8}\sqrt{\dfrac{p_{1}}{p_{2}q_{1}^{3}q_{2}}}-\dfrac{1}{8}\sqrt{\dfrac{p_{2}}{p_{1}q_{1}q_{2}^{3}}}\\
        -\dfrac{1}{8}\sqrt{\dfrac{p_{1}}{p_{2}q_{1}^{3}q_{2}}}-\dfrac{1}{8}\sqrt{\dfrac{p_{2}}{p_{1}q_{1}q_{2}^{3}}} &
        -\dfrac{1}{8}\sqrt{\dfrac{p_{1}}{p_{2}q_{1}q_{2}^{3}}}+\dfrac{3}{8}\sqrt{\dfrac{p_{2}q_{1}}{p_{1}q_{2}^{5}}}
    \end{bmatrix},
\end{eqnarray*}
for $q\in\mathbb{R}^{2}_{++}$. Therefore, $\eta(p)=0$, $\nabla \eta(p)=0$ and
\begin{eqnarray*}
    \mathbf{H}\eta(p)=\begin{bmatrix}
        \dfrac{1}{4p_{1}^{2}} & 
        -\dfrac{1}{4p_{1}p_{2}}\\
        -\dfrac{1}{4p_{1}p_{2}} &
        \dfrac{1}{4p_{2}^{2}}
    \end{bmatrix}=-\frac{1}{4p_{1}p_{2}}\begin{bmatrix}
        -\dfrac{p_{2}}{p_{1}} & 
        1\\
        1 &
        -\dfrac{p_{1}}{p_{2}}
    \end{bmatrix}=-\mathbf{H}_{11}e(p,v_{n}(p)).
\end{eqnarray*}
Next, we have
\begin{eqnarray*}
    \zeta(q)=\frac{p_{1}q_{2}}{4p_{2}q_{1}}+\frac{p_{2}q_{1}}{4p_{1}q_{2}}-\frac{1}{2},
\end{eqnarray*}
so that
\begin{eqnarray*}
    \nabla\zeta(q)=\biggr(-\frac{p_{1}q_{2}}{4p_{2}q_{1}^{2}}+\frac{p_{2}}{4p_{1}q_{2}},\frac{p_{1}}{4p_{2}q_{1}}-\frac{p_{2}q_{1}}{4p_{1}q_{2}^{2}}\biggr),
\end{eqnarray*}
and
\begin{eqnarray*}
    \mathbf{H}\zeta(q)=\begin{bmatrix}
        \dfrac{p_{1}q_{2}}{2p_{2}q_{1}^{3}} & 
        -\dfrac{p_{1}}{4p_{2}q_{1}^{2}}-\dfrac{p_{2}}{4p_{1}q_{2}^{2}}\\
        -\dfrac{p_{1}}{4p_{2}q_{1}^{2}}-\dfrac{p_{2}}{4p_{1}q_{2}^{2}} &
        \dfrac{p_{2}q_{1}}{2p_{1}q_{2}^{3}}
    \end{bmatrix},
\end{eqnarray*}
for $q\in\mathbb{R}^{2}_{++}$. Therefore, $\zeta(p)=0$, $\nabla\zeta(p)=0$ and
\begin{eqnarray*}
     \mathbf{H}\zeta(p)=\begin{bmatrix}
        \dfrac{1}{2p_{1}^{2}} &
        -\dfrac{1}{2p_{1}p_{2}}\\
        -\dfrac{1}{2p_{1}p_{2}} &
        \dfrac{1}{2p_{2}^{2}}
    \end{bmatrix}=-2\mathbf{H}_{11}e(p,v_{n}(p))=2\mathbf{H}\eta(p).
\end{eqnarray*}
Notice also that $\partial e(p,v_{n}(p))/\partial u=p_{1}p_{2}/2$, so that $\lambda_{n}(p)\partial e(p,v_{n}(p))/\partial u=1$, in accordance with the duality identity used in the proof of Theorem \ref{theoAwesome}.
\end{example}

Before stating the next results, it is convenient to define $f:\mathbb{R}^{L}_{++}\rightarrow\mathbb{R}^{L-1}_{++}$ as\footnote{It is important to notice that this auxiliary function \textit{is not} the flattening diffeomorphism from \textcite{Dognini_2025b} (although it coincides with the first $L-1$ coordinates of it).}
\begin{eqnarray*}
    f(x)=\frac{\nabla_{L-1}u(x)}{\partial_{L}u(x)}
\end{eqnarray*}
for $x\in\mathbb{R}^{L}_{++}$, with $\nabla_{L-1}u(x)\in\mathbb{R}^{L-1}$ the first $L-1$ coordinates of $\nabla u(x)$ (i.e., $f(x)$ collects the marginal rates of substitution between each of the first $L-1$ goods and the $L$-th one). Also, let $e_{L}=(0,\ldots,0,1)\in\mathbb{R}^{L}$ be the $L$-th canonical vector and $M=\begin{bmatrix}
    I & 0
\end{bmatrix}\in\mathbb{R}^{(L-1)\times L}$, with $I$ the identity matrix of the appropriate dimension (here, $(L-1)\times(L-1)$), so that $f(x)=\nabla u(x)M^{T}/(\nabla u(x)\cdot e_{L})$. 

The following lemma collects the consequences of the smoothness of $x(\cdot)$ that will be needed for the remaining results of this section. Items i and ii are the content of \textcite{Debreu_1972}: a smooth demand function implies that the Slutsky matrix has the smallest possible kernel and that the indifference hypersurfaces have nonzero Gaussian curvature at every point.\footnote{Since $u(\cdot)$ is strictly quasiconcave, $v\mathbf{H}u(c)v^{T}\leq0$ for all $v\in\mathbf{span}\{\nabla u(c)\}^{\perp}$ always holds, so that item ii only rules out the directions in which the indifference hypersurface is flat up to second order.} Item iii is the uniform version of item ii that is used in the proof of Theorem \ref{theoFundamentalIneq}.

\begin{lemma}\label{lemmaLowerBoundHessianUtility}
    Let $c\in\mathbb{R}^{L}_{++}$, $p=x^{-1}_{n}(c)$ and $f(\cdot)$ be as defined above. Then:
    \begin{itemize}
        \item[i.] $\mathbf{ker}\,\mathbf{H}_{11}e(p,v_{n}(p))=\mathbf{span}\{p\}$;
        \item[ii.] $v\mathbf{H}u(c)v^{T}<0$ for all $v\in \mathbf{span}\{\nabla u(c)\}^{\perp}$, $v\neq0$ (i.e., the Hessian of $u(\cdot)$ is negative definite on the tangent space of $\mathcal{I}(c)$ at $c$ or, equivalently, the Gaussian curvature of $\mathcal{I}(c)$ at $c$ is nonzero);
        \item[iii.] there is $\varepsilon\in(0,1)$ such that
    \begin{eqnarray*}
         \biggr\Vert\frac{(\tilde{c}-c)}{\Vert \tilde{c}-c\Vert}\mathbf{H}u(c)(I-e_{L}^{T}(f(c),1))\biggr\Vert^{2}>\varepsilon,
    \end{eqnarray*}
    for all $\tilde{c}\in B(c,\varepsilon)\cap\mathcal{I}(c)\setminus\{c\}\subset\mathbb{R}^{L}_{++}$.
    \end{itemize}
\end{lemma}

Now we are able to state the following corollary of Theorem \ref{theoAwesome}.
\begin{corollary}\label{corDepthProj}
    Let $c\in\mathbb{R}^{L}_{++}$, $p=x_{n}^{-1}(c)$, $\pi_{i}(\cdot)$, $\pi_{o}(\cdot)$ and $f(\cdot)$ be as defined above. Then, $\pi_{i}(f(c))=\pi_{o}(f(c))=0$, $\nabla \pi_{i}(f(c))=\nabla \pi_{o}(f(c))=0$, and
    \begin{eqnarray*}
        \mathbf{H}\pi_{i}(f(c))=-\frac{p_{L}^{2}}{\Vert p \Vert}M \mathbf{H}_{11}e(p,v_{n}(p)) M^{T}\quad\textrm{and}\quad \mathbf{H}\pi_{o}(f(c))=2\mathbf{H}\pi_{i}(f(c)),
    \end{eqnarray*}
    with $M\mathbf{H}_{11}e(p,v_{n}(p))M^{T}\in\mathbb{R}^{(L-1)\times(L-1)}$ a symmetric negative definite matrix.
\end{corollary}

Corollary \ref{corDepthProj} reveals that, up to second order around $c\in\mathbb{R}^{L}_{++}$, the offer hypersurface always runs exactly twice as deep as the indifference hypersurface, regardless of preferences and prices. This universal relation is a consequence of the duality identity $\lambda_{n}(p)\partial e(p,v_{n}(p))/\partial u=1$, $p\in\mathbb{R}^{L}_{++}$. The following example illustrates Corollary \ref{corDepthProj}.
\begin{example}\label{ex2}
I build on Example \ref{ex1}. For $c\in\mathbb{R}^{2}_{++}$, Definition \ref{defDepthProjection} implies
\begin{eqnarray*}
    \pi_{i}(q)&=&\sqrt{\frac{c_{1}c_{2}}{q(c_{1}^{2}+c_{2}^{2})}}(\sqrt{c_{2}}-\sqrt{qc_{1}})^{2}\\
    \pi_{o}(q)&=&\frac{(c_{2}-qc_{1})^{2}}{2q\sqrt{c_{1}^{2}+c_{2}^{2}}}
\end{eqnarray*}
for $q>0$. Also, $f(c)=\partial_{1} u(c)/\partial_{2} u(c)=c_{2}/c_{1}$ and, therefore, $\pi_{i}(f(c))=\pi_{o}(f(c))=0$. Next,
\begin{eqnarray*}
    \pi_{i}^{\prime}(q)&=&-\frac{1}{2}\sqrt{\frac{c_{1}c_{2}}{q^{3}(c_{1}^{2}+c_{2}^{2})}}(\sqrt{c_{2}}-\sqrt{qc_{1}})^{2}-\sqrt{\frac{c_{1}c_{2}}{q(c_{1}^{2}+c_{2}^{2})}}(\sqrt{c_{2}}-\sqrt{qc_{1}})\sqrt{\frac{c_{1}}{q}}\\
    \pi_{o}^{\prime}(q)&=&-\frac{(c_{2}-qc_{1})^{2}}{2q^{2}\sqrt{c_{1}^{2}+c_{2}^{2}}}-\frac{(c_{2}-qc_{1})}{q\sqrt{c_{1}^{2}+c_{2}^{2}}}c_{1}
\end{eqnarray*}
for $q>0$. Therefore, $\pi_{i}^{\prime}(f(c))=\pi_{o}^{\prime}(f(c))=0$. Finally, notice that
\begin{eqnarray*}
    \pi^{\prime\prime}_{i}(f(c))&=&\dfrac{c_{1}^{3}}{2c_{2}\sqrt{c_{1}^{2}+c_{2}^{2}}}>0\\
    \pi^{\prime\prime}_{o}(f(c))&=&\dfrac{c_{1}^{3}}{c_{2}\sqrt{c_{1}^{2}+c_{2}^{2}}}>0.
\end{eqnarray*}
Since $p=x_{n}^{-1}(c)=(1/2c_{1},1/2c_{2})$, Example \ref{ex1} implies
\begin{eqnarray*}
    -\frac{p_{2}^{2}}{\Vert p \Vert} M \mathbf{H}_{11}e(p,v_{n}(p)) M^{T}=\frac{p_{2}^{2}}{\Vert p \Vert }\frac{1}{4p_{1}^{2}}=\frac{c_{1}^{3}}{2c_{2}\sqrt{c_{1}^{2}+c_{2}^{2}}}=\pi^{\prime\prime}_{i}(f(c))=\frac{\pi^{\prime\prime}_{o}(f(c))}{2},
\end{eqnarray*}
in accordance with Corollary \ref{corDepthProj}.
\end{example}

The next result uses Corollary \ref{corDepthProj} and the Taylor expansion to provide upper and lower bounds for the depth projections in a neighborhood of $c\in\mathbb{R}^{L}_{++}$.

\begin{proposition}\label{propBoundsHessian}
    Let $c\in\mathbb{R}^{L}_{++}$, $p=x^{-1}_{n}(c)$, $\pi_{i}(\cdot)$ and $\pi_{o}(\cdot)$ be as defined above. Then, for $0<\mu_{\min}(p)\leq\mu_{\max}(p)$ the smallest and greatest eigenvalues of $-M\mathbf{H}_{11}e(p,v_{n}(p))M^{T}\in\mathbb{R}^{(L-1)\times(L-1)}$, there is $\varepsilon>0$ such that
    \begin{eqnarray*}
        \frac{p_{L}^{2}}{\Vert p\Vert}\frac{\mu_{\min}(p)}{4}\Vert q -f(c)\Vert^{2}\leq \pi_{i}(q)\leq\frac{p_{L}^{2}}{\Vert p\Vert}\mu_{\max}(p)\Vert q -f(c)\Vert^{2}
    \end{eqnarray*}
    and
    \begin{eqnarray*}
        \frac{p_{L}^{2}}{\Vert p\Vert}\frac{\mu_{\min}(p)}{2}\Vert q -f(c)\Vert^{2}\leq \pi_{o}(q)\leq\frac{p_{L}^{2}}{\Vert p\Vert}2\mu_{\max}(p)\Vert q -f(c)\Vert^{2},
    \end{eqnarray*}
    for $q\in B(f(c),\varepsilon)$.
\end{proposition}

Proposition \ref{propBoundsHessian} furnishes upper and lower bounds for the depth projections when one considers $q\in\mathbb{R}^{L-1}_{++}$ in a neighborhood of $f(c)$. For the results regarding the Cass criterion, however, we will need bounds that are provided by the shift in the consumption bundle itself, not in the marginal rates of substitution. 

\begin{theorem}\label{theoFundamentalIneq}
   Let $c\in\mathbb{R}^{L}_{++}$, $p=x_{n}^{-1}(c)$, $\pi_{i}(\cdot)$ and $f(\cdot)$ be as defined above. Then, there is $\varepsilon\in(0,1)$ such that $\pi_{i}(f(\tilde{c}))\geq \varepsilon\Vert\tilde{c}-c\Vert^{2}/\Vert c\Vert$, for $\tilde{c}\in B(c,\varepsilon\Vert c\Vert)\cap\mathcal{I}(c)\subseteq\mathbb{R}^{L}_{++}$.
\end{theorem}

Theorem \ref{theoFundamentalIneq} reveals a local quadratic lower bound for the depth projection of the indifference hypersurface when one considers not the marginal rates of substitution as arguments but the consumption vectors themselves. Similarly, the next result provides a local quadratic upper bound when $L=2$.

\begin{theorem}\label{theoUpperBoundCass}
    Let $L=2$, $c\in\mathbb{R}^{2}_{++}$, $p=x_{n}^{-1}(c)$, $\pi_{i}(\cdot)$ and $f(\cdot)$ be as defined above. Then, there is $\varepsilon\in(0,1)$ such that $\pi_{i}(f(\tilde{c}))\leq (\tilde{c}_{2}-c_{2})^{2}/(\varepsilon\Vert c\Vert)$, for $\tilde{c}\in B(c,\varepsilon\Vert c\Vert)\cap\mathcal{I}(c)\subset\mathbb{R}^{2}_{++}$.
\end{theorem}

Theorems \ref{theoFundamentalIneq} and \ref{theoUpperBoundCass} support the statements of the local assumptions that ensure the sufficiency and necessity of the Cass criterion in \hyperref[sec4]{Section 4}. The last result of this section is a lemma that will also prove useful.

\begin{lemma}\label{lemmaInclined}
    Let $L=2$, $c\in\mathbb{R}^{2}_{++}$ and $p=x_{n}^{-1}(c)$. Then, there is $\varepsilon\in(0,1)$ such that if $\tilde{c}\in B(c,\varepsilon \Vert c\Vert)$, $\tilde{c}_{1}<c_{1}$, $\tilde{c}_{2}-c_{2}\geq (c_{1}-\tilde{c}_{1})/\varepsilon$ and $p\cdot(\tilde{c}-c)>0$, then $u(\tilde{c})>u(c)$.
\end{lemma}

\section{The depth projection of the transformation frontier}\label{sec3}

The depth projection of the transformation frontier is defined in complete analogy with that of the indifference hypersurface, with the profit function taking over the role of the expenditure function. The \textit{production set}\footnote{I follow \textcite[Chapter 5]{Mas-ColellWhinstonGreen_1995} for definitions and results on production.} is $Y=\{z\in\mathbb{R}^{L}\mid F(z)\leq0\}$, where $F:\mathbb{R}^{L}\rightarrow\mathbb{R}$ is a smooth \textit{transformation function}, and the \textit{transformation frontier} is $\mathcal{T}=\{z\in\mathbb{R}^{L}\mid F(z)=0\}$. Production plans $z\in Y$ follow the usual sign convention (negative coordinates are inputs and positive coordinates are outputs). I assume that $Y$ is strictly convex, that $\nabla F(z)\in\mathbb{R}^{L}_{++}$ for $z\in\mathcal{T}$ (so that $\mathcal{T}$ is a smooth hypersurface with a strictly positive normal vector at each of its points, the analog of $u(\cdot)$ being strictly increasing), and that the profit maximization problem
\begin{eqnarray}\label{PMP}
    \Pi(p)=\max_{z\in Y}p\cdot z
\end{eqnarray}
has a solution for every $p\in\mathbb{R}^{L}_{++}$. Then, the \textit{supply function} $s:\mathbb{R}^{L}_{++}\rightarrow\mathbb{R}^{L}$, $s(p)=\argmax_{z\in Y}p\cdot z$, is single-valued and the \textit{profit function} $\Pi:\mathbb{R}^{L}_{++}\rightarrow\mathbb{R}$ is given by (\ref{PMP}), so that $\Pi(p)=p\cdot s(p)$. I also assume that $s(\cdot)$ is smooth.\footnote{As in the consumer case, this is a nondegeneracy assumption: if $\mathbf{H}F(z)$ is nonsingular on the tangent space of $\mathcal{T}$ at every $z\in\mathcal{T}$, then the implicit function theorem applied to the first-order conditions of (\ref{PMP}) implies that $s(\cdot)$ is smooth.}

The following properties are standard \parencite[Proposition 5.C.1]{Mas-ColellWhinstonGreen_1995}: $\Pi(\cdot)$ is convex and homogeneous of degree one; $s(\cdot)$ is homogeneous of degree zero; $\nabla\Pi(p)=s(p)$ (Hotelling's lemma); and $\mathbf{J}s(p)=\mathbf{H}\Pi(p)$ is symmetric, positive semi-definite, and satisfies $p\mathbf{H}\Pi(p)=0$, $p\in\mathbb{R}^{L}_{++}$. Furthermore, the first-order conditions of (\ref{PMP}) imply that there is $\lambda>0$ such that $p=\lambda\nabla F(s(p))$, so that $s(p)\in\mathcal{T}$ and the normal vector to $\mathcal{T}$ at $s(p)$ is proportional to $p$.

For $y\in\mathcal{T}$, the analog of the trade hyperplane is the \textit{iso-profit hyperplane} $\mathcal{H}_{Y}(y)=\{z\in\mathbb{R}^{L}\mid \nabla F(y)\cdot(z-y)=0\}$, which supports the strictly convex set $Y$ at $y$. As for the indifference hypersurface, $\mathcal{T}$ can also be parametrized by its normal vector. To ease notation, let $f_{f}(z)=\nabla_{L-1}F(z)/\partial_{L}F(z)$, which is well defined and smooth on an open neighborhood of $\mathcal{T}$ and takes values in $\mathbb{R}^{L-1}_{++}$ on $\mathcal{T}$.

\begin{proposition}\label{propFrontierParametrization}
    Let $Y$, $\mathcal{T}$, $s(\cdot)$ and $f_{f}(\cdot)$ be as defined above. Then, $\mathcal{T}=\{s(q,1)\mid q\in\mathbb{R}^{L-1}_{++}\}$ and
    \begin{eqnarray*}
        \frac{\nabla F(s(q,1))}{\Vert \nabla F(s(q,1))\Vert}=\frac{(q,1)}{\Vert (q,1)\Vert},
    \end{eqnarray*}
    for $q\in\mathbb{R}^{L-1}_{++}$. Furthermore, for $z\in\mathcal{T}$, $z=s(f_{f}(z),1)$.
\end{proposition}

This motivates the following definition.

\begin{definition}\label{defDepthProjectionProduction}
    For $y\in\mathcal{T}$, the \textit{depth projection of the transformation frontier $\mathcal{T}$ at $y$}, $\pi_{f}:\mathbb{R}^{L-1}_{++}\rightarrow\mathbb{R}_{+}$, is given by
    \begin{eqnarray*}
        \pi_{f}(q)=\frac{\nabla F(y)}{\Vert \nabla F(y)\Vert}\cdot(y-s(q,1)),
    \end{eqnarray*}
    for $q\in\mathbb{R}^{L-1}_{++}$.
\end{definition}

The depth projection in Definition \ref{defDepthProjectionProduction} has an intuitive economic interpretation through the lens of the profit maximization problem (PMP). Think of $\nabla F(y)/\Vert \nabla F(y)\Vert\in\mathbb{R}^{L}_{++}$ as a price vector. Then, the PMP tells us that $y\in\mathcal{T}$ is the point that maximizes profits under prices $\nabla F(y)/\Vert \nabla F(y)\Vert$. Therefore, if we look at any point on the transformation frontier $s(q,1)$, $(q,1)\nparallel\nabla F(y)$, it will imply a profit loss of $\pi_{f}(q)=(\nabla F(y)/\Vert \nabla F(y)\Vert)\cdot(y-s(q,1))>0$.

Geometrically, the depth projection $\pi_{f}(\cdot)$ is simply the length of the orthogonal projection, in the direction $\nabla F(y)/\Vert \nabla F(y)\Vert$, of the segment joining $y$ to the corresponding points on the transformation frontier, as seen in Figure \ref{FigProduction}. In contrast to Figure \ref{FigDepth}, notice that this depth projection now runs into the production set $Y$, in the direction $-\nabla F(y)/\Vert\nabla F(y)\Vert$, since $Y$ lies below the iso-profit hyperplane $\mathcal{H}_{Y}(y)$.

\begin{figure}[H]
\centering
\includegraphics[width=0.5\textwidth]{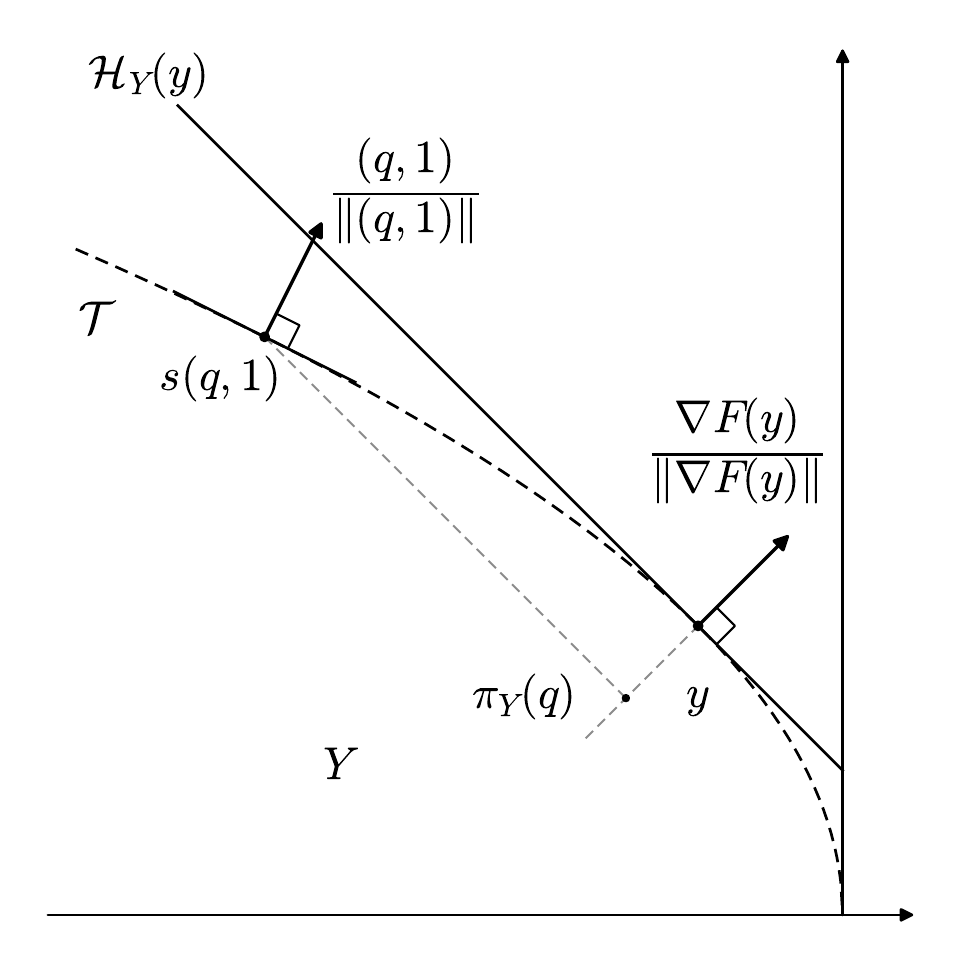}
\caption{The depth projection $\pi_{f}(\cdot)$ for the technology of Example \ref{exProduction} and $y=(-1/2,1)$.}
\label{FigProduction}
\end{figure}

For a given $p\in\mathbb{R}^{L}_{++}$, define the auxiliary function $\xi:\mathbb{R}^{L}_{++}\rightarrow\mathbb{R}_{+}$ as
\begin{eqnarray*}
    \xi(q)=p\cdot(s(p)-s(q)),
\end{eqnarray*}
where, once again, the dependence on $p$ is omitted. Notice that $\xi(q)=\Pi(p)-p\cdot s(q)\geq0$, $q\in\mathbb{R}^{L}_{++}$. Furthermore, for $y=s(p)$, Proposition \ref{propFrontierParametrization} implies $\nabla F(y)/\Vert\nabla F(y)\Vert=p/\Vert p\Vert$ and, therefore, $\pi_{f}(q)=\xi(q,1)/\Vert p\Vert$, $q\in\mathbb{R}^{L-1}_{++}$. The next result is the production analog of Theorem \ref{theoAwesome}.

\begin{theorem}\label{theoProduction}
    Let $p\in\mathbb{R}^{L}_{++}$ and $\xi(\cdot)$ be as defined above. Then, $\xi(p)=0$, $\nabla\xi(p)=0$, and
    \begin{eqnarray*}
        \mathbf{H}\xi(p)=\mathbf{H}\Pi(p)=\mathbf{J}s(p).
    \end{eqnarray*}
\end{theorem}

An immediate corollary, analogous to Corollary \ref{corDepthProj}, is the following.

\begin{corollary}\label{corDepthProjProduction}
    Let $y\in\mathcal{T}$, $p=(f_{f}(y),1)$, $\pi_{f}(\cdot)$ and $M$ be as defined above. Then, $\pi_{f}(f_{f}(y))=0$, $\nabla\pi_{f}(f_{f}(y))=0$, and
    \begin{eqnarray*}
        \mathbf{H}\pi_{f}(f_{f}(y))=\frac{1}{\Vert p\Vert}M\mathbf{H}\Pi(p)M^{T}.
    \end{eqnarray*}
\end{corollary}

Theorem \ref{theoProduction} and Corollary \ref{corDepthProjProduction} reveal that the second-order behavior of the depth projection of the transformation frontier is entirely governed by the Hessian of the profit function (i.e., by the supply substitution matrix $\mathbf{J}s(p)$), exactly as the second-order behavior of $\pi_{i}(\cdot)$ is governed by the Hessian of the expenditure function (i.e., by the Slutsky matrix). The sign difference between $\mathbf{H}\eta(p)=-\mathbf{H}_{11}e(p,v_{n}(p))$ and $\mathbf{H}\xi(p)=\mathbf{H}\Pi(p)$ simply reflects the concavity of $e(\cdot,u)$ and the convexity of $\Pi(\cdot)$, so that both Hessians are positive semi-definite. The following example illustrates these results.

\begin{example}\label{exProduction}
    Let $L=2$ and $F(z)=z_{2}-\sqrt{-2z_{1}}$ for $z_{1}\leq0$ (i.e., a single-input, single-output technology with production function $z_{1}\mapsto\sqrt{-2z_{1}}$), so that $\nabla F(z)=(1/\sqrt{-2z_{1}},1)$ on $\mathcal{T}$ and $Y$ is strictly convex. For $p\in\mathbb{R}^{2}_{++}$, the supply function is $s(p)=(-p_{2}^{2}/2p_{1}^{2},\,p_{2}/p_{1})$, the profit function is $\Pi(p)=p_{2}^{2}/2p_{1}$, and
    \begin{eqnarray*}
        \mathbf{H}\Pi(p)=\begin{bmatrix}
            \dfrac{p_{2}^{2}}{p_{1}^{3}} & -\dfrac{p_{2}}{p_{1}^{2}}\\
            -\dfrac{p_{2}}{p_{1}^{2}} & \dfrac{1}{p_{1}}
        \end{bmatrix},
    \end{eqnarray*}
    which is positive semi-definite and satisfies $p\mathbf{H}\Pi(p)=0$. Take $y=(-1/2,1)\in\mathcal{T}$, so that $f_{f}(y)=1$, $p=(1,1)$, $s(p)=y$ and $\Vert p\Vert=\sqrt{2}$. Then, Definition \ref{defDepthProjectionProduction} implies
    \begin{eqnarray*}
        \pi_{f}(q)=\frac{1}{\Vert p\Vert}\biggr(\Pi(p)-p\cdot s(q,1)\biggr)=\frac{1}{\sqrt{2}}\biggr(\frac{1}{2}+\frac{1}{2q^{2}}-\frac{1}{q}\biggr)=\frac{(q-1)^{2}}{2\sqrt{2}q^{2}},
    \end{eqnarray*}
{\sloppy for $q>0$, so that $\pi_{f}(1)=\pi_{f}^{\prime}(1)=0$ and $\pi_{f}^{\prime\prime}(1)=1/\sqrt{2}$. On the other hand, $M\mathbf{H}\Pi(p)M^{T}/\Vert p\Vert=(1/\sqrt{2})(p_{2}^{2}/p_{1}^{3})=1/\sqrt{2}$, in accordance with Corollary \ref{corDepthProjProduction}. Figure \ref{FigProduction} depicts the case $q=1/2$, for which $s(q,1)=(-2,2)$ and $\pi_{f}(1/2)=\sqrt{2}/4$, the same value as $\pi_{i}(4)$ for $c=(1,1)$ in Example \ref{ex2}, which is the depth projection depicted in Figure \ref{FigDepth} (both figures are drawn at the same scale, with $\nabla F(y)=\nabla u(c)=(1,1)$).\par}
\end{example}

The last results of this section are the analogs of Proposition \ref{propBoundsHessian} and Theorem \ref{theoFundamentalIneq} for the transformation frontier. In particular, Theorem \ref{theoFundamentalIneqProduction} provides a local quadratic lower bound for $\pi_{f}(\cdot)$ in terms of the production plans themselves, which is what the sufficiency of the Cass criterion requires from the production side of the economy in \hyperref[sec4]{Section 4}.

\begin{proposition}\label{propBoundsHessianProduction}
    Let $y\in\mathcal{T}$, $p=(f_{f}(y),1)$ and $\pi_{f}(\cdot)$ be as defined above. Then, $M\mathbf{H}\Pi(p)M^{T}\in\mathbb{R}^{(L-1)\times(L-1)}$ is symmetric and positive definite and, for $0<\mu^{f}_{\min}(p)\leq\mu^{f}_{\max}(p)$ its smallest and greatest eigenvalues, there is $\varepsilon>0$ such that
    \begin{eqnarray*}
        \frac{\mu^{f}_{\min}(p)}{4\Vert p\Vert}\Vert q-f_{f}(y)\Vert^{2}\leq\pi_{f}(q)\leq\frac{\mu^{f}_{\max}(p)}{\Vert p\Vert}\Vert q-f_{f}(y)\Vert^{2},
    \end{eqnarray*}
    for $q\in B(f_{f}(y),\varepsilon)$.
\end{proposition}

\begin{theorem}\label{theoFundamentalIneqProduction}
    Let $y\in\mathcal{T}\setminus\{0\}$, $\pi_{f}(\cdot)$ and $f_{f}(\cdot)$ be as defined above. Then, there is $\varepsilon\in(0,1)$ such that $\pi_{f}(f_{f}(\tilde{y}))\geq \varepsilon \Vert\tilde{y}-y\Vert^{2}/\Vert y\Vert$, for $\tilde{y}\in B(y,\varepsilon\Vert y\Vert)\cap\mathcal{T}$.
\end{theorem}

\section{The local nature of the Cass criterion}\label{sec4}
This section derives general statements for the necessity and sufficiency of the Cass criterion in a production overlapping generations economy, with possibly unbounded dynamics for the demography and per capita endowments. The assumptions are made ``on demand'' in a step-by-step proof that closely follows the reasoning by \textcite{Benveniste1976}. 

Before defining the economy, however, there is a fundamental point that must be made regarding the \textcite{BalaskoCassShell_1980} relabeling algorithm, since this will justify the way in which I later define the overlapping generations economy.

\subsection{The Balasko--Cass--Shell relabeling algorithm}\label{subsecBCS}

The relabeling algorithm from \textcite{BalaskoCassShell_1980} is a procedure to pose any overlapping generations economy as a two-period one (i.e., an economy in which each household lives for two periods). 

Let me illustrate it in a simple consumption-loan economy. The economy is a 3-period one with a single household in each ``natural time-counting generation'' (so that generations and households share a common index) and a single perishable good in each period. The initial period is $t=1$, so that household $h=-1$ lives for a single period; $h=0$ lives for two periods; and $h\geq1$ lives for three periods. Accordingly, let the utility functions be given by: $u^{-1}:\mathbb{R}_{+}\rightarrow\mathbb{R}$, $u^{-1}(c_{3})= c_{3}$; $u^{0}:\mathbb{R}^{2}_{+}\rightarrow\mathbb{R}$, $u^{0}(c_{2},c_{3})=c_{2}c_{3}$; $u^{t}:\mathbb{R}^{3}_{+}\rightarrow\mathbb{R}$, $u^{t}(c_{1},c_{2},c_{3})=c_{1}c_{2}c_{3}$ (here, the subscript of $c_{i}$ refers to the $i$-th period of the household's life, so that $c_{3}$ is the consumption of $h=-1$ in $t=1$ and $(c_{2},c_{3})$ is that of $h=0$ in $t=1,2$; after the relabeling below, subscripts refer instead to periods and goods). 
    
Following \textcite{BalaskoCassShell_1980}, we can relabel time periods and generations as $G_{t}=\{2t-1,2t\}$, $t\geq0$, and pose this overlapping generations economy as a two-period one with two goods in each period. To see this, let us first consider household $h=-1\in G_{0}$. There are two goods in $t=1$, so we can represent his preferences by $U^{-1}:\mathbb{R}^{2}_{+}\rightarrow\mathbb{R}$, $U^{-1}(c_{11},c_{12})=c_{11}=u^{-1}(c_{11})$ (the subscript in $c_{12}$ must be read as ``the second good of the first period'' and so on). For household $h=0\in G_{0}$, we have $U^{0}(c_{11},c_{12})= c_{11}c_{12}=u^{0}(c_{11},c_{12})$. 

Next, consider generation $G_{t}$, $t\geq1$. Then, for $h=2t-1\in G_{t}$, the utility function is written as $U^{2t-1}(c_{t1},c_{t2},c_{(t+1)1},c_{(t+1)2})=c_{t1}c_{t2}c_{(t+1)1}=u^{2t-1}(c_{t1},c_{t2},c_{(t+1)1})$, and, for $h=2t\in G_{t}$, as $U^{2t}(c_{t1},c_{t2},c_{(t+1)1},c_{(t+1)2})=c_{t2}c_{(t+1)1}c_{(t+1)2}=u^{2t}(c_{t2},c_{(t+1)1},c_{(t+1)2})$.

To ease the expressions above, we can define convenient \textit{selection functions} $S^{h}(\cdot)$, $h\geq -1$, such that $U^{h}(c)=u^{h}(S^{h}(c))$, for $c\in\mathbb{R}^{2}_{++}$, if $h\in\{-1,0\}$, or $c\in\mathbb{R}^{4}_{++}$, if $h\geq1$.

This setting leads us to two fundamental remarks. The first is that, in order to have no loss of generality when working with two-period overlapping generations economies due to the \textcite{BalaskoCassShell_1980} relabeling algorithm, one must work with utility functions that are nondecreasing and semi-strictly quasiconcave (i.e., as $U^{h}(\cdot)$). Furthermore, since these utility functions actually focus only on a \textit{selected range of commodities} available in each relabeled period, preferences over these desired commodities can be modeled by a selection function and a strictly increasing and strictly quasiconcave utility function (i.e., $S^{h}(\cdot)$ and $u^{h}(\cdot)$, respectively). 

The second remark is the following. When applying the \textcite{BalaskoCassShell_1980} relabeling algorithm and, therefore, potentially ``enlarging'' the domain of the utility functions (i.e., passing from $u^{h}(\cdot)$ to $U^{h}(\cdot)$), one is actually setting the Gaussian curvature of all indifference curves to zero.

One intuitive way of seeing this is to think of an infinite cylinder in $\mathbb{R}^{3}$ and a cross-sectional circle of it (notice that these are the analogs of the indifference hypersurfaces of $U^{h}(\cdot)$ and $u^{h}(\cdot)$). The cross-sectional circle, seen as a hypersurface in $\mathbb{R}^{2}$, has a strictly positive Gaussian curvature everywhere equal to the inverse of its radius. The cylinder, however, has zero Gaussian curvature everywhere due to the Theorema Egregium, since it is locally isometric to the plane. 

Therefore, if one wants to have the full generality of the aforementioned two-period overlapping generations economy due to the \textcite{BalaskoCassShell_1980} algorithm, then one cannot rely on the nonzero Gaussian curvature of the indifference hypersurfaces. This is particularly important since there are results in the literature that rely on uniform bounds on the Gaussian curvature (e.g., \textcite[Proposition 5.6, pp. 296--297]{BalaskoShell_1980}) to prove the validity of the Cass criterion. Our goal in this section, therefore, is to prove the Cass criterion without recourse to assumptions about the Gaussian curvature of $U^{h}(\cdot)$, so that we can always invoke it along with the \textcite{BalaskoCassShell_1980} relabeling algorithm.

\subsection{The production overlapping generations economy}\label{subsecModel}

The economy $\mathcal{E}$ is a production overlapping generations one with discrete time periods $t\in\mathbb{N}$ and $L_{t}\geq1$ perishable commodities in each period. Households are indexed by $h\in \mathbb{N}$ and gathered in generations $G_{t}=\{h\in\mathbb{N}\mid h \textrm{ is born in period }t\}$, $t\geq0$, according to their period of birth. All households live for two periods, the one they are born into and the next, except for those from generation $G_{0}$ who only live during period $t=1$. Generation $G_{t}$, $t\geq0$, has a finite number of households given by $H_{t}\in\mathbb{N}$. 

Household $h\in G_{t}$, $t\geq1$, is defined through a consumption set $X^{h}=\mathbb{R}^{L_{t}+L_{t+1}}_{+}$, an endowment $\omega^{h}\in \mathbb{R}^{L_{t}+L_{t+1}}_{+}\setminus\{0\}$, and a smooth, nondecreasing, and semi-strictly quasiconcave utility function $U^{h}:\mathbb{R}^{L_{t}+L_{t+1}}_{+}\rightarrow\mathbb{R}$ that can be written as $U^{h}(c)=u^{h}(S^{h}(c))$, with $S^{h}:\mathbb{R}^{L_{t}+L_{t+1}}\rightarrow\mathbb{R}^{L_{h}}$, $2\leq L_{h}\leq L_{t}+L_{t+1}$, a \textit{selection function}, that is, a coordinate-projection matrix that selects $L_{h}$ coordinates from its $(L_{t}+L_{t+1})$-dimensional argument, and $u^{h}:\mathbb{R}^{L_{h}}\rightarrow\mathbb{R}$ a smooth, strictly increasing, and strictly quasiconcave utility function. 

The coordinates selected by $S^{h}(\cdot)$ define the \textit{desired commodities} by household $h$ and, when convenient, we write $S^{h}(\cdot)=(S^{h}_{t}(\cdot),S^{h}_{t+1}(\cdot))$, $h\in G_{t}$, $t\geq1$ (i.e., we separate the desired commodities by period). Household $h\in G_{0}$ is defined similarly, but with the dimension of the consumption set given by $L_{1}$ (i.e., $X^{h}=\mathbb{R}^{L_{1}}_{+}$ and $\omega^{h}\in\mathbb{R}^{L_{1}}_{+}\setminus\{0\}$).

Firms are indexed by $j\in\mathbb{N}$ and are gathered according to the generation of their owners $G^{f}_{t}=\{j\in\mathbb{N}\mid j \textrm{ is owned only by households in } G_{t}\}$, with $G^{f}_{t}$ having a finite number of firms given by $H^{f}_{t}\in\mathbb{N}\cup\{0\}$, $t\geq0$. Consequently, for each $j\in G^{f}_{t}$, there are $\theta_{jh}\in[0,1]$, $h\in G_{t}$, $\sum_{h\in G_{t}}\theta_{jh}=1$, representing the fraction of firm $j\in G^{f}_{t}$ owned by $h\in G_{t}$. Each firm $j\in G^{f}_{t}$, $t\geq1$, is defined, similarly to the household case, through a selection function (coordinate-projection matrix) $S^{j}:\mathbb{R}^{L_{t}+L_{t+1}}\rightarrow\mathbb{R}^{L_{j}}$, $2\leq L_{j}\leq L_{t}+L_{t+1}$, and a smooth transformation function $F^{j}:\mathbb{R}^{L_{j}}\rightarrow\mathbb{R}$ satisfying the assumptions of Section \ref{sec3}, with production set
\begin{eqnarray*}
    \tilde{Y}^{j}=\{z\in\mathbb{R}^{L_{t}+L_{t+1}}\mid F^{j}(S^{j}(z))\leq 0 \textrm{ and } z_{i}=0 \textrm{ if } i \textrm{ is not selected by } S^{j}(\cdot)\}
\end{eqnarray*}
(i.e., firm $j$ only trades the commodities selected by $S^{j}(\cdot)$; without the second requirement, positive amounts of the non-selected commodities could be produced at no cost). Firms $j\in G^{f}_{0}$ are defined in the same way, with $\mathbb{R}^{L_{1}}$ in place of $\mathbb{R}^{L_{t}+L_{t+1}}$.

I also assume that $F^{j}(0)=0$, $j\geq1$, so that inaction is possible and lies on the transformation frontier. Since $Y^{j}=\{z\in\mathbb{R}^{L_{j}}\mid F^{j}(z)\leq0\}$ is convex and $\nabla F^{j}(0)\in\mathbb{R}^{L_{j}}_{++}$, the supporting hyperplane of $Y^{j}$ at $0$ yields $\nabla F^{j}(0)\cdot z\leq0$, $z\in Y^{j}$, so that $Y^{j}\cap\mathbb{R}^{L_{j}}_{+}=\{0\}$ (firm $j$ cannot produce something from nothing) and $\Pi^{j}(p)\geq0$, $p\in\mathbb{R}^{L_{j}}_{++}$.

The supply and profit functions $s^{j}(\cdot)$ and $\Pi^{j}(\cdot)$ of Section \ref{sec3} are, with a slight abuse of notation, regarded as functions of the full price vector $(p_{t},p_{t+1})\in\mathbb{R}^{L_{t}+L_{t+1}}_{++}$ (they are evaluated at the selected prices $S^{j}(p_{t},p_{t+1})$ and, in the case of $s^{j}(\cdot)$, the resulting production plan is embedded in $\mathbb{R}^{L_{t}+L_{t+1}}$ by placing zeros in the non-selected coordinates), and the same convention applies to the normalized Walrasian demand $x^{h}_{n}(\cdot)$ of household $h$, $h\geq1$. For $t\geq1$, let $\omega_{t}=\sum_{h\in G_{t-1}\cup G_{t}}\omega^{h}_{t}\in\mathbb{R}^{L_{t}}_{+}$ be the aggregate endowment of period $t$.

I proceed with the definition of an equilibrium for this economy. For $h\in G_{t}$, $t\geq1$, the wealth of household $h$\footnote{The upper indices $h$ (e.g., $x^{h}_{n}(\cdot)$ and $\omega^{h}\in\mathbb{R}^{L_{t}+L_{t+1}}_{+}$) are used to reference functions and vectors, while the lower ones are used to reference scalars (e.g., wealth $w_{h}(\cdot)$).} at prices $(p_{t},p_{t+1})\in\mathbb{R}^{L_{t}+L_{t+1}}_{++}$ is
\begin{eqnarray*}
    w_{h}(p_{t},p_{t+1})=(p_{t},p_{t+1})\cdot\omega^{h}+\sum_{j\in G^{f}_{t}}\theta_{jh}\Pi^{j}(p_{t},p_{t+1})>0,
\end{eqnarray*}
(the inequality follows from $\omega^{h}\neq0$ and $\Pi^{j}(\cdot)\geq0$) and his \textit{demand function} $x^{h}:\mathbb{R}^{L_{t}+L_{t+1}}_{++}\rightarrow\mathbb{R}^{L_{t}+L_{t+1}}_{+}$ is
\begin{eqnarray*}
    x^{h}(p_{t},p_{t+1})=x^{h}_{n}\biggr(\frac{(p_{t},p_{t+1})}{w_{h}(p_{t},p_{t+1})}\biggr),
\end{eqnarray*}
which is smooth by Lemma \ref{lemmaBasicIdentities}. For $h\in G_{0}$, $w_{h}(\cdot)$ and $x^{h}(\cdot)$ are defined in the same way as functions of $p_{1}\in\mathbb{R}^{L_{1}}_{++}$, with profits from $j\in G^{f}_{0}$. When convenient, I write $x^{h}(\cdot)=(x^{h}_{t}(\cdot),x^{h}_{t+1}(\cdot))$ and $s^{j}(\cdot)=(s^{j}_{t}(\cdot),s^{j}_{t+1}(\cdot))$ for the period-$t$ and period-$(t+1)$ coordinates of demand and supply, $h\in G_{t}$, $j\in G^{f}_{t}$, $t\geq1$.

\begin{definition}\label{defEquilibrium}
    Let $\mathcal{E}$ be the production overlapping generations economy described above. Then, $p=(p_{1},p_{2},\ldots)\in\mathbb{R}^{\infty}_{++}$ is an \textit{equilibrium} if
    \begin{eqnarray}\label{eqMarketClearing}
        \sum_{h\in G_{t-1}}x^{h}_{t}(p_{t-1},p_{t})+\sum_{h\in G_{t}}x^{h}_{t}(p_{t},p_{t+1})=\omega_{t}+\sum_{j\in G^{f}_{t-1}} s^{j}_{t}(p_{t-1},p_{t})+\sum_{j\in G^{f}_{t}} s^{j}_{t}(p_{t},p_{t+1}),
    \end{eqnarray}
for $t\geq1$, where, for $t=1$, the terms indexed by $G_{0}$ and $G^{f}_{0}$ are evaluated at $p_{1}$ alone (i.e., $(p_{0},p_{1})$ must be read as $p_{1}$).
\end{definition}

In words, $p$ is an equilibrium if every household maximizes utility, every firm maximizes profits, and all period-$t$ markets clear, $t\geq1$. Notice that both sides of (\ref{eqMarketClearing}) are homogeneous of degree zero in $p$, so that equilibrium prices can be normalized by $p_{11}=1$, and the \textit{set of equilibria} of $\mathcal{E}$ is
\begin{eqnarray*}
    \mathcal{M}=\{p\in\mathbb{R}^{\infty}_{++}\mid p_{11}=1 \textrm{ and (\ref{eqMarketClearing}) holds for } t\geq1\}.
\end{eqnarray*}
As a topological space, $\mathbb{R}^{\infty}$ is endowed with the product topology and, when convenient, vectors in $\mathbb{R}^{L_{1}+\cdots+L_{J+1}}$ are regarded as elements of $\mathbb{R}^{\infty}$ with the appropriate embedding.

The existence of equilibrium is established by the reasoning in \textcite{Dognini_2026a}, adapted to the presence of production: the equilibrium existence theorem of \textcite[Theorem 5, p. 119]{ArrowHahn_1971} is applied to the finite \textit{insulated economies} formed by the generations born up to a given period, and an equilibrium of $\mathcal{E}$ is then obtained as a limit point of these finite equilibria through Tychonoff's Theorem.\footnote{This truncation argument goes back to \textcite{BalaskoCassShell_1980} and \textcite{Wilson_1981}; see also \textcite{Burke_1988}. Its main technical requirement is a uniform bound on the equilibrium prices of the insulated economies, which is provided by Assumption \ref{assBorderPrices} below.}

For $J\geq1$, the \textit{insulated economy} $\mathcal{E}_{J}$ is the finite economy with periods $t=1,\ldots,J+1$, commodity space $\mathbb{R}^{L_{1}+\cdots+L_{J+1}}$, households $h\in\bigcup_{t=0}^{J}G_{t}$, and firms $j\in\bigcup_{t=0}^{J}G^{f}_{t}$, all defined exactly as in $\mathcal{E}$ (the consumption set, endowment, utility function and production set of each agent are embedded in $\mathbb{R}^{L_{1}+\cdots+L_{J+1}}$ by placing zeros in the coordinates of the periods in which the agent is not active). A \textit{production allocation} of $\mathcal{E}$ is $y=\{y^{j}\}_{j\geq1}$ with $y^{j}\in\tilde{Y}^{j}$, $j\geq1$, and it is called \textit{finite} if $y^{j}\neq0$ for finitely many $j$. The \textit{net output} of period $t\geq1$ under $y$ is $\sum_{j\in G^{f}_{t-1}\cup G^{f}_{t}}y^{j}_{t}\in\mathbb{R}^{L_{t}}$, which I abbreviate as $\sum_{j}y^{j}_{t}$, and $y$ is \textit{feasible up to period} $T\geq1$ if $\omega_{t}+\sum_{j}y^{j}_{t}\geq0$, $1\leq t\leq T$.

I proceed with the following assumptions on the technology, which are, respectively, Assumptions A.3.4 and A.3.5 in \textcite[pp. 63--65]{ArrowHahn_1971} imposed on each insulated economy, and a uniform version of the boundedness of feasible production that these assumptions deliver in each insulated economy separately \parencite[Theorem 2, p. 67]{ArrowHahn_1971}.

\begin{assumption}[Technology]\label{assTechnology}\leavevmode
    \begin{itemize}
        \item[i.] If $y$ is a finite production allocation of $\mathcal{E}$ and $\sum_{j\geq1}y^{j}\geq0$, then $y^{j}=0$ for all $j\geq1$;
        \item[ii.] For every $J\geq1$, there is a production allocation $y$ of $\mathcal{E}_{J}$ such that $\omega_{t}+\sum_{j}y^{j}_{t}\in\mathbb{R}^{L_{t}}_{++}$, $1\leq t\leq J+1$;
        \item[iii.] There is $\{\bar{y}_{t}\}_{t\geq1}\subset\mathbb{R}_{++}$ such that, for every $T\geq1$ and every finite production allocation $y$ of $\mathcal{E}$ feasible up to period $T$, $\sum_{j}y^{j}_{t}\leq \bar{y}_{t}(1,\ldots,1)$, $1\leq t\leq T$.
    \end{itemize}
\end{assumption}

Item i of Assumption \ref{assTechnology} states that the firms of any insulated economy cannot jointly produce something from nothing and that no production process can be undone by another one (i.e., irreversibility). It must be imposed on the aggregate, since it does not follow from the corresponding property of each firm: the classical example in \textcite[p. 64]{ArrowHahn_1971} is a pair of firms, one producing two units of a commodity from one unit of another, and the other doing the reverse.

Item ii states that the resources and technology of each insulated economy permit a strictly positive net supply of every commodity (e.g., it holds trivially if $\omega_{t}\in\mathbb{R}^{L_{t}}_{++}$, $t\geq1$, by taking $y=0$).

Item iii states that the net output of any period is bounded, with a bound that does not depend on the number of generations that are taken into account (i.e., the arrival of future generations does not enlarge the feasible net output of the current period); the sequence $\{\bar{y}_{t}\}_{t\geq1}$ plays here the role of the uniform bound on endowments in \textcite[Assumption 1]{Dognini_2026a}. It is, therefore, a statement about the time structure of production, and it can be derived from more primitive assumptions.\footnote{For instance, if commodities of period $t+1$ are never inputs for the firms of $G^{f}_{t}$ (i.e., $z\in\tilde{Y}^{j}$ implies $(z_{t},z^{+}_{t+1})\in\tilde{Y}^{j}$, $j\in G^{f}_{t}$, where $z^{+}_{t+1}$ is the positive part of $z_{t+1}$), so that production never runs backwards in time, and item i is imposed on the firms of each pair of adjacent generations $G^{f}_{t-1}\cup G^{f}_{t}$, $t\geq1$, then a straightforward induction on $t$ shows that item i holds for every insulated economy and that item iii follows from \textcite[Theorem 2, p. 67]{ArrowHahn_1971} applied period by period. I prefer to state item iii directly, since it is exactly what the proofs below require.}

The next assumption gathers the two conditions of \textcite[Chapter 5]{ArrowHahn_1971} that ensure that a compensated equilibrium is a competitive one and that its prices are strictly positive. An \textit{allocation} $(c,y)$ of $\mathcal{E}_{J}$ consists of consumption vectors $c^{h}\in X^{h}$ and production plans $y^{j}\in\tilde{Y}^{j}$ for its households and firms, and it is \textit{feasible} under endowments $\{\omega^{h}\}$ if $\sum_{h}c^{h}_{t}\leq\omega_{t}+\sum_{j}y^{j}_{t}$, $1\leq t\leq J+1$.

Following \textcite[Definitions 4 and 5, pp. 117--118]{ArrowHahn_1971}, household $h^{\prime}$ of $\mathcal{E}_{J}$ is \textit{resource related} to household $h^{\prime\prime}$ of $\mathcal{E}_{J}$ if, for every feasible allocation $(c,y)$ of $\mathcal{E}_{J}$, there are an allocation $(c^{\prime},y^{\prime})$ and endowments $\{\omega^{\prime h}\}$ of $\mathcal{E}_{J}$ such that $(c^{\prime},y^{\prime})$ is feasible under $\{\omega^{\prime h}\}$, $U^{h}(c^{\prime h})\geq U^{h}(c^{h})$ for all $h$, $U^{h^{\prime\prime}}(c^{\prime h^{\prime\prime}})>U^{h^{\prime\prime}}(c^{h^{\prime\prime}})$, and $\{\omega^{\prime h}\}$ differs from $\{\omega^{h}\}$ only by an increase in the commodities held by $h^{\prime}$ in positive amounts; and $h^{\prime}$ is \textit{indirectly resource related} to $h^{\prime\prime}$ if there is a finite chain of households from $h^{\prime}$ to $h^{\prime\prime}$ in which each household is resource related to the next.

\begin{assumption}[Resource relatedness]\label{assResourceRelated}\leavevmode
    \begin{itemize}
        \item[i.] For every $J\geq1$, every household of $\mathcal{E}_{J}$ is indirectly resource related to every other household of $\mathcal{E}_{J}$;
        \item[ii.] Every commodity is desired by some household: for $t\geq1$ and $1\leq i\leq L_{t}$, there is $h\in G_{t-1}\cup G_{t}$ such that $S^{h}(\cdot)$ selects the $i$-th commodity of period $t$.
    \end{itemize}
\end{assumption}

Item i of Assumption \ref{assResourceRelated} is a weak condition (see the discussion in \textcite[pp. 117--119]{ArrowHahn_1971}) ensuring that every household has a positive income in equilibrium, and item ii ensures that all equilibrium prices are strictly positive, so that the equilibria of the insulated economies are of the form described in Definition \ref{defEquilibrium}.

The last assumption provides the uniform bound on equilibrium prices. For $t\geq1$ and $\sigma\in(0,1)$, let $\mathcal{B}_{t}(\sigma)=\{p\in\mathbb{R}^{L_{t}+L_{t+1}}_{++}\mid \sigma\leq p_{i}\leq\sigma^{-1}, 1\leq i\leq L_{t}+L_{t+1}\}$ and $\mathcal{B}_{0}(\sigma)=\{p\in\mathbb{R}^{L_{1}}_{++}\mid \sigma\leq p_{i}\leq\sigma^{-1}, 1\leq i\leq L_{1}\}$ be the (compact) ``box-sets'' of prices faced by generations $G_{t}$ and $G_{0}$, respectively. Also, let $b_{0}=(\Vert\omega_{1}\Vert_{\infty}+\bar{y}_{1})/H_{0}$ and
\begin{eqnarray*}
    b_{t}=\frac{\max\{\Vert\omega_{t}\Vert_{\infty}+\bar{y}_{t},\Vert\omega_{t+1}\Vert_{\infty}+\bar{y}_{t+1}\}}{H_{t}},
\end{eqnarray*}
for $t\geq1$, where $\Vert c\Vert_{\infty}=\max_{i}|c_{i}|$. By Assumption \ref{assTechnology}, $b_{t}$ bounds the average demand of generation $G_{t}$ in any allocation that is feasible up to period $t+1$.

\begin{assumption}[Unbounded demand at border prices]\label{assBorderPrices}
    For all $t\geq0$, there is $\sigma_{t}\in(0,1)$ such that
    \begin{eqnarray*}
        \frac{(p_{t},p_{t+1})}{p_{t1}}\notin\mathcal{B}_{t}(\sigma_{t})\implies\biggr\Vert\sum_{h\in G_{t}}\frac{x^{h}(p_{t},p_{t+1})}{H_{t}}\biggr\Vert_{\infty}>b_{t},
    \end{eqnarray*}
    for $(p_{t},p_{t+1})\in\mathbb{R}^{L_{t}+L_{t+1}}_{++}$ (for $t=0$, $(p_{0},p_{1})$ must be read as $p_{1}$).
\end{assumption}

Assumption \ref{assBorderPrices} states that if the prices faced by generation $G_{t}$ move sufficiently close to the border of the corresponding simplex, then its aggregate demand exceeds the resources that could possibly be available in periods $t$ and $t+1$.

It is also convenient to define $\kappa_{t}=\prod_{i=0}^{t-1}\sigma_{i}$, the compact sets $\mathcal{K}_{t}=\{p\in\mathbb{R}^{L_{t}}_{++}\mid \kappa_{t}\leq p_{i}\leq\kappa_{t}^{-1}, 1\leq i\leq L_{t}\}$, $t\geq1$, and $\mathcal{K}=\prod^{\infty}_{t=1}\mathcal{K}_{t}\subset\mathbb{R}^{\infty}_{++}$, which is compact by Tychonoff's Theorem. The following definition of short-sighted equilibrium, adapted from \textcite[Definition 3]{Dognini_2026a}, is also needed.

\begin{definition}\label{defSightedEquilibria}
    The \textit{set of $J$-sighted equilibria} $\mathcal{V}_{J}\subset\mathbb{R}^{L_{1}+\cdots+L_{J+1}}_{++}$, $J\geq1$, of economy $\mathcal{E}$ is the set of $p=(p_{1},\ldots,p_{J+1})\in\mathbb{R}^{L_{1}+\cdots+L_{J+1}}_{++}$ such that $p_{11}=1$, (\ref{eqMarketClearing}) holds for $1\leq t\leq J$, and $(p_{J},p_{J+1})/p_{J1}\in\mathcal{B}_{J}(\sigma_{J})$.
\end{definition}

Definition \ref{defSightedEquilibria} focuses only on ``short-term'' market clearing, limited to the first $J$ periods, and requires the prices faced by the generation on the edge to be compatible with the future overlaps according to Assumption \ref{assBorderPrices}. Let $\mathcal{V}^{\infty}_{J}=\mathcal{V}_{J}\times\prod^{\infty}_{t=J+2}\mathbb{R}^{L_{t}}_{+}$, $J\geq1$, be the natural embedding of $\mathcal{V}_{J}$ into $\mathbb{R}^{\infty}_{+}$.

\begin{lemma}\label{lemmaSightedEquilibria}
    Under Assumptions \ref{assTechnology}--\ref{assBorderPrices}, the set of $J$-sighted equilibria $\mathcal{V}_{J}$, $J\geq1$, is nonempty and compact, $\mathcal{V}_{J}\subseteq\prod^{J+1}_{t=1}\mathcal{K}_{t}$ and $\mathcal{V}^{\infty}_{J+1}\subseteq\mathcal{V}^{\infty}_{J}$.
\end{lemma}

Lemma \ref{lemmaSightedEquilibria} states that $\{\mathcal{V}^{\infty}_{J}\}_{J\geq1}$ is a nested sequence of nonempty sets. Let $\lim_{J\rightarrow\infty}\mathcal{V}^{\infty}_{J}$ denote $\bigcap_{J\geq1}\mathcal{V}^{\infty}_{J}$. The next result is the equilibrium existence theorem for the production overlapping generations economy $\mathcal{E}$.

\begin{theorem}\label{theoExistence}
    Under Assumptions \ref{assTechnology}--\ref{assBorderPrices}, the set of equilibria $\mathcal{M}\subset\mathbb{R}^{\infty}_{++}$ of $\mathcal{E}$ is nonempty and compact, with $\mathcal{M}=\lim_{J\rightarrow\infty}\mathcal{V}^{\infty}_{J}\subseteq\mathcal{K}$.
\end{theorem}

Theorem \ref{theoExistence} shows that equilibria of $\mathcal{E}$ exist and can be approximated by the $J$-sighted equilibria, which are equilibria of the insulated economies $\mathcal{E}_{J}$ up to the market clearing conditions of the last period. Notice that the only role of production in its proof is played by Assumption \ref{assTechnology}: once the feasible net output of each period is bounded, the argument in \textcite{Dognini_2026a} goes through with the households' demand evaluated at wealth levels that include profits.

An \textit{allocation} of $\mathcal{E}$ is a pair $(c,y)$, where $c=\{c^{h}\}_{h\geq1}$ with $c^{h}\in X^{h}$, $h\geq1$, and $y=\{y^{j}\}_{j\geq1}$ is a production allocation (i.e., $y^{j}\in\tilde{Y}^{j}$, $j\geq1$). Its \textit{implicit aggregate endowment} in period $t\geq1$ is
\begin{eqnarray*}
    \omega_{t}(c,y)=\sum_{h\in G_{t-1}\cup G_{t}}c^{h}_{t}-\sum_{j\in G^{f}_{t-1}\cup G^{f}_{t}}y^{j}_{t}\in\mathbb{R}^{L_{t}},
\end{eqnarray*}
that is, the aggregate endowment that would make $(c,y)$ satisfy the market clearing condition of period $t$. The set of \textit{Pareto optimal allocations} is then given by
\begin{eqnarray*}
    \mathcal{P}&=&\{(c,y)\mid (c,y) \textrm{ is an allocation and } \nexists\textrm{ an allocation }(\tilde{c},\tilde{y})\textrm{ s.t. }\\
    &&\omega_{t}(\tilde{c},\tilde{y})=\omega_{t}(c,y), t\geq1,\textrm{ and } U^{h}(\tilde{c}^{h})\geq U^{h}(c^{h}), h\geq1,\\
    &&\textrm{with at least one strict inequality}\}.
\end{eqnarray*}

Notice that this definition of $\mathcal{P}$ does not use any concept of \textit{feasibility} and focuses only on the impossibility of utility enhancing \textit{redistribution} of goods and \textit{redefinition} of production plans, given the aggregate resources implicitly defined by the allocation itself. In particular, it does not depend on the aggregate endowment of the economy\footnote{See \textcite[Definition 2]{Dognini_2026a} for the same approach in a consumption-loan economy and \textcite[Section 4]{Dognini_2025b} for a discussion on the manifold of Pareto optimal allocations.}.

Every $p\in\mathcal{M}$ induces the allocation $(c(p),y(p))$ given by $c^{h}(p)=x^{h}(p_{t},p_{t+1})$ and $y^{j}(p)=s^{j}(p_{t},p_{t+1})$, $h\in G_{t}$, $j\in G^{f}_{t}$, $t\geq0$ (evaluated at $p_{1}$ alone for $t=0$), and (\ref{eqMarketClearing}) reads $\omega_{t}(c(p),y(p))=\omega_{t}$, $t\geq1$. The \textit{subset of Pareto optimal equilibria} of $\mathcal{E}$ is
\begin{eqnarray*}
    \mathcal{M}^{PO}=\{p\in\mathcal{M}\mid (c(p),y(p))\in\mathcal{P}\}.
\end{eqnarray*}
It is well known that in this overlapping generations setting the First Welfare Theorem fails and, therefore, $\mathcal{M}^{PO}$ may be a proper subset of $\mathcal{M}$ (i.e., equilibria may be inefficient). The next subsection provides a necessary condition for $p\in\mathcal{M}$ to belong to $\mathcal{M}^{PO}$.

\subsection{Necessity of the Cass criterion}\label{subsecNec}

The aim of this and the next subsections is to provide a step-by-step proof of general statements for the necessity and sufficiency of the Cass criterion for the equilibria of $\mathcal{E}$, which highlights the assumptions behind it. Assumptions \ref{assTechnology}--\ref{assBorderPrices} are maintained throughout. An additional assumption will be called \textit{pointwise} if it involves only the equilibrium price sequence $p\in\mathcal{M}$ (possibly through the allocation $(c(p),y(p))$ it induces) and the primitives of $\mathcal{E}$, such as endowments, demography, and the bounds of Assumptions \ref{assTechnology}--\ref{assBorderPrices}. It will be called \textit{local} if it depends on the behavior of the indifference hypersurfaces (and of the transformation frontiers) in a neighborhood of each $c^{h}(p)$, $h\geq1$ (and of each $y^{j}(p)$, $j\geq1$).

To ease notation, I write $c^{h}$ and $y^{j}$ for $c^{h}(p)$ and $y^{j}(p)$ whenever $p\in\mathcal{M}$ is fixed, and $\pi^{h}_{i}(\cdot)$ and $f^{h}(\cdot)$ for the depth projection of $\mathcal{I}(S^{h}(c^{h}))$ and the marginal rates of substitution function of $u^{h}(\cdot)$, as defined in \hyperref[sec2]{Section 2}.

Two consequences of Assumptions \ref{assTechnology}--\ref{assBorderPrices} for $p\in\mathcal{M}$ will be used repeatedly. First, Claim \ref{claimBoxes} in the \hyperref[appx]{Appendix} implies $(p_{t},p_{t+1})/p_{t1}\in\mathcal{B}_{t}(\sigma_{t})$, $t\geq0$, so that the ratio between any two coordinates of $(p_{t},p_{t+1})$ lies in $[\sigma^{2}_{t},\sigma^{-2}_{t}]$. Second, since the budget constraint of $h\in G_{t}$ binds and $\Pi^{j}(\cdot)\geq0$, we have $(p_{t},p_{t+1})\cdot c^{h}=w_{h}(p_{t},p_{t+1})\geq(p_{t},p_{t+1})\cdot\omega^{h}$, so that the Cauchy--Schwarz inequality and the previous bounds yield
\begin{eqnarray}\label{eqLowerBoundConsumption}
    \Vert c^{h}\Vert\geq\frac{(p_{t},p_{t+1})\cdot\omega^{h}}{\Vert(p_{t},p_{t+1})\Vert}\geq\frac{\min_{i}\{p_{ti},p_{(t+1)i}\}}{\sqrt{L_{t}+L_{t+1}}\max_{i}\{p_{ti},p_{(t+1)i}\}}\Vert\omega^{h}\Vert_{1}\geq\frac{\sigma^{2}_{t}}{\sqrt{L_{t}+L_{t+1}}}\Vert\omega^{h}\Vert,
\end{eqnarray}
for $h\in G_{t}$, $t\geq1$, where $\Vert z\Vert_{1}=\sum_{i}|z_{i}|\geq\Vert z\Vert$. On the other hand, since $(c,y)$ satisfies (\ref{eqMarketClearing}) for all $t\geq1$, $y$ is feasible up to every period and item iii of Assumption \ref{assTechnology} implies $\sum_{h\in G_{t-1}\cup G_{t}}c^{h}_{t}\leq\omega_{t}+\bar{y}_{t}(1,\ldots,1)$, $t\geq1$. Therefore, defining the \textit{maximum resources} of period $t$ as $R_{t}=\Vert\omega_{t}\Vert_{1}+L_{t}\bar{y}_{t}$, we have
\begin{eqnarray}\label{eqUpperBoundConsumption}
    \sum_{h\in G_{t}}\Vert c^{h}\Vert\leq\sum_{h\in G_{t}}\big(\Vert c^{h}_{t}\Vert+\Vert c^{h}_{t+1}\Vert\big)\leq\sum_{h\in G_{t}}\big(\Vert c^{h}_{t}\Vert_{1}+\Vert c^{h}_{t+1}\Vert_{1}\big)\leq R_{t}+R_{t+1},
\end{eqnarray}
for $t\geq1$. In words, the equilibrium consumption of each household is bounded below by a fraction of his endowment due to (\ref{eqLowerBoundConsumption}), and the aggregate consumption of each generation is bounded above by the total resources of the periods in which it lives due to (\ref{eqUpperBoundConsumption}).

For the necessity of the Cass criterion, I first suppose that $L_{t}=1$, $t\geq1$, and follow the reasoning in \textcite{Benveniste1976}. In this case, $c^{h}=S^{h}(c^{h})\in\mathbb{R}^{2}_{++}$, $h\in G_{t}$, $t\geq1$, and, since $(p_{t},p_{t+1})$ is proportional to $\nabla u^{h}(c^{h})$ by the first-order conditions of utility maximization, $(p_{t},p_{t+1})/\Vert(p_{t},p_{t+1})\Vert$ is the unit normal vector to $\mathcal{I}(c^{h})$ at $c^{h}$.

Let $p\in\mathcal{M}$ be such that $\sum^{\infty}_{t=1}1/(p_{t}\sum_{h\in G_{t}}\Vert c^{h}\Vert)=A<\infty$. Since $\lim_{t\rightarrow\infty}p_{t}\sum_{h\in G_{t}}\Vert c^{h}\Vert=\infty$, there is $\varepsilon\in (0,1)$ such that $p_{t}\sum_{h\in G_{t}}\Vert c^{h}\Vert>\varepsilon$, $t\geq 1$. For $\delta_{1}, N>0$ (to be chosen later), let $\{\delta_{t}\}_{t\geq1}$ be the sequence of strictly positive numbers defined by
\begin{eqnarray}\label{eqDeltaDef}
    p_{t+1}\delta_{t+1}=p_{t}\delta_{t}+\frac{N}{p_{t}\sum_{h\in G_{t}}\Vert c^{h}\Vert}(p_{t}\delta_{t})^{2},
\end{eqnarray}
for $t\geq1$. By inverting and rearranging terms, we obtain
\begin{eqnarray*}
    \frac{1}{p_{t+1}\delta_{t+1}}=\frac{1}{p_{t}\delta_{t}}-\frac{N}{p_{t}\sum_{h\in G_{t}}\Vert c^{h}\Vert+Np_{t}\delta_{t}}\geq \frac{1}{p_{t}\delta_{t}}-\frac{N}{p_{t}\sum_{h\in G_{t}}\Vert c^{h}\Vert},
\end{eqnarray*}
so that
\begin{eqnarray*}
    \frac{1}{p_{1}\delta_{1}}-\frac{1}{p_{t}\delta_{t}}=\sum_{i=1}^{t-1}\biggr(\frac{1}{p_{i}\delta_{i}}-\frac{1}{p_{i+1}\delta_{i+1}}\biggr)\leq \sum^{t-1}_{i=1}\frac{N}{p_{i}\sum_{h\in G_{i}}\Vert c^{h}\Vert}\leq AN<\infty,
\end{eqnarray*}
for $t\geq2$. Therefore, for $0<\delta_{1}< \varepsilon^{3}/p_{1}(1+\varepsilon^{3} AN)$, we have 
\begin{eqnarray*}
    \frac{1}{p_{t}\delta_{t}}\geq \frac{1}{p_{1}\delta_{1}}-AN\geq\frac{1+\varepsilon^{3} AN}{\varepsilon^{3}}-AN=\varepsilon^{-3},
\end{eqnarray*}
so that $p_{t}\delta_{t}\leq \varepsilon^{3}< \varepsilon^{2}p_{t}\sum_{h\in G_{t}}\Vert c^{h}\Vert$ and
\begin{eqnarray}\label{eqBoundDeltaT}
    \delta_{t}<\varepsilon^{2}\sum_{h\in G_{t}}\Vert c^{h}\Vert,
\end{eqnarray}
for $t\geq2$. Furthermore, we can also assume $\delta_{1}<\varepsilon^{2}\sum_{h\in G_{1}}\Vert c^{h}\Vert$, so that (\ref{eqBoundDeltaT}) becomes valid for $t\geq1$. I proceed with the following two pointwise assumptions.

\begin{assumption}[Pointwise]\label{assUniformBox}
    $\inf_{t\geq0}\sigma_{t}>0$.
\end{assumption}

\begin{assumption}[Pointwise]\label{assEndowmentGrowth}
    $\inf_{t\geq1}\sum_{h\in G_{t}}\Vert \omega^{h}\Vert/(R_{t+1}+R_{t+2})>0$.
\end{assumption}

Assumption \ref{assUniformBox} is a uniform version of Assumption \ref{assBorderPrices}: the box-sets that bound the prices faced by each generation can be taken to be the same for all generations, so that, along any equilibrium, the ratio between the prices of any two goods of adjacent periods lies in a fixed compact interval (i.e., real interest rates are uniformly bounded). Assumption \ref{assEndowmentGrowth} requires the endowment of each generation to be a nonvanishing fraction of the maximum resources of the two subsequent periods, so that maximum resources cannot have arbitrarily large growth rates. 

Let $\sigma=\inf_{t\geq0}\sigma_{t}$ and $\gamma=\inf_{t\geq1}\sum_{h\in G_{t}}\Vert \omega^{h}\Vert/(R_{t+1}+R_{t+2})$. Then, (\ref{eqLowerBoundConsumption}) and (\ref{eqUpperBoundConsumption}) imply
\begin{eqnarray*}
    \sum_{h\in G_{t+1}}\Vert c^{h}\Vert\leq R_{t+1}+R_{t+2}\leq\frac{1}{\gamma}\sum_{h\in G_{t}}\Vert \omega^{h}\Vert\leq\frac{\sqrt{2}}{\gamma\sigma^{2}}\sum_{h\in G_{t}}\Vert c^{h}\Vert,
\end{eqnarray*}
for $t\geq1$, so that we can take, without loss of generality,
\begin{eqnarray}\label{eqEndowmentGrowth}
\varepsilon<\frac{\gamma\sigma^{2}}{\gamma\sigma^{2}+\sqrt{2}}\leq\inf_{t\geq1} \frac{\sum_{h\in G_{t}}\Vert c^{h}\Vert}{\sum_{h\in G_{t}}\Vert c^{h}\Vert+\sum_{h\in G_{t+1}}\Vert c^{h}\Vert}.
\end{eqnarray}

Next, let $\alpha_{h}=\Vert c^{h}\Vert/\sum_{i\in G_{t}}\Vert c^{i}\Vert\in (0,1)$, $\delta_{h,t}=\alpha_{h}\delta_{t}$, $\delta_{h,t+1}=\alpha_{h}\delta_{t+1}$, and $\tilde{c}^{h}=(c^{h}_{t}-\delta_{h,t},c^{h}_{t+1}+\delta_{h,t+1})$, $h\in G_{t}$, $t\geq1$ (i.e., each young household gives up $\delta_{h,t}$ units of the period-$t$ good and receives $\delta_{h,t+1}$ units of the period-$(t+1)$ good, with $\sum_{h\in G_{t}}\delta_{h,t}=\delta_{t}$ and $\sum_{h\in G_{t}}\delta_{h,t+1}=\delta_{t+1}$). Notice that (\ref{eqBoundDeltaT}) and (\ref{eqEndowmentGrowth}) imply
\begin{eqnarray}\label{eqConsumptionRadius}
    \Vert \tilde{c}^{h}-c^{h}\Vert=\alpha_{h}\Vert(\delta_{t},\delta_{t+1})\Vert
    <\Vert c^{h}\Vert\varepsilon^{2}\frac{\sum_{i\in G_{t}}\Vert c^{i}\Vert+\sum_{i\in G_{t+1}}\Vert c^{i}\Vert}{\sum_{i\in G_{t}}\Vert c^{i}\Vert}
    <\varepsilon \Vert c^{h}\Vert,
\end{eqnarray}
for $h\in G_{t}$, $t\geq1$. Furthermore, (\ref{eqDeltaDef}) implies
\begin{eqnarray*}
    p_{t+1}\delta_{h,t+1}=p_{t}\delta_{h,t}+\frac{N}{\alpha_{h} p_{t}\sum_{i\in G_{t}}\Vert c^{i}\Vert}(p_{t}\delta_{h,t})^{2}=p_{t}\delta_{h,t}+\frac{N}{\Vert c^{h}\Vert p_{t}}(p_{t}\delta_{h,t})^{2},
\end{eqnarray*}
which can also be written as
\begin{eqnarray}\label{eqPositiveProjection}
    \frac{(p_{t},p_{t+1})}{\Vert(p_{t},p_{t+1})\Vert}\cdot(\tilde{c}^{h}-c^{h})=\frac{N}{\Vert c^{h}\Vert}\frac{p_{t}}{\Vert (p_{t},p_{t+1})\Vert}(\tilde{c}_{t}^{h}-c_{t}^{h})^{2}>0,
\end{eqnarray}
for $h\in G_{t}$, $t\geq1$ (i.e., $\tilde{c}^{h}$ lies strictly above the trade hyperplane $\mathcal{H}(c^{h})$). Let
\begin{eqnarray*}
    \tilde{c}^{h}-c^{h}=a_{h} \frac{(p_{t},p_{t+1})}{\Vert(p_{t},p_{t+1})\Vert}+b_{h} \frac{(-p_{t+1},p_{t})}{\Vert(p_{t},p_{t+1})\Vert},
\end{eqnarray*}
for $h\in G_{t}$, $t\geq1$, and notice that (\ref{eqPositiveProjection}) implies $a_{h}>0$, whereas $\tilde{c}^{h}_{t}-c^{h}_{t}=-\delta_{h,t}<0$ implies $b_{h}>0$. I proceed with the following local assumption.

\begin{assumption}[Local]\label{assUniformEpsilon2}
    Let $p\in\mathcal{M}$. Then, there is a sequence $\{\varepsilon^{h}_{2}\}_{h\geq1}\subset (0,1)$ such that: 
    \begin{itemize}
        \item[i.] $\pi^{h}_{i}(f^{h}(z))\leq (z_{t+1}-c^{h}_{t+1})^{2}/(\varepsilon^{h}_{2}\Vert c^{h}\Vert)$, for $z\in B(c^{h},\varepsilon^{h}_{2} \Vert c^{h}\Vert)\subseteq\mathbb{R}^{2}_{++}$, $u^{h}(z)=u^{h}(c^{h})$, $h\in G_{t}$, $t\geq1$;
        \item[ii.] $\inf\{\varepsilon^{h}_{2}\}_{h\in G_{t},\,t\geq1}>0$.
    \end{itemize} 
\end{assumption}

Assumption \ref{assUniformEpsilon2} can be seen as placing a uniform lower bound on the constants given by Theorem \ref{theoUpperBoundCass}. We can take, without loss of generality, $\varepsilon<\inf\{\varepsilon^{h}_{2}\}_{h\in G_{t},\, t\geq1}$, so that (\ref{eqConsumptionRadius}) implies $\tilde{c}^{h}\in B(c^{h},\varepsilon^{h}_{2} \Vert c^{h}\Vert)\subseteq\mathbb{R}^{2}_{++}$, $h\in G_{t}$, $t\geq1$.

Now suppose, by contradiction, that $u^{h}(\tilde{c}^{h})\leq u^{h}(c^{h})$ for some $h\in G_{t}$, $t\geq1$. Since $a_{h}>0$ and $u^{h}(\cdot)$ is strictly increasing, $u^{h}(c^{h}+a_{h}(p_{t},p_{t+1})/\Vert(p_{t},p_{t+1})\Vert)>u^{h}(c^{h})$, so that the continuity of $u^{h}(\cdot)$ implies that there is $\tau_{h}\in(0,1]$ such that
\begin{eqnarray*}
    \bar{c}^{h}=c^{h}+a_{h} \frac{(p_{t},p_{t+1})}{\Vert(p_{t},p_{t+1})\Vert}+b_{h}\tau_{h} \frac{(-p_{t+1},p_{t})}{\Vert(p_{t},p_{t+1})\Vert}
\end{eqnarray*}
satisfies $u^{h}(\bar{c}^{h})=u^{h}(c^{h})$. Since $\Vert \bar{c}^{h}-c^{h}\Vert\leq \Vert \tilde{c}^{h}-c^{h}\Vert< \varepsilon \Vert c^{h}\Vert$ by (\ref{eqConsumptionRadius}), we have $\bar{c}^{h}\in B(c^{h},\varepsilon \Vert c^{h}\Vert)\subseteq B(c^{h},\varepsilon^{h}_{2} \Vert c^{h}\Vert)$. Also, $\bar{c}^{h}\in\mathcal{I}(c^{h})$ implies $\pi^{h}_{i}(f^{h}(\bar{c}^{h}))=(p_{t},p_{t+1})\cdot(\bar{c}^{h}-c^{h})/\Vert(p_{t},p_{t+1})\Vert$ (Definition \ref{defDepthProjection} and the parametrization of $\mathcal{I}(c^{h})$ by its normal vector), $\tilde{c}^{h}-\bar{c}^{h}\perp (p_{t},p_{t+1})$, and $c^{h}_{t+1}<\bar{c}^{h}_{t+1}\leq\tilde{c}^{h}_{t+1}$ (since $a_{h},b_{h}>0$ and $\tau_{h}\in(0,1]$). Therefore, Assumption \ref{assUniformEpsilon2} and (\ref{eqPositiveProjection}) yield
\begin{eqnarray*}
    \frac{N}{\Vert c^{h}\Vert}\frac{p_{t}}{\Vert (p_{t},p_{t+1})\Vert}(\tilde{c}_{t}^{h}-c_{t}^{h})^{2}&=&\frac{(p_{t},p_{t+1})}{\Vert(p_{t},p_{t+1})\Vert}\cdot (\tilde{c}^{h}-c^{h})=\pi^{h}_{i}(f^{h}(\bar{c}^{h}))\\
    &\leq& \frac{(\bar{c}^{h}_{t+1}-c^{h}_{t+1})^{2}}{\varepsilon\Vert c^{h}\Vert}\leq\frac{(\tilde{c}^{h}_{t+1}-c^{h}_{t+1})^{2}}{\varepsilon\Vert c^{h}\Vert},
\end{eqnarray*}
that is, after canceling $\Vert c^{h}\Vert$,
\begin{eqnarray}\label{eqKBound}
    \sqrt{\frac{\varepsilon N p_{t}}{\Vert (p_{t},p_{t+1})\Vert}}\leq \frac{\tilde{c}^{h}_{t+1}-c^{h}_{t+1}}{c^{h}_{t}-\tilde{c}^{h}_{t}}.
\end{eqnarray}
I proceed with the following local assumption.

\begin{assumption}[Local]\label{assBallEpsilon}
    Let $p\in\mathcal{M}$. Then, there is a sequence $\{\varepsilon^{h}_{3}\}_{h\geq1}\subset(0,1)$ such that: 
    \begin{itemize}
        \item[i.] $u^{h}(z)>u^{h}(c^{h})$, for $z\in B(c^{h},\varepsilon^{h}_{3}\Vert c^{h}\Vert)$ with $z_{t}<c^{h}_{t}$, $z_{t+1}-c^{h}_{t+1}\geq (c^{h}_{t}-z_{t})/\varepsilon^{h}_{3}$, and $(p_{t},p_{t+1})\cdot(z-c^{h})>0$, $h\in G_{t}$, $t\geq1$;
        \item[ii.] $\inf \{\varepsilon^{h}_{3}\}_{h\in G_{t},t\geq1}>0$.
    \end{itemize}
\end{assumption}

Assumption \ref{assBallEpsilon} can be seen as placing a uniform lower bound on the constants given by Lemma \ref{lemmaInclined}. We can take, without loss of generality, $\varepsilon<\inf\{\varepsilon^{h}_{3}\}_{h\in G_{t},\, t\geq1}$, so that (\ref{eqConsumptionRadius}) implies $\tilde{c}^{h}\in B(c^{h},\varepsilon^{h}_{3}\Vert c^{h}\Vert)$ and, since $1/\varepsilon^{h}_{3}<1/\varepsilon$, the slope requirement in item i of Assumption \ref{assBallEpsilon} is implied by $\tilde{c}^{h}_{t+1}-c^{h}_{t+1}\geq(c^{h}_{t}-\tilde{c}^{h}_{t})/\varepsilon$, $h\in G_{t}$, $t\geq1$. Since $(p_{t},p_{t+1})/p_{t}\in\mathcal{B}_{t}(\sigma)$ implies $p_{t+1}/p_{t}\leq\sigma^{-1}$, we have $\Vert(p_{t},p_{t+1})\Vert/p_{t}\leq\sqrt{1+\sigma^{-2}}<\sqrt{2}/\sigma$, $t\geq1$, and Assumption \ref{assUniformBox} allows us to choose
\begin{eqnarray}\label{eqLastBound}
    N\geq \frac{\sqrt{2}}{\varepsilon^{3}\sigma}\geq\sup_{t\geq1}\frac{\Vert (p_{t},p_{t+1})\Vert}{\varepsilon^{3} p_{t}}\implies \frac{1}{\varepsilon}\leq \sqrt{\frac{\varepsilon N p_{t}}{\Vert (p_{t},p_{t+1})\Vert}},
\end{eqnarray}
for $t\geq1$. Therefore, $\tilde{c}^{h}\in B(c^{h},\varepsilon^{h}_{3}\Vert c^{h}\Vert)$, $\tilde{c}^{h}_{t}<c^{h}_{t}$, (\ref{eqPositiveProjection}), (\ref{eqKBound}) and (\ref{eqLastBound}) imply that $\tilde{c}^{h}$ satisfies all the requirements in item i of Assumption \ref{assBallEpsilon}, so that $u^{h}(\tilde{c}^{h})>u^{h}(c^{h})$, contradicting the supposition $u^{h}(\tilde{c}^{h})\leq u^{h}(c^{h})$. We conclude that $u^{h}(\tilde{c}^{h})>u^{h}(c^{h})$, $h\in G_{t}$, $t\geq1$.

Finally, for $h\in G_{0}$, let $\tilde{c}^{h}=c^{h}+\alpha_{h}\delta_{1}$, with $\alpha_{h}=\Vert c^{h}\Vert/\sum_{i\in G_{0}}\Vert c^{i}\Vert$, so that $u^{h}(\tilde{c}^{h})>u^{h}(c^{h})$, $h\in G_{0}$, and $\sum_{h\in G_{0}}(\tilde{c}^{h}-c^{h})=\delta_{1}$. Then, $(\tilde{c},y)$ is an allocation of $\mathcal{E}$ with
\begin{eqnarray*}
    \omega_{t}(\tilde{c},y)-\omega_{t}(c,y)=\sum_{h\in G_{t-1}}(\tilde{c}^{h}_{t}-c^{h}_{t})+\sum_{h\in G_{t}}(\tilde{c}^{h}_{t}-c^{h}_{t})=\delta_{t}-\delta_{t}=0,
\end{eqnarray*}
$t\geq1$, and $U^{h}(\tilde{c}^{h})>U^{h}(c^{h})$, $h\geq1$ (i.e., each period-$t$ transfer $\delta_{t}$ from the young to the old is compensated, for the young, by the transfer $\delta_{t+1}$ they receive in the next period, while generation $G_{0}$ receives $\delta_{1}$ for free). Hence, $(c,y)\notin\mathcal{P}$ and $p\notin\mathcal{M}^{PO}$.

Now consider the general case $L_{t}\geq1$, $t\geq1$. Let $p\in\mathcal{M}$ with $\sum^{\infty}_{t=1}1/(\Vert p_{t}\Vert\sum_{h\in G_{t}}\Vert c^{h}\Vert)<\infty$. Notice that, for all $t\geq1$, there is $1\leq l(t)\leq L_{t}$ such that
\begin{eqnarray}\label{eqSupNorm}
    \Vert p_{t} \Vert_{\infty}=p_{t,l(t)}\geq \frac{\Vert p_{t}\Vert}{L_{t}}.
\end{eqnarray}
Since a household holds none of the commodities he does not desire, the transfers of the two-good argument can only be made among the households that desire both $l(t)$ and $l(t+1)$. Let $D_{t}=\{h\in G_{t}\mid S^{h}(\cdot)\textrm{ selects the goods } l(t) \textrm{ and } l(t+1)\}$, $t\geq1$, be the set of such households. I proceed with the following pointwise assumption.

\begin{assumption}[Pointwise]\label{assGoodsFraction}
    Let $p\in\mathcal{M}$, and $l(\cdot)$ and $D_{t}$ be as defined above. Then: 
    \begin{itemize}
        \item[i.] $\sup_{t\geq1}L_{t}<\infty$;
        \item[ii.] the good $l(1)$ is desired by some $h\in G_{0}$;
        \item[iii.] $\inf_{t\geq1}\sum_{h\in D_{t}}\Vert(c^{h}_{t,l(t)},c^{h}_{t+1,l(t+1)})\Vert/\sum_{h\in G_{t}}\Vert c^{h}\Vert>0$.
    \end{itemize}
   
\end{assumption}
Assumption \ref{assGoodsFraction} requires a uniform bound on the number of goods available in each period and that the households that desire both goods identified by $l(\cdot)$ (i.e., the most expensive goods of each period) hold, in these goods, a nonvanishing fraction of the aggregate consumption of their generation; in particular, $D_{t}\neq\emptyset$, $t\geq1$. Notice that it does not require the goods $l(t)$ and $l(t+1)$ to be desired by every household of $G_{t}$, nor a lower bound on the number of households in $D_{t}$: what matters is the fraction of aggregate consumption held by $D_{t}$, not its size. Let $\bar{L}=\sup_{t\geq1}L_{t}$ and $\gamma^{\prime}$ be the infimum in item iii. Under Assumption \ref{assGoodsFraction}, (\ref{eqSupNorm}) allows us to write
\begin{eqnarray*}
    \sum^{\infty}_{t=1}\frac{1}{\Vert p_{t}\Vert\sum_{h\in G_{t}}\Vert c^{h}\Vert}<\infty\implies \sum^{\infty}_{t=1}\frac{1}{ p_{t,l(t)}\sum_{h\in D_{t}} \Vert (c^{h}_{t,l(t)},c^{h}_{t+1,l(t+1)})\Vert}<\infty.
\end{eqnarray*}
Then, we can define a 2-dimensional overlapping generations economy $\tilde{\mathcal{E}}$, whose generation $t$ consists of the households in $D_{t}$, with prices $(p_{t,l(t)},p_{t+1,l(t+1)})$ and consumption bundles $\tilde{c}^{h}=(c^{h}_{t,l(t)},c^{h}_{t+1,l(t+1)})\in\mathbb{R}^{2}_{++}$, through the following ``restricted'' utility functions $\tilde{u}^{h}:\mathbb{R}^{2}_{++}\rightarrow\mathbb{R}$, $h\in D_{t}$, $t\geq1$,
\begin{eqnarray*}
    \tilde{u}^{h}(z_{t},z_{t+1})=U^{h}(c^{h}_{t,1},\ldots, z_{t},\ldots, z_{t+1},\ldots,c^{h}_{t+1,L_{t+1}}),
\end{eqnarray*}
where the arguments $z_{t}$ and $z_{t+1}$ enter the $l(t)$ and $l(t+1)$ coordinates of the corresponding $L_{t}$ and $L_{t+1}$-dimensional consumption vectors.

The argument for the case $L_{t}=1$ applies verbatim to $\tilde{\mathcal{E}}$\footnote{Notice that: (a) for $h\in D_{t}$, $\tilde{u}^{h}(\cdot)$ is smooth, strictly increasing and strictly quasiconcave, since $l(t)$ and $l(t+1)$ are selected by $S^{h}(\cdot)$, and $(p_{t,l(t)},p_{t+1,l(t+1)})$ is proportional to $\nabla\tilde{u}^{h}(\tilde{c}^{h})$ by the first-order conditions of utility maximization; (b) $p_{t+1,l(t+1)}/p_{t,l(t)}\leq\sigma^{-2}$, so that $\sigma^{2}$ plays the role of $\sigma$ in (\ref{eqLastBound}); (c) $\sum_{h\in D_{t}}\Vert\tilde{c}^{h}\Vert\geq\gamma^{\prime}\sum_{h\in G_{t}}\Vert c^{h}\Vert\geq\gamma^{\prime}\sigma^{2}(2\bar{L})^{-1/2}\sum_{h\in G_{t}}\Vert\omega^{h}\Vert$ by (\ref{eqLowerBoundConsumption}), whereas $\sum_{h\in D_{t+1}}\Vert\tilde{c}^{h}\Vert\leq\sum_{h\in G_{t+1}}\Vert c^{h}\Vert\leq R_{t+1}+R_{t+2}$ by (\ref{eqUpperBoundConsumption}), so that Assumption \ref{assEndowmentGrowth} plays the same role as before (with different constants); and (d) the transfers $\delta_{t}$ are made in the good $l(t)$ and only among the households in $D_{t}$ (with weights $\alpha_{h}=\Vert\tilde{c}^{h}\Vert/\sum_{i\in D_{t}}\Vert\tilde{c}^{i}\Vert$, $h\in D_{t}$), the households in $G_{t}\setminus D_{t}$ being left untouched, so that the resulting allocation of $\mathcal{E}$ has the same implicit aggregate endowments as $(c,y)$, increases $U^{h}(\cdot)$ for every $h\in D_{t}$, $t\geq1$, and leaves $U^{h}(\cdot)$ unchanged for every other $h\geq1$ (for $h\in G_{0}$, the transfer $\delta_{1}$ is given to a household that desires the good $l(1)$).}. Therefore, if the local Assumptions \ref{assUniformEpsilon2} and \ref{assBallEpsilon} are valid for the lower dimensional economy $\tilde{\mathcal{E}}$, we arrive at the following general statement for the necessity of the Cass criterion.
\begin{theorem}\label{theoCassNecessity}
    Let $p\in\mathcal{M}$, and $l(\cdot)$, $D_{t}$ and $\tilde{\mathcal{E}}$ be as defined above. Under Assumptions \ref{assUniformBox}, \ref{assEndowmentGrowth} and \ref{assGoodsFraction} for $\mathcal{E}$, and Assumptions \ref{assUniformEpsilon2} and \ref{assBallEpsilon} for $\tilde{\mathcal{E}}$, if $\sum^{\infty}_{t=1}1/(\Vert p_{t}\Vert \sum_{h\in G_{t}}\Vert c^{h}\Vert)<\infty$, then $p\notin \mathcal{M}^{PO}$.
\end{theorem}

\subsection{Sufficiency of the Cass criterion}\label{subsecSuff}

I now turn to the sufficiency of the Cass criterion, the converse of Theorem \ref{theoCassNecessity}. Let $p\in\mathcal{M}$ and $(c,y)=(c(p),y(p))$. For $j\in G^{f}_{t}$, let $\tilde{\mathcal{T}}^{j}=\{z\in\tilde{Y}^{j}\mid F^{j}(S^{j}(z))=0\}$ be the transformation frontier of firm $j$ (embedded in $\mathbb{R}^{L_{t}+L_{t+1}}$), and let $\pi^{j}_{f}(\cdot)$ and $f^{j}_{f}(\cdot)$ denote the depth projection of the transformation frontier of $F^{j}(\cdot)$ at $S^{j}(y^{j})$ and the marginal rates of transformation function of $F^{j}(\cdot)$, as defined in \hyperref[sec3]{Section 3}.

\begin{sloppypar}
Notice that $y^{j}\in\tilde{\mathcal{T}}^{j}$ and, by the first-order conditions, $S^{j}(p_{t},p_{t+1})/\Vert S^{j}(p_{t},p_{t+1})\Vert$ is the unit normal vector to the transformation frontier at $S^{j}(y^{j})$, exactly as $S^{h}(p_{t},p_{t+1})/\Vert S^{h}(p_{t},p_{t+1})\Vert$ is the unit normal vector to $\mathcal{I}(S^{h}(c^{h}))$ at $S^{h}(c^{h})$, $h\in G_{t}$.
\end{sloppypar}

An allocation $(\tilde{c},\tilde{y})$ is called \textit{non-wasteful} if $\tilde{c}^{h}$ is zero in the coordinates not selected by $S^{h}(\cdot)$, $h\geq1$ (i.e., every household only holds his desired commodities); in particular, $(c,y)$ is non-wasteful. I start with the following two local assumptions.

\begin{assumption}[Local]\label{assParetoDominated}
    Let $p\in\mathcal{M}$. If $p\notin\mathcal{M}^{PO}$, then for all $\varepsilon>0$, there is a non-wasteful allocation $(\tilde{c},\tilde{y})$ such that: 
    \begin{itemize}
        \item[i.] $\tilde{c}^{h}\in B(c^{h},\varepsilon \Vert c^{h}\Vert)$, $h\in G_{t}$, and $\tilde{y}^{j}\in B(y^{j},\varepsilon\Vert y^{j}\Vert)\cap\tilde{\mathcal{T}}^{j}$, $j\in G^{f}_{t}$, $t\geq1$; 
        \item[ii.] $U^{h^{\prime}}(\tilde{c}^{h^{\prime}})>U^{h^{\prime}}(c^{h^{\prime}})$, for some $h^{\prime}\in G_{0}$, and $U^{h}(\tilde{c}^{h})=U^{h}(c^{h})$, for all $h\neq h^{\prime}$;
        \item[iii.] $\omega_{t}(\tilde{c},\tilde{y})=\omega_{t}(c,y)$, $t\geq1$.
    \end{itemize}
\end{assumption}

Assumption \ref{assParetoDominated} implies that every inefficient equilibrium allocation can be improved by ``bringing consumption from the far future'' to some household from the older generation in $t=1$, while leaving all other households indifferent and all firms on their transformation frontiers, through relatively small shifts in consumption bundles and production plans.

The requirements that $(\tilde{c},\tilde{y})$ be non-wasteful and that $\tilde{y}^{j}\in\tilde{\mathcal{T}}^{j}$ entail no loss of generality: commodities held by households that do not desire them, as well as the additional output that a firm with $F^{j}(S^{j}(\tilde{y}^{j}))<0$ could produce, can be transferred to households that desire them (item ii of Assumption \ref{assResourceRelated}), which preserves the implicit aggregate endowments and weakly improves everyone.

\begin{assumption}[Local]\label{assUniformEpsilon}
    Let $p\in\mathcal{M}$. Then, there is a sequence $\{\varepsilon^{h}_{1}\}_{h\geq1}\cup\{\varepsilon^{j}_{1}\}_{j\geq1}\subset(0,1)$ such that: 
    \begin{itemize}
        \item[i.] $\pi^{h}_{i}(f^{h}(z))\geq \varepsilon^{h}_{1}\Vert z-S^{h}(c^{h})\Vert^{2}/\Vert S^{h}(c^{h})\Vert$, for $z\in B(S^{h}(c^{h}),\varepsilon^{h}_{1} \Vert S^{h}(c^{h})\Vert)\subseteq\mathbb{R}^{L_{h}}_{++}$, $u^{h}(z)=u^{h}(S^{h}(c^{h}))$, $h\in G_{t}$, $t\geq1$;
        \item[ii.] $\pi^{j}_{f}(f^{j}_{f}(z))\geq\varepsilon^{j}_{1}\Vert z-S^{j}(y^{j})\Vert^{2}/\Vert S^{j}(y^{j})\Vert$, for $z\in B(S^{j}(y^{j}),\varepsilon^{j}_{1}\Vert S^{j}(y^{j})\Vert)$, $F^{j}(z)=0$, $j\in G^{f}_{t}$ with $y^{j}\neq0$, $t\geq1$;
        \item[iii.] $\inf\{\varepsilon^{h}_{1}\}_{h\in G_{t},\, t\geq1}\cup\{\varepsilon^{j}_{1}\}_{j\in G^{f}_{t},\, t\geq1}>0$.
    \end{itemize}
\end{assumption}

Assumption \ref{assUniformEpsilon} can be seen as placing a uniform lower bound on the constants given by Theorems \ref{theoFundamentalIneq} and \ref{theoFundamentalIneqProduction}.

I follow, once again, the reasoning of \textcite{Benveniste1976} to reach a general statement for the sufficiency of the Cass criterion. If $p\notin\mathcal{M}^{PO}$, Assumptions \ref{assParetoDominated}--\ref{assUniformEpsilon} allow us to take $0<\varepsilon<\inf\{\varepsilon^{h}_{1}\}_{h\in G_{t}, t\geq1}\cup\{\varepsilon^{j}_{1}\}_{j\in G^{f}_{t},t\geq1}$ and find a non-wasteful allocation $(\tilde{c},\tilde{y})$ such that: (i) $\tilde{c}^{h}\in B(c^{h},\varepsilon \Vert c^{h}\Vert)\subseteq B(c^{h},\varepsilon^{h}_{1} \Vert c^{h}\Vert)$, $h\in G_{t}$, and $\tilde{y}^{j}\in B(y^{j},\varepsilon\Vert y^{j}\Vert)\cap\tilde{\mathcal{T}}^{j}\subseteq B(y^{j},\varepsilon^{j}_{1}\Vert y^{j}\Vert)\cap\tilde{\mathcal{T}}^{j}$, $j\in G^{f}_{t}$, $t\geq1$; (ii) $U^{h^{\prime}}(\tilde{c}^{h^{\prime}})>U^{h^{\prime}}(c^{h^{\prime}})$, and $U^{h}(\tilde{c}^{h})=U^{h}(c^{h})$, $h\neq h^{\prime}$; and (iii) $\omega_{t}(\tilde{c},\tilde{y})=\omega_{t}(c,y)$, $t\geq1$.

Since $(\tilde{c},\tilde{y})$ is non-wasteful and $(c,y)$ only holds the selected coordinates, $\Vert S^{h}(\tilde{c}^{h})-S^{h}(c^{h})\Vert=\Vert\tilde{c}^{h}-c^{h}\Vert$ and $\Vert S^{h}(c^{h})\Vert=\Vert c^{h}\Vert$, $h\in G_{t}$, and, likewise, $\Vert S^{j}(\tilde{y}^{j})-S^{j}(y^{j})\Vert=\Vert\tilde{y}^{j}-y^{j}\Vert$ and $\Vert S^{j}(y^{j})\Vert=\Vert y^{j}\Vert$, $j\in G^{f}_{t}$, $t\geq1$. Then, Assumption \ref{assUniformEpsilon} allows us to write
\begin{eqnarray}\label{eqIneqIndiff}
    \pi^{h}_{i}(f^{h}(S^{h}(\tilde{c}^{h})))\geq \varepsilon\frac{\Vert \tilde{c}^{h}-c^{h}\Vert^{2}}{\Vert c^{h}\Vert}\quad\textrm{and}\quad\pi^{j}_{f}(f^{j}_{f}(S^{j}(\tilde{y}^{j})))\geq\varepsilon\frac{\Vert\tilde{y}^{j}-y^{j}\Vert^{2}}{\Vert y^{j}\Vert},
\end{eqnarray}
for $h\in G_{t}$ and $j\in G^{f}_{t}$ with $y^{j}\neq0$, $t\geq1$ (for $y^{j}=0$, we have $\tilde{y}^{j}=y^{j}$, and the corresponding term is understood to be zero in what follows).

Let
\begin{eqnarray*}
    \delta^{t}_{1}&=&\sum_{h\in G_{t}}p_{t}\cdot(\tilde{c}^{h}_{t}-c^{h}_{t})-\sum_{j\in G^{f}_{t}}p_{t}\cdot(\tilde{y}^{j}_{t}-y^{j}_{t}),\\
    \delta^{t}_{2}&=&\sum_{h\in G_{t-1}}p_{t}\cdot(\tilde{c}^{h}_{t}-c^{h}_{t})-\sum_{j\in G^{f}_{t-1}}p_{t}\cdot(\tilde{y}^{j}_{t}-y^{j}_{t}),
\end{eqnarray*}
$t\geq1$, be the values, at the equilibrium prices, of the net changes in the period-$t$ excess demands of the young and of the old generations (households and the firms they own), respectively. Since $\omega_{t}(\tilde{c},\tilde{y})=\omega_{t}(c,y)$, we have $\delta^{t}_{1}+\delta^{t}_{2}=0$, $t\geq1$. Furthermore, for $h\in G_{t}$, $t\geq1$, $U^{h}(\tilde{c}^{h})=U^{h}(c^{h})$ and the non-wastefulness of $(\tilde{c},\tilde{y})$ imply $S^{h}(\tilde{c}^{h})\in\mathcal{I}(S^{h}(c^{h}))$, so that Definition \ref{defDepthProjection} and the parametrization of $\mathcal{I}(S^{h}(c^{h}))$ by its normal vector give
\begin{eqnarray}\label{eqSelectedProjectionHousehold}
    (p_{t},p_{t+1})\cdot(\tilde{c}^{h}-c^{h})&=&S^{h}(p_{t},p_{t+1})\cdot(S^{h}(\tilde{c}^{h})-S^{h}(c^{h}))\nonumber\\
    &=&\Vert S^{h}(p_{t},p_{t+1})\Vert\, \pi^{h}_{i}(f^{h}(S^{h}(\tilde{c}^{h})))\geq0,
\end{eqnarray}
whereas, for $j\in G^{f}_{t}$, $t\geq1$, $\tilde{y}^{j}\in\tilde{\mathcal{T}}^{j}$, Definition \ref{defDepthProjectionProduction} and Proposition \ref{propFrontierParametrization} give
\begin{eqnarray}\label{eqSelectedProjectionFirm}
    (p_{t},p_{t+1})\cdot(y^{j}-\tilde{y}^{j})&=&S^{j}(p_{t},p_{t+1})\cdot(S^{j}(y^{j})-S^{j}(\tilde{y}^{j}))\nonumber\\
    &=&\Vert S^{j}(p_{t},p_{t+1})\Vert\, \pi^{j}_{f}(f^{j}_{f}(S^{j}(\tilde{y}^{j})))\geq0
\end{eqnarray}
(i.e., the shift of a household along his indifference hypersurface costs him a nonnegative amount at the equilibrium prices, and the shift of a firm along its transformation frontier reduces its profits). Therefore,
\begin{eqnarray}\label{eqSelectedProjection}
    \delta^{t}_{1}+\delta^{t+1}_{2}=\sum_{h\in G_{t}}(p_{t},p_{t+1})\cdot(\tilde{c}^{h}-c^{h})+\sum_{j\in G^{f}_{t}}(p_{t},p_{t+1})\cdot(y^{j}-\tilde{y}^{j})\geq0,
\end{eqnarray}
for $t\geq1$, and $\delta^{t+1}_{2}\geq-\delta^{t}_{1}=\delta^{t}_{2}$ and $\delta^{t+1}_{1}\leq\delta^{t}_{1}$, $t\geq1$. In period $t=1$, since households from $G_{0}$ live only in $t=1$, $c^{h^{\prime}}$ maximizes $U^{h^{\prime}}(\cdot)$ on the budget set of $h^{\prime}$ and $U^{h^{\prime}}(\tilde{c}^{h^{\prime}})>U^{h^{\prime}}(c^{h^{\prime}})$, we have $p_{1}\cdot\tilde{c}^{h^{\prime}}>w_{h^{\prime}}(p_{1})=p_{1}\cdot c^{h^{\prime}}$, whereas $p_{1}\cdot(\tilde{c}^{h}-c^{h})\geq0$, $h\in G_{0}$, $h\neq h^{\prime}$, and $p_{1}\cdot(y^{j}-\tilde{y}^{j})\geq0$, $j\in G^{f}_{0}$, by the same reasoning as in (\ref{eqSelectedProjectionHousehold}) and by profit maximization. Hence,
\begin{eqnarray*}
    \delta^{1}_{2}=\sum_{h\in G_{0}}p_{1}\cdot(\tilde{c}^{h}-c^{h})+\sum_{j\in G^{f}_{0}}p_{1}\cdot(y^{j}-\tilde{y}^{j})>0.
\end{eqnarray*}
We conclude that $\{\delta^{t}_{2}\}_{t\geq1}$ is a nondecreasing sequence of positive numbers and $\{\delta^{t}_{1}\}_{t\geq1}$ is a nonincreasing sequence of negative numbers.

Next, notice that Assumption \ref{assUniformBox} and item i of Assumption \ref{assGoodsFraction} provide a uniform lower bound for the norms of the selected prices in (\ref{eqSelectedProjectionHousehold})--(\ref{eqSelectedProjectionFirm}) relative to $\Vert p_{t}\Vert$. Indeed, since $(p_{t},p_{t+1})/p_{t1}\in\mathcal{B}_{t}(\sigma)$, every coordinate of $(p_{t},p_{t+1})$ is at least $\sigma p_{t1}$ and $\Vert p_{t}\Vert\leq\sqrt{L_{t}}\max_{i}p_{ti}\leq\sqrt{\bar{L}}p_{t1}/\sigma$, so that
\begin{eqnarray}\label{eqNu1}
    \frac{\Vert S^{h}(p_{t},p_{t+1})\Vert}{\Vert p_{t}\Vert}\geq\frac{\sigma^{2}}{\sqrt{\bar{L}}}=\nu_{1}\quad\textrm{and}\quad\frac{\Vert S^{j}(p_{t},p_{t+1})\Vert}{\Vert p_{t}\Vert}\geq\nu_{1},
\end{eqnarray}
for $h\in G_{t}$, $j\in G^{f}_{t}$, $t\geq1$. I proceed with the following pointwise assumption.

\begin{assumption}[Pointwise]\label{assProductionScale}
    Let $p\in\mathcal{M}$. Then, $\sup_{t\geq1}\sum_{j\in G^{f}_{t}}\Vert y^{j}\Vert/\sum_{h\in G_{t}}\Vert c^{h}\Vert<\infty$.
\end{assumption}

Assumption \ref{assProductionScale} requires the scale of the production plans of each generation to be comparable to that of its consumption bundles (it holds trivially in a consumption-loan economy, with $H^{f}_{t}=0$). Let $\nu_{2}=\sup_{t\geq1}\sum_{j\in G^{f}_{t}}\Vert y^{j}\Vert/\sum_{h\in G_{t}}\Vert c^{h}\Vert$ and $V_{t}=\Vert p_{t}\Vert\sum_{h\in G_{t}}\Vert c^{h}\Vert$, $t\geq1$, so that the Cass criterion reads $\sum^{\infty}_{t=1}1/V_{t}=\infty$. Then, (\ref{eqIneqIndiff})--(\ref{eqSelectedProjection}), (\ref{eqNu1}) and Assumption \ref{assProductionScale} imply
\begin{eqnarray}\label{eqInvertible}
   -\vert \delta^{t}_{1}\vert+\vert\delta^{t+1}_{1}\vert=\delta^{t}_{1}-\delta^{t+1}_{1}&\geq&\varepsilon\nu_{1}\Vert p_{t}\Vert\biggr(\sum_{h\in G_{t}} \frac{\Vert \tilde{c}^{h}-c^{h}\Vert^{2}}{\Vert c^{h}\Vert}+\sum_{j\in G^{f}_{t}}\frac{\Vert\tilde{y}^{j}-y^{j}\Vert^{2}}{\Vert y^{j}\Vert}\biggr)\nonumber\\
   &\geq&\frac{\varepsilon\nu_{1}}{\Vert p_{t}\Vert}\biggr(\sum_{h\in G_{t}}\frac{(p_{t}\cdot(\tilde{c}^{h}_{t}-c^{h}_{t}))^{2}}{\Vert c^{h}\Vert}+\sum_{j\in G^{f}_{t}}\frac{(p_{t}\cdot(\tilde{y}^{j}_{t}-y^{j}_{t}))^{2}}{\Vert y^{j}\Vert}\biggr)\nonumber\\
   &\geq&\frac{\varepsilon\nu_{1}}{\Vert p_{t}\Vert}\frac{(\delta^{t}_{1})^{2}}{\sum_{h\in G_{t}}\Vert c^{h}\Vert+\sum_{j\in G^{f}_{t}}\Vert y^{j}\Vert}\geq \frac{\nu_{1} \varepsilon}{(1+\nu_{2})V_{t}}(\delta^{t}_{1})^{2},
\end{eqnarray}
for $t\geq1$, where the second inequality is due to $\Vert\tilde{c}^{h}-c^{h}\Vert\geq\Vert\tilde{c}^{h}_{t}-c^{h}_{t}\Vert\geq\vert p_{t}\cdot(\tilde{c}^{h}_{t}-c^{h}_{t})\vert/\Vert p_{t}\Vert$ (and likewise for firms), and the third one to the weighted Cauchy--Schwarz inequality $(\sum^{n}_{i=1}z_{i})^{2}\leq(\sum^{n}_{i=1}\varpi_{i})(\sum^{n}_{i=1}z^{2}_{i}/\varpi_{i})$, $z\in\mathbb{R}^{n}$, $\varpi\in\mathbb{R}^{n}_{++}$, applied to the terms of $\delta^{t}_{1}$ with weights $\Vert c^{h}\Vert$ and $\Vert y^{j}\Vert$. Notice that the normalization in Assumption \ref{assUniformEpsilon} is exactly what makes the denominator of the Cass criterion appear in (\ref{eqInvertible}), with no reference to the number of households or firms. Then, rearranging and taking the inverse of the terms in (\ref{eqInvertible}), we obtain
\begin{eqnarray*}
    \frac{1}{\vert \delta^{t+1}_{1}\vert}\leq \frac{(1+\nu_{2})V_{t}}{\vert\delta^{t}_{1}\vert((1+\nu_{2})V_{t}+\nu_{1}\varepsilon\vert\delta^{t}_{1}\vert)}=\frac{1}{\vert\delta^{t}_{1}\vert}-\frac{\nu_{1}\varepsilon}{(1+\nu_{2})V_{t}+\nu_{1}\varepsilon\vert \delta^{t}_{1}\vert},
\end{eqnarray*}
for $t\geq1$, and therefore,
\begin{eqnarray}\label{eqConvergence}
    \sum^{\infty}_{t=1}\frac{\nu_{1}\varepsilon}{(1+\nu_{2})V_{t}+\nu_{1}\varepsilon\vert \delta^{t}_{1}\vert}\leq \lim_{T\rightarrow\infty} \sum^{T}_{t=1}\biggr(\frac{1}{\vert \delta^{t}_{1}\vert}-\frac{1}{\vert\delta^{t+1}_{1}\vert}\biggr)=\frac{1}{\vert \delta^{1}_{1}\vert}-\lim_{T\rightarrow\infty}\frac{1}{\vert\delta^{T+1}_{1}\vert}<\infty.
\end{eqnarray}
Next, item i of Assumption \ref{assParetoDominated} gives
\begin{eqnarray}\label{eqBoundsRatio}
    \frac{\vert \delta^{t}_{1}\vert}{\Vert p_{t}\Vert \sum_{h\in G_{t}}\Vert c^{h}\Vert}&\leq&\frac{1}{\sum_{h\in G_{t}}\Vert c^{h}\Vert}\biggr(\sum_{h\in G_{t}}\Vert \tilde{c}^{h}_{t}-c^{h}_{t}\Vert+\sum_{j\in G^{f}_{t}}\Vert\tilde{y}^{j}_{t}-y^{j}_{t}\Vert\biggr)\nonumber\\
    &\leq&\frac{1}{\sum_{h\in G_{t}}\Vert c^{h}\Vert}\biggr(\varepsilon\sum_{h\in G_{t}}\Vert c^{h}\Vert+\varepsilon\sum_{j\in G^{f}_{t}}\Vert y^{j}\Vert\biggr)\leq\varepsilon(1+\nu_{2}),
\end{eqnarray}
for $t\geq1$, that is, $\vert\delta^{t}_{1}\vert\leq\varepsilon(1+\nu_{2})V_{t}$. Therefore, (\ref{eqConvergence}) and (\ref{eqBoundsRatio}) allow us to write
\begin{eqnarray*}
    \sum^{\infty}_{t=1}\frac{\nu_{1}\varepsilon}{(1+\nu_{2})V_{t}+\nu_{1}\varepsilon\vert \delta^{t}_{1}\vert}\geq\frac{\nu_{1}\varepsilon}{(1+\nu_{2})(1+\nu_{1}\varepsilon^{2})}\sum^{\infty}_{t=1}\frac{1}{V_{t}},
\end{eqnarray*}
and, therefore, $\sum^{\infty}_{t=1}1/V_{t}=\sum^{\infty}_{t=1}1/(\Vert p_{t}\Vert \sum_{h\in G_{t}}\Vert c^{h}\Vert)<\infty$. We state this as the following theorem.

\begin{theorem}\label{theoCassSufficiency}
    Let $p\in\mathcal{M}$. Under Assumption \ref{assUniformBox}, item i of Assumption \ref{assGoodsFraction}, and Assumptions \ref{assParetoDominated}--\ref{assProductionScale}, if $\sum^{\infty}_{t=1}1/(\Vert p_{t}\Vert \sum_{h\in G_{t}}\Vert c^{h}\Vert)=\infty$, then $p\in \mathcal{M}^{PO}$.
\end{theorem}

Notice that production enters the proof of Theorem \ref{theoCassSufficiency} only through the profit losses in (\ref{eqSelectedProjectionFirm}), which are governed by the depth projections of the transformation frontiers exactly as the costs of the households' shifts in (\ref{eqSelectedProjectionHousehold}) are governed by the depth projections of the indifference hypersurfaces. The recursion (\ref{eqInvertible}) then aggregates both kinds of quadratic terms, which is why Assumption \ref{assUniformEpsilon} bounds the depth projections of households and firms alike.

One can build examples in which the Cass criterion fails when we drop assumptions of Theorem \ref{theoCassSufficiency}. The following one is due to \textcite[p. 805]{OkunoZilcha_1980}.

\begin{example}
Let $L_{t}=1$ and $H_{t}=1$, $t\geq0$, with no firms, $u^{0}(c_{1})=\ln c_{1}$, $u^{1}(c_{1},c_{2})= \ln c_{1}+\ln c_{2}$ and $u^{t}(c_{t},c_{t+1})=c_{t}^{1-t^{-2}}+c_{t+1}^{1-t^{-2}}$, $t\geq2$, and endowments $\omega^{0}=1$ and $\omega^{t}=(3,1)$, $t\geq1$. Since $(p_{t},p_{t+1})$ is proportional to $\nabla u^{t}(3,1)$ for $p_{1}=1$ and $p_{t+1}=3^{t^{-2}} p_{t}$, $t\geq1$, this $p$ is an (autarkic) equilibrium, $p\in\mathcal{M}$, with $c^{h}(p)=\omega^{h}$, $h\geq0$. Therefore, $p_{t}=3^{\sum^{t-1}_{n=1}n^{-2}}$, $t\geq2$, and, since $\lim_{t\rightarrow\infty}p_{t}=3^{\pi^{2}/6}<\infty$, we have $\sum^{\infty}_{t=1} 1/p_{t}=\infty$. This allocation, however, is Pareto dominated by $\tilde{c}^{0}=2$ and $\tilde{c}^{t}=(2,2)$, $t\geq1$, which has the same implicit aggregate endowments, so that $p\notin\mathcal{M}^{PO}$.

After some careful thought, one may notice that Assumptions \ref{assUniformBox}, \ref{assParetoDominated} and \ref{assProductionScale} are satisfied (for the first one, notice that $p_{t+1}/p_{t}\in[1,3]$ and that the demand of each generation exceeds the total resources of the economy whenever $p_{t+1}/p_{t}$ is far enough from these values). However, Assumption \ref{assUniformEpsilon} is not. This is because the indifference curves become flatter around $(3,1)$ as $t$ grows, since $\lim_{t\rightarrow\infty} u^{t}(z_{1},z_{2})=z_{1}+z_{2}$, $(z_{1},z_{2})\in\mathbb{R}^{2}_{+}$. Therefore, the indifference hypersurfaces are approaching the trade hyperplanes (see Figure \ref{FigCanonical}), and this implies that the depth projections get uniformly closer to zero, as $t$ grows, in any fixed neighborhood of $(3,1)$, so Assumption \ref{assUniformEpsilon} cannot be valid.
\end{example}

%--- Section ---%
\section{Concluding remarks}\label{sec5}

This paper introduces depth projections as a tool for consumer and producer theory to measure the local behavior of the indifference and offer hypersurfaces and of the transformation frontier. Up to second order, these projections are governed by the Hessians of the expenditure and profit functions, the depth projection of the offer hypersurface is always exactly twice that of the indifference hypersurface, and both are locally bounded by quadratic functions of the distance to the reference bundle or production plan.

The paper then shows how they are related to the efficiency of equilibria in overlapping generations economies with production, for which an equilibrium existence theorem is provided, by deriving general statements for the necessity and sufficiency of the Cass criterion in the form $\sum^{\infty}_{t=1}1/(\Vert p_{t}\Vert\sum_{h\in G_{t}}\Vert c^{h}\Vert)=\infty$, which is the same for consumption-loan and production economies and allows unbounded dynamics for the demography and for per capita endowments.

The proofs of these statements reveal the local nature of the Cass criterion: apart from pointwise assumptions on the equilibrium price sequence and on the primitives of the economy, only uniform and scale-invariant bounds on the depth projections in a neighborhood of each household's consumption bundle and of each firm's production plan are needed, so that no global information about the indifference hypersurfaces or the transformation frontiers is required. This is what reconciles the Cass criterion with the relabeling algorithm of \textcite{BalaskoCassShell_1980}, which sets the Gaussian curvature of the indifference hypersurfaces to zero and, therefore, rules out the uniform curvature conditions of the classical characterizations.

Furthermore, these depth projections serve as an auxiliary and tractable method for measuring the curvature of the canonical manifolds of consumer and producer theory, a method that this paper has shown to be possibly more convenient for economic applications than the traditional Gaussian curvature.

%-------------------------------------------
% Appendix
%-------------------------------------------
\appendix
\section*{Appendix}\label{appx}

All proofs are stated in this \hyperref[appx]{Appendix}.

\begin{proof}[Proof of Lemma \ref{lemmaBasicIdentities}]
Items i, ii, iii and v are given by Lemma 1 and Proposition 2 in \textcite{Dognini_2025b} and I rewrite their proofs here for convenience. Let $p\in\mathbb{R}^{L}_{++}$. Since $x_{n}(p)\in\mathbb{R}^{L}_{++}$ is an interior solution of (\ref{UMP}) with $w=1$, the first-order condition reads $\nabla u(x_{n}(p))=\lambda_{n}(p)p$, which is the first part of item i. Walras' law implies $p\cdot x_{n}(p)=1$, for $p\in\mathbb{R}^{L}_{++}$, and differentiating it yields $p\mathbf{J}x_{n}(p)+x_{n}(p)=0$. The first-order condition implies $x_{n}(p)\cdot\nabla u(x_{n}(p))=\lambda_{n}(p)p\cdot x_{n}(p)=\lambda_{n}(p)$, which concludes item ii. The definition of $v_{n}(\cdot)$ implies $\nabla v_{n}(p)=\nabla u(x_{n}(p))\mathbf{J}x_{n}(p)=\lambda_{n}(p)p\mathbf{J}x_{n}(p)=-\lambda_{n}(p)x_{n}(p)$ and, therefore, $p\cdot\nabla v_{n}(p)=-\lambda_{n}(p)p\cdot x_{n}(p) =-\lambda_{n}(p)$, which is item iii.

For the second part of item i, let $c\in\mathbb{R}^{L}_{++}$ and $p=\nabla u(c)/(c\cdot\nabla u(c))\in\mathbb{R}^{L}_{++}$. Then, $p\cdot c=1$ and $\nabla u(c)=(c\cdot\nabla u(c))p$, so that $c$ satisfies the first-order conditions of (\ref{UMP}) with $w=1$, which are sufficient by the strict quasiconcavity of $u(\cdot)$. Therefore, $x_{n}(p)=c$ and $x_{n}(\cdot)$ is onto $\mathbb{R}^{L}_{++}$. Conversely, if $x_{n}(p)=c$, then the first part of item i and $p\cdot c=1$ imply $p=\nabla u(c)/(c\cdot\nabla u(c))$, so that $x_{n}(\cdot)$ is one-to-one with inverse $x^{-1}_{n}(c)=\nabla u(c)/(c\cdot\nabla u(c))$. Since $x_{n}(\cdot)$ and $x_{n}^{-1}(\cdot)$ are smooth, $x_{n}(\cdot)$ is a diffeomorphism.

Next, let $(p,u)\in\mathbb{R}^{L+1}_{++}$. Since $u(\cdot)$ is strictly increasing and $u(0)=0$, then $x_{n}(p)=h(p,v_{n}(p))$ and $h(p,u)=x_{n}(p/e(p,u))$ \parencite[p. 58, Proposition 3.E.1]{Mas-ColellWhinstonGreen_1995}. Furthermore, $e(p,v_{n}(p))=p\cdot h(p,v_{n}(p))=p\cdot x_{n}(p)=1$, which concludes item v.

Regarding smoothness, since $u(\cdot)$ is strictly increasing and $h(p,u)\in\mathbb{R}^{L}_{++}$, $e(\cdot,u)$ is differentiable with $\nabla_{1}e(p,u)=h(p,u)$ (Shephard's lemma; \textcite[p. 68, Proposition 3.G.1]{Mas-ColellWhinstonGreen_1995}), which is the first part of item iv. In particular, $e(\cdot,u)$ is continuously differentiable, since $h(\cdot,u)$ is continuous. Then, $h(p,u)=x_{n}(p/e(p,u))$ and the smoothness of $x_{n}(\cdot)$ imply that $h(\cdot,u)$ is continuously differentiable, so that $e(\cdot,u)$ is twice continuously differentiable, and, proceeding inductively, $h(\cdot,u)$ and $e(\cdot,u)$ are smooth in $p$.

As for $u$, for a given $p\in\mathbb{R}^{L}_{++}$, $e(p,\cdot)$ is the inverse of the map $w\rightarrow v(p,w)=u(x(p,w))$, which is smooth with $\partial v(p,w)/\partial w=\lambda(p,w)>0$ (the envelope theorem; \textcite[p. 58, Proposition 3.D.3]{Mas-ColellWhinstonGreen_1995}), so that $e(p,\cdot)$ is smooth by the inverse function theorem, and so is $h(p,\cdot)=x_{n}(p/e(p,\cdot))$.

Since $h(\cdot,u)$ is homogeneous of degree zero in $p$, Euler's formula yields $\mathbf{H}_{11}e(p,u)p^{T}=\mathbf{J}_{1}h(p,u)p^{T}=0$ and, since $\mathbf{H}_{11}e(p,u)$ is symmetric, $p\mathbf{H}_{11}e(p,u)=0$, which concludes item iv.

For item vi, notice that the duality between the expenditure and the indirect utility functions implies $e(p,v(p,w))=w$, $w>0$, and, as noticed above, $\partial v(p,w)/\partial w=\lambda(p,w)$, $w>0$. Differentiating the former identity with respect to $w$, we obtain
\begin{eqnarray*}
    \frac{\partial e(p,v(p,w))}{\partial u}\frac{\partial v(p,w)}{\partial w}=\lambda(p,w)\frac{\partial e(p,v(p,w))}{\partial u}=1,
\end{eqnarray*}
for $w>0$. Finally, taking $w=1$ yields $\lambda_{n}(p)\partial e(p,v_{n}(p))/\partial u=1$.
\end{proof}

\begin{proof}[Proof of Theorem \ref{theoAwesome}]
Lemma \ref{lemmaBasicIdentities} implies $h(p,v_{n}(p))=x_{n}(p)$ and, therefore, $\eta(p)=0$. Also,
\begin{eqnarray}\label{eqAuxA1}
    \nabla \eta(q)=p\mathbf{J}_{1}h(q,v_{n}(p))=p\mathbf{J}_{1}\nabla_{1}e(q,v_{n}(p))=p\mathbf{H}_{11}e(q,v_{n}(p)),
\end{eqnarray}
where the second-to-last equality is due to Lemma \ref{lemmaBasicIdentities}. Then, Lemma \ref{lemmaBasicIdentities} allows us to conclude that $\nabla \eta(p)=p\mathbf{H}_{11}e(p,v_{n}(p))=0$. Since $q\mathbf{H}_{11}e(q,v_{n}(p))=0$, $q\in\mathbb{R}^{L}_{++}$ (Lemma \ref{lemmaBasicIdentities} with $u=v_{n}(p)$), we can differentiate this expression to obtain
\begin{eqnarray*}
    \mathbf{J}(q\mathbf{H}_{11}e(q,v_{n}(p)))=\mathbf{H}_{11}e(q,v_{n}(p))+\sum^{L}_{i=1}q_{i}\mathbf{H}_{11}\partial_{i}e(q,v_{n}(p))=0,
\end{eqnarray*}
with the operator $\partial_{i}$, $1\leq i\leq L$, denoting the partial derivative relative to the $i$-th price coordinate. Therefore, taking $q=p$, we obtain
\begin{eqnarray}\label{eqHessianExpIdentity}
    \sum^{L}_{i=1}p_{i}\mathbf{H}_{11}\partial_{i}e(p,v_{n}(p))=-\mathbf{H}_{11}e(p,v_{n}(p)).
\end{eqnarray}
Then, differentiating (\ref{eqAuxA1}),
\begin{eqnarray*}
    \mathbf{H}\eta(q)=\mathbf{J}\nabla \eta(q)=\mathbf{J}\biggr(\sum^{L}_{i=1}p_{i}\nabla_{1}\partial_{i}e(q,v_{n}(p))\biggr)=\sum^{L}_{i=1}p_{i}\mathbf{H}_{11}\partial_{i}e(q,v_{n}(p)),
\end{eqnarray*}
for $q\in\mathbb{R}^{L}_{++}$. We conclude through (\ref{eqHessianExpIdentity}) that $\mathbf{H}\eta(p)=-\mathbf{H}_{11}e(p,v_{n}(p))$.

Next, Lemma \ref{lemmaBasicIdentities} implies $p\cdot x_{n}(p)=1$, so that $\zeta(p)=0$. It is convenient to define $m:\mathbb{R}^{L}_{++}\rightarrow\mathbb{R}^{L}_{++}$ as $m(q)=q/(q\cdot x_{n}(p))$, $q\in\mathbb{R}^{L}_{++}$, so that $\zeta(q)=p\cdot(x_{n}(m(q))-x_{n}(p))$. Lemma \ref{lemmaBasicIdentities} implies $m(p)=p$ and we have
\begin{eqnarray}\label{eqAuxA2}
    \mathbf{J}m(q)=\frac{1}{q\cdot x_{n}(p)}\biggr(I - m(q)^{T}x_{n}(p)\biggr),
\end{eqnarray}
for $q\in\mathbb{R}^{L}_{++}$. Since $x_{n}(m(q))=h(m(q),v_{n}(m(q)))=\nabla_{1}e(m(q),v_{n}(m(q)))$ (Lemma \ref{lemmaBasicIdentities}) and $\nabla_{1}e(\cdot,u)$ is homogeneous of degree zero, we can write $ \zeta(q)=p\cdot(\nabla_{1}e(q,v_{n}(m(q)))-x_{n}(p))$, $q\in\mathbb{R}^{L}_{++}$, so that
\begin{eqnarray}\label{eqAuxA3}
    \nabla \zeta(q)&=& p\cdot\mathbf{J}\nabla_{1} e(q,v_{n}(m(q)))\begin{bmatrix}
        I\\
        \nabla v_{n}(m(q))\mathbf{J}m(q)
    \end{bmatrix}\nonumber\\
    &=& p\cdot\begin{bmatrix}
        \mathbf{H}_{11} e(q,v_{n}(m(q))) & \mathbf{H}_{21}e(q,v_{n}(m(q)))
    \end{bmatrix}\begin{bmatrix}
        I\\
        \nabla v_{n}(m(q))\mathbf{J}m(q)
    \end{bmatrix}\nonumber\\
    &=&p\cdot(\mathbf{H}_{11} e(q,v_{n}(m(q)))+\mathbf{H}_{21}e(q,v_{n}(m(q)))^{T}\nabla v_{n}(m(q))\mathbf{J}m(q)).
\end{eqnarray}
Notice that Lemma \ref{lemmaBasicIdentities} and (\ref{eqAuxA2}) imply
\begin{eqnarray}\label{eqAuxExtra}
    \nabla v_{n}(m(q))\mathbf{J}m(q)&=&-\frac{\lambda_{n}(m(q))x_{n}(m(q))}{q\cdot x_{n}(p)}\biggr(I - m(q)^{T}x_{n}(p)\biggr)\nonumber\\
    &=&-\frac{\lambda_{n}(m(q))}{q\cdot x_{n}(p)}\biggr(x_{n}(m(q)) - x_{n}(p)\biggr),
\end{eqnarray}
for $q\in\mathbb{R}^{L}_{++}$. Since $m(p)=p$, we have 
\begin{eqnarray}\label{eqAuxA4}
 \nabla v_{n}(m(p))\mathbf{J}m(p)=0,   
\end{eqnarray}
and, therefore, Lemma \ref{lemmaBasicIdentities} and (\ref{eqAuxA3}) imply $\nabla \zeta(p)=0$. It is convenient to define the following auxiliary functions $a:\mathbb{R}^{L}_{++}\rightarrow\mathbb{R}$, $g_{1}:\mathbb{R}^{L}_{++}\rightarrow\mathbb{R}^{L}$ and $g_{2}:\mathbb{R}^{L}_{++}\rightarrow\mathbb{R}^{L}$, given by
\begin{eqnarray*}
    g_{1}(q)&=&p\cdot \mathbf{H}_{11}e(q,v_{n}(m(q)))=\sum^{L}_{i=1}p_{i}\nabla_{1}\partial_{i}e(q,v_{n}(m(q)))\\
    a(q)&=&\frac{\lambda_{n}(m(q))\, p\cdot \mathbf{H}_{21}e(q,v_{n}(m(q)))}{q\cdot x_{n}(p)}\\
    g_{2}(q)&=&-a(q)(\nabla_{1}e(q,v_{n}(m(q))) - x_{n}(p)),
\end{eqnarray*}
so that (\ref{eqAuxA3}) and (\ref{eqAuxExtra}) imply $\nabla \zeta(q)=g_{1}(q)+g_{2}(q)$, $q\in\mathbb{R}^{L}_{++}$. Notice that
\begin{eqnarray*}
    \mathbf{J}g_{1}(q)&=&\sum^{L}_{i=1}p_{i}\mathbf{J}\nabla_{1}\partial_{i}e(q,v_{n}(m(q)))\begin{bmatrix}
        I\\
        \nabla v_{n}(m(q))\mathbf{J}m(q)
    \end{bmatrix}\\
    &=&\sum^{L}_{i=1}p_{i}\biggr(\mathbf{H}_{11} \partial_{i}e(q,v_{n}(m(q)))+\mathbf{H}_{21}\partial_{i}e(q,v_{n}(m(q)))^{T}\nabla v_{n}(m(q))\mathbf{J}m(q)\biggr),
\end{eqnarray*}
for $q\in\mathbb{R}^{L}_{++}$. Therefore, (\ref{eqHessianExpIdentity}) and (\ref{eqAuxA4}) imply
\begin{eqnarray}\label{eqAuxA6}
    \mathbf{J}g_{1}(p)=-\mathbf{H}_{11} e(p,v_{n}(p)).
\end{eqnarray}
Next,
\begin{eqnarray*}
    \mathbf{J}g_{2}(q)&=&-(\nabla_{1}e(q,v_{n}(m(q))) - x_{n}(p))^{T}\nabla a(q)\\
    &&- a(q)\mathbf{J}\nabla_{1}e(q,v_{n}(m(q)))\begin{bmatrix}
        I\\
        \nabla v_{n}(m(q))\mathbf{J}m(q)
    \end{bmatrix}\\
    &=&-(\nabla_{1}e(q,v_{n}(m(q))) - x_{n}(p))^{T}\nabla a(q)\\
    &&- a(q)\biggr(\mathbf{H}_{11} e(q,v_{n}(m(q)))+\mathbf{H}_{21}e(q,v_{n}(m(q)))^{T}\nabla v_{n}(m(q))\mathbf{J}m(q)\biggr),
\end{eqnarray*}
for $q\in\mathbb{R}^{L}_{++}$. Since $m(p)=p$, $\nabla_{1}e(p,v_{n}(p))=x_{n}(p)$ and $p\cdot x_{n}(p)=1$ (Lemma \ref{lemmaBasicIdentities}), we have $a(p)=\lambda_{n}(p)\,p\cdot \mathbf{H}_{21}e(p,v_{n}(p))$, and (\ref{eqAuxA4}) implies
\begin{eqnarray}\label{eqAuxA5}
    \mathbf{J}g_{2}(p)=-\lambda_{n}(p)\, p\cdot \mathbf{H}_{21}e(p,v_{n}(p))\,\mathbf{H}_{11} e(p,v_{n}(p)).
\end{eqnarray}
Next, (\ref{eqAuxA5}) and (\ref{eqAuxA6}) imply
\begin{eqnarray*}
    \mathbf{H}\zeta(p)=\mathbf{J}g_{1}(p)+\mathbf{J}g_{2}(p)=-(1+\lambda_{n}(p) p\cdot \mathbf{H}_{21}e(p,v_{n}(p)))\mathbf{H}_{11} e(p,v_{n}(p)),
\end{eqnarray*}
and Lemma \ref{lemmaBasicIdentities} allows us to write
\begin{eqnarray*}
    \mathbf{H}\zeta(p)&=&-\biggr(1+\lambda_{n}(p) p\cdot \frac{\partial h(p,v_{n}(p))}{\partial u}\biggr)\mathbf{H}_{11}e(p,v_{n}(p))\nonumber\\
    &=&-\biggr(1+\lambda_{n}(p) \frac{\partial e(p,v_{n}(p))}{\partial u}\biggr)\mathbf{H}_{11}e(p,v_{n}(p))\\
    &=&-2\mathbf{H}_{11}e(p,v_{n}(p)),
\end{eqnarray*}
where the last equality is due to item vi of Lemma \ref{lemmaBasicIdentities}.
\end{proof}

\begin{proof}[Proof of Lemma~{\upshape\ref{lemmaLowerBoundHessianUtility}}]
To ease notation, let $u=u(c)=v_{n}(p)$ and $g:\mathbb{R}^{L-1}_{++}\rightarrow\mathcal{I}(c)$, $g(q)=h((q,1),u)$, which is smooth by Lemma \ref{lemmaBasicIdentities} and onto $\mathcal{I}(c)$ (see the global parametrization of $\mathcal{I}(c)$ stated after Definition \ref{defCanonical}). Since $h((q,1),u)$ is an interior solution of the expenditure minimization problem, its first-order conditions give $\nabla u(g(q))=\lambda(q,1)$ for some $\lambda>0$, so that $f(g(q))=q$ for all $q\in\mathbb{R}^{L-1}_{++}$; in particular, $g(f(c))=c$ and $f|_{\mathcal{I}(c)}$ is the inverse of $g(\cdot)$.

(i) Differentiating $f(g(q))=q$ at $q=f(c)$, with $\mathbf{J}g(f(c))=\mathbf{H}_{11}e((f(c),1),u)M^{T}$, we obtain $\mathbf{J}f(c)\mathbf{H}_{11}e((f(c),1),u)M^{T}=I\in\mathbb{R}^{(L-1)\times(L-1)}$, so that $\mathbf{H}_{11}e((f(c),1),u)$ has rank at least $L-1$. Since $h(\cdot,u)$ is homogeneous of degree zero and $(f(c),1)=p/p_{L}$, $\mathbf{H}_{11}e((f(c),1),u)=p_{L}\mathbf{H}_{11}e(p,u)$, so that $\mathbf{H}_{11}e(p,v_{n}(p))$ also has rank at least $L-1$. As $p\mathbf{H}_{11}e(p,v_{n}(p))=0$ (Lemma \ref{lemmaBasicIdentities}) and $\mathbf{H}_{11}e(p,v_{n}(p))$ is symmetric, its kernel is exactly $\mathbf{span}\{p\}$.

(ii) Since $g(\cdot)$ and $f(\cdot)$ are smooth, $f|_{\mathcal{I}(c)}$ is a diffeomorphism onto $\mathbb{R}^{L-1}_{++}$, and its differential at $c$, which is the restriction of $\mathbf{J}f(c)$ to the tangent space $\mathbf{span}\{\nabla u(c)\}^{\perp}$, is injective. Now let $v\in\mathbf{span}\{\nabla u(c)\}^{\perp}$, $v\neq0$, and suppose that $v\mathbf{H}u(c)v^{T}=0$. Since $u(\cdot)$ is strictly quasiconcave, the quadratic form $z\rightarrow z\mathbf{H}u(c)z^{T}$ is negative semi-definite on $\mathbf{span}\{\nabla u(c)\}^{\perp}$, so that $v$ lies in the kernel of its restriction to this subspace, that is, $v\mathbf{H}u(c)z^{T}=0$ for all $z\in\mathbf{span}\{\nabla u(c)\}^{\perp}$ and, therefore, $v\mathbf{H}u(c)=\beta\nabla u(c)=\beta\partial_{L}u(c)(f(c),1)$ for some $\beta\in\mathbb{R}$. As shown in the proof of Theorem \ref{theoFundamentalIneq} below, $\mathbf{J}f(c)=M(I-(f(c),1)^{T}e_{L})\mathbf{H}u(c)/\partial_{L}u(c)$, and since $(I-(f(c),1)^{T}e_{L})(f(c),1)^{T}=0$, we obtain $\mathbf{J}f(c)v^{T}=0$, contradicting the injectivity of $\mathbf{J}f(c)$ on the tangent space. Hence, $v\mathbf{H}u(c)v^{T}<0$.

(iii) Suppose the claim is false. Then, for every $n\geq1$, there is $\tilde{c}_{n}\in B(c,1/n)\cap\mathcal{I}(c)\setminus\{c\}$ with
\begin{eqnarray*}
    \biggr\Vert\frac{(\tilde{c}_{n}-c)}{\Vert \tilde{c}_{n}-c\Vert}\mathbf{H}u(c)(I-e_{L}^{T}(f(c),1))\biggr\Vert^{2}\leq\frac{1}{n}.
\end{eqnarray*}
Let $v_{n}=(\tilde{c}_{n}-c)/\Vert \tilde{c}_{n}-c\Vert$, $n\geq1$. Since $\Vert v_{n}\Vert=1$, we can assume, without loss of generality (i.e., passing to a subsequence), that $\lim_{n\rightarrow\infty}v_{n}=v$, with $\Vert v\Vert=1$, so that $v\neq0$ and
\begin{eqnarray*}
    \Vert v\mathbf{H}u(c)(I-e_{L}^{T}(f(c),1))\Vert^{2}=\lim_{n\rightarrow\infty}\biggr\Vert\frac{(\tilde{c}_{n}-c)}{\Vert \tilde{c}_{n}-c\Vert}\mathbf{H}u(c)(I-e_{L}^{T}(f(c),1))\biggr\Vert^{2}=0.
\end{eqnarray*}
First, notice that
\begin{eqnarray*}
    0=\frac{u(\tilde{c}_{n})-u(c)}{\Vert \tilde{c}_{n}-c\Vert}=\nabla u(c)\cdot v_{n}+ o(1),
\end{eqnarray*}
so that, by taking the limit, $v\in\mathbf{span}\{\nabla u(c)\}^{\perp}$. Next, since $z(I-e_{L}^{T}(f(c),1))=z-z_{L}(f(c),1)$, $z\in\mathbb{R}^{L}$, and $(f(c),1)=\nabla u(c)/\partial_{L}u(c)$, we have
\begin{eqnarray*}
    v\mathbf{H}u(c)(I-e_{L}^{T}(f(c),1))=0\implies v\mathbf{H}u(c)=\beta\nabla u(c),
\end{eqnarray*}
for some $\beta\in\mathbb{R}$. Therefore, $v\mathbf{H}u(c)v^{T}=\beta \nabla u(c)\cdot v=0$, with $v\in \mathbf{span}\{\nabla u(c)\}^{\perp}$ and $v\neq0$, contradicting item ii.
\end{proof}

\begin{proof}[Proof of Corollary~{\upshape\ref{corDepthProj}}]
Since $p=x_{n}^{-1}(c)$, we have $\pi_{i}(q)=\eta(q,1)/\Vert p\Vert$ and $\pi_{o}(q)=\zeta(q,1)/\Vert p\Vert $, $q\in\mathbb{R}^{L-1}_{++}$. Since $h(\cdot,u)$ and $m(\cdot)$ (defined in the proof of Theorem \ref{theoAwesome}) are homogeneous of degree zero, so are $\eta(\cdot)$ and $\zeta(\cdot)$, and we can write
\begin{eqnarray*}
    \pi_{i}(q)=\frac{\eta(p_{L}q,p_{L})}{\Vert p\Vert}\quad\textrm{and}\quad \pi_{o}(q)=\frac{\zeta(p_{L}q,p_{L})}{\Vert p\Vert},
\end{eqnarray*}
for $q\in\mathbb{R}^{L-1}_{++}$. Notice also that Lemma \ref{lemmaBasicIdentities} implies
\begin{eqnarray*}
    p=x^{-1}_{n}(c)=\frac{\nabla u(c)}{\nabla u(c)\cdot c}=\frac{\partial_{L}u(c)}{\nabla u(c)\cdot c}(f(c),1)=(p_{L}f(c),p_{L}).
\end{eqnarray*}
Therefore, Theorem \ref{theoAwesome} implies $\pi_{i}(f(c))=\eta(p_{L}f(c),p_{L})/\Vert p\Vert =\eta(p)/\Vert p\Vert =0$ and $\pi_{o}(f(c))=\zeta(p_{L}f(c),p_{L})/\Vert p\Vert =\zeta(p)/\Vert p\Vert =0$. Next, since the map $q\rightarrow(p_{L}q,p_{L})$ has Jacobian $p_{L}M^{T}\in\mathbb{R}^{L\times(L-1)}$, the chain rule yields $\nabla \pi_{i}(q)= p_{L}\nabla\eta(p_{L}q,p_{L})M^{T}/\Vert p\Vert $ and $\nabla \pi_{o}(q)=p_{L}\nabla\zeta(p_{L}q,p_{L})M^{T}/\Vert p\Vert $, $q\in\mathbb{R}^{L-1}_{++}$, so that Theorem \ref{theoAwesome} implies $\nabla \pi_{i}(f(c))=0$ and $\nabla \pi_{o}(f(c))=0$. Finally, we have
\begin{eqnarray*}
    \mathbf{H}\pi_{i}(q)&=&\frac{p_{L}^{2}}{\Vert p\Vert} M \mathbf{H}\eta(p_{L}q,p_{L}) M^{T}\\
    \mathbf{H}\pi_{o}(q)&=&\frac{p_{L}^{2}}{\Vert p\Vert} M \mathbf{H}\zeta(p_{L}q,p_{L}) M^{T},
\end{eqnarray*}
and Theorem \ref{theoAwesome} allows us to conclude that
\begin{eqnarray*}
    \mathbf{H}\pi_{i}(f(c))&=&-\frac{p_{L}^{2}}{\Vert p\Vert}M \mathbf{H}_{11}e(p,v_{n}(p)) M^{T}\\
    \mathbf{H}\pi_{o}(f(c))&=&-2\frac{p_{L}^{2}}{\Vert p\Vert}M \mathbf{H}_{11}e(p,v_{n}(p)) M^{T}=2\mathbf{H}\pi_{i}(f(c)).
\end{eqnarray*}

\begin{sloppypar}
Finally, let $v\in \mathbf{ker}\,M\mathbf{H}_{11}e(p,v_{n}(p))M^{T}$ and $z=vM=(v,0)\in\mathbb{R}^{L}$. Then, $z\mathbf{H}_{11}e(p,v_{n}(p))M^{T}=0$ and the first $L-1$ coordinates of $z\mathbf{H}_{11}e(p,v_{n}(p))$ vanish, so that $z\mathbf{H}_{11}e(p,v_{n}(p))=\alpha e_{L}$ for some $\alpha\in\mathbb{R}$. Since $\mathbf{H}_{11}e(p,v_{n}(p))$ is symmetric and $p\mathbf{H}_{11}e(p,v_{n}(p))=0$ (Lemma \ref{lemmaBasicIdentities}), we have $\alpha p_{L}=z\mathbf{H}_{11}e(p,v_{n}(p))p^{T}=0$ and, since $p_{L}>0$, $\alpha=0$. Therefore, $z\in\mathbf{ker}\,\mathbf{H}_{11}e(p,v_{n}(p))=\mathbf{span}\{p\}$, where the equality is item i of Lemma \ref{lemmaLowerBoundHessianUtility}. Then, $z_{L}=0$ and $p_{L}>0$ imply $z=0$ and, therefore, $v=0$. Since $e(\cdot,u)$ is concave, $\mathbf{H}_{11}e(p,v_{n}(p))$ is negative semi-definite, we conclude that $M\mathbf{H}_{11}e(p,v_{n}(p))M^{T}\in\mathbb{R}^{(L-1)\times(L-1)}$ is a symmetric negative definite matrix.
\end{sloppypar}

\end{proof}

\begin{proof}[Proof of Proposition~{\upshape\ref{propBoundsHessian}}]
Corollary \ref{corDepthProj} allows us to write the following Taylor expansion of $\pi_{i}(\cdot)$ around $f(c)\in\mathbb{R}^{L-1}_{++}$
\begin{eqnarray}\label{eqAux1BoundsHessian}
    \pi_{i}(q)=-\frac{1}{2}\frac{p_{L}^{2}}{\Vert p\Vert}(q-f(c))M\mathbf{H}_{11}e(p,v_{n}(p))M^{T}(q-f(c))^{T}+\Vert q-f(c)\Vert^{2}o(1)
\end{eqnarray}
for $p=x_{n}^{-1}(c)$. 

Next, let $0<\mu_{\min}(p)\leq\mu_{\max}(p)$ be the smallest and the greatest eigenvalues of the positive definite matrix $-M\mathbf{H}_{11}e(p,v_{n}(p))M^{T}$, so that the first term in the right-hand side of (\ref{eqAux1BoundsHessian}) lies between $p_{L}^{2}\mu_{\min}(p)\Vert q -f(c)\Vert^{2}/(2\Vert p\Vert)$ and $p_{L}^{2}\mu_{\max}(p)\Vert q -f(c)\Vert^{2}/(2\Vert p\Vert)$. Then, taking $\varepsilon>0$ such that the $o(1)$ term in (\ref{eqAux1BoundsHessian}) is bounded by $p_{L}^{2}\mu_{\min}(p)/(4\Vert p\Vert)$ in absolute value for $q\in B(f(c),\varepsilon)$, we obtain
\begin{eqnarray*}
    \frac{p_{L}^{2}}{\Vert p\Vert}\frac{\mu_{\min}(p)}{4}\Vert q -f(c)\Vert^{2}\leq \pi_{i}(q)\leq\frac{p_{L}^{2}}{\Vert p\Vert}\mu_{\max}(p)\Vert q -f(c)\Vert^{2}
\end{eqnarray*}
for $q\in B(f(c),\varepsilon)$. Since Corollary \ref{corDepthProj} implies that the Taylor expansion of $\pi_{o}(\cdot)$ around $f(c)$ has exactly twice the leading term of (\ref{eqAux1BoundsHessian}), the same reasoning (reducing $\varepsilon>0$ if necessary) yields the bounds for $\pi_{o}(\cdot)$.
\end{proof}

\begin{proof}[Proof of Theorem~{\upshape\ref{theoFundamentalIneq}}]
First, notice that if $\tilde{c}=c$, then $\pi_{i}(f(\tilde{c}))=\pi_{i}(f(c))=0$ and the inequality holds trivially. Let $\tilde{c}\in\mathcal{I}(c)\setminus\{c\}$. Proposition \ref{propBoundsHessian} and the continuity of $f(\cdot)$ imply that there is $r>0$ such that if $\tilde{c}\in B(c,r)$, then
\begin{eqnarray}\label{eqAuxT6a}
    \pi_{i}(f(\tilde{c}))\geq \frac{p_{L}^{2}}{\Vert p\Vert}\frac{\mu_{\min}(p)}{4}\Vert f(\tilde{c}) -f(c)\Vert^{2}.
\end{eqnarray}
Define the auxiliary function $\psi:\mathbb{R}^{L}_{++}\rightarrow\mathbb{R}$, $\psi(x)=\Vert f(x)-f(c) \Vert^{2}$. Then, $ \nabla \psi(x)=2(f(x)-f(c))\mathbf{J}f(x)$, for $x\in\mathbb{R}^{L}_{++}$, so that $\psi(c)=0$, $\nabla\psi(c)=0$ and $\mathbf{H} \psi(c)=2\mathbf{J}f(c)^{T}\mathbf{J}f(c)$.

Next, since $f(x)=\nabla u(x)M^{T}/(\nabla u(x)\cdot e_{L})$, $x\in\mathbb{R}^{L}_{++}$, the quotient rule yields
 \begin{eqnarray*}
     \mathbf{J}f(x)&=&\frac{(\nabla u(x)\cdot e_{L}) M\mathbf{H} u(x)-M\nabla u(x)^{T} e_{L}\mathbf{H} u(x)}{(\nabla u(x)\cdot e_{L})^{2}}\\
     &=&\frac{1}{(\nabla u(x)\cdot e_{L})^{2}}M
     \biggr((\nabla u(x)\cdot e_{L}) I-\nabla u(x)^{T}e_{L}\biggr)\mathbf{H} u(x)\\
     &=&\frac{1}{\nabla u(x)\cdot e_{L}}M\biggr(I-(f(x),1)^{T}e_{L}\biggr)\mathbf{H} u(x),
 \end{eqnarray*}
for $x\in\mathbb{R}^{L}_{++}$. Notice that the last row of $I-(f(c),1)^{T}e_{L}$ is zero, so that $M^{T}M(I-(f(c),1)^{T}e_{L})=I-(f(c),1)^{T}e_{L}$ and we can write
\begin{eqnarray*}
    \mathbf{J}f(c)^{T}\mathbf{J}f(c)=\frac{1}{(\nabla u(c)\cdot e_{L})^{2}}\mathbf{H}u(c)\biggr(I-e_{L}^{T}(f(c),1)\biggr)\biggr(I-(f(c),1)^{T}e_{L}\biggr)\mathbf{H}u(c).
\end{eqnarray*}
Therefore, the Taylor expansion of $\psi(\cdot)$ around $c$ reads
\begin{eqnarray}\label{eqTaylorRho}
    \Vert f(x)-f(c) \Vert^{2}=\frac{\Vert x-c\Vert^{2}}{(\nabla u(c)\cdot e_{L})^{2}} \biggr\Vert\frac{(x-c)}{\Vert x-c\Vert}\mathbf{H}u(c)(I-e_{L}^{T}(f(c),1))\biggr\Vert^{2}+\Vert x-c\Vert^{2}o(1),
\end{eqnarray}
for $x\in\mathbb{R}^{L}_{++}$.

Next, item iii of Lemma \ref{lemmaLowerBoundHessianUtility} provides $\varepsilon^{\prime}\in(0,1)$ such that
\begin{eqnarray*}
    \biggr\Vert\frac{(\tilde{c}-c)}{\Vert \tilde{c}-c\Vert}\mathbf{H}u(c)(I-e_{L}^{T}(f(c),1))\biggr\Vert^{2}>\varepsilon^{\prime},
\end{eqnarray*}
for $\tilde{c}\in B(c,\varepsilon^{\prime})\cap\mathcal{I}(c)\setminus\{c\}\subset\mathbb{R}^{L}_{++}$. Reducing $r>0$, if necessary, so that $r\leq\varepsilon^{\prime}$ and the $o(1)$ term in (\ref{eqTaylorRho}) is bounded by $\varepsilon^{\prime}/(2(\nabla u(c)\cdot e_{L})^{2})$ in absolute value for $x\in B(c,r)$, we obtain
\begin{eqnarray*}
    \Vert f(\tilde{c})-f(c) \Vert^{2}\geq \frac{\varepsilon^{\prime}}{2(\nabla u(c)\cdot e_{L})^{2}}\Vert \tilde{c}-c\Vert^{2},
\end{eqnarray*}
for $\tilde{c}\in B(c,r)\cap\mathcal{I}(c)\subset\mathbb{R}^{L}_{++}$, and, therefore, (\ref{eqAuxT6a}) implies
\begin{eqnarray}\label{eqAuxT6b}
    \pi_{i}(f(\tilde{c}))\geq \frac{p_{L}^{2}}{\Vert p\Vert}\frac{\mu_{\min}(p)\varepsilon^{\prime}}{8(\nabla u(c)\cdot e_{L})^{2}}\Vert \tilde{c}-c\Vert^{2}=\frac{\mu_{\min}(p)\varepsilon^{\prime}}{8\Vert p\Vert\lambda_{n}(p)^{2}}\Vert \tilde{c}-c\Vert^{2}=k\Vert\tilde{c}-c\Vert^{2},
\end{eqnarray}
for $\tilde{c}\in B(c,r)\cap\mathcal{I}(c)$, where the second-to-last equality is due to $p_{L}=\partial_{L}u(c)/(\nabla u(c)\cdot c)$ and $\lambda_{n}(p)=c\cdot \nabla u(c)$.

Finally, let $\varepsilon\in(0,1)$ be such that $\varepsilon\Vert c\Vert\leq r$ and $\varepsilon\leq k\Vert c\Vert$. Then, $B(c,\varepsilon\Vert c\Vert)\subseteq B(c,r)$ and $k\geq\varepsilon/\Vert c\Vert$, so that (\ref{eqAuxT6b}) yields $\pi_{i}(f(\tilde{c}))\geq\varepsilon\Vert\tilde{c}-c\Vert^{2}/\Vert c\Vert$ for $\tilde{c}\in B(c,\varepsilon\Vert c\Vert)\cap\mathcal{I}(c)$.
\end{proof}

\begin{proof}[Proof of Theorem~{\upshape\ref{theoUpperBoundCass}}]
First, notice that Proposition \ref{propBoundsHessian} implies that there is $r>0$ such that $q\in B(f(c),r)$ implies $\pi_{i}(q)\leq p_{2}^{2}\mu_{\max}(p)(q-f(c))^{2}/\Vert p\Vert$.

Next, since $\partial_{1}u(c)>0$, the Implicit Function Theorem implies that there are $\delta>0$, an open neighborhood $W\ni c$, and a smooth function $\gamma:(c_{2}-\delta,c_{2}+\delta)\rightarrow\mathbb{R}_{++}$ such that $y\rightarrow(\gamma(y),y)$ is a local parametrization of $\mathcal{I}(c)$ around $c$. Define $\varphi:(c_{2}-\delta,c_{2}+\delta)\rightarrow\mathbb{R}_{++}$ as $\varphi(y)=f(\gamma(y),y)$, which is continuously differentiable, and let $K=\sup\{\vert\varphi^{\prime}(y)\vert\mid \vert y-c_{2}\vert\leq\delta/2\}<\infty$.

Since $f(\cdot)$ is continuous, we can assume, reducing $r>0$ if necessary, that $B(c,r)\subseteq W$, $r<\delta/2$, and $\tilde{c}\in B(c,r)$ implies $f(\tilde{c})\in B(f(c),r)$. Therefore, for $\tilde{c}\in B(c,r)\cap\mathcal{I}(c)$, the Mean Value Theorem yields
    \begin{eqnarray*}
        \vert f(\tilde{c})-f(c)\vert=\vert\varphi(\tilde{c}_{2})-\varphi(c_{2})\vert\leq K\vert \tilde{c}_{2}-c_{2}\vert,
    \end{eqnarray*}
    and, therefore,
    \begin{eqnarray}\label{eqAuxT7}
        \pi_{i}(f(\tilde{c}))\leq \frac{p_{2}^{2}\mu_{\max}(p)}{\Vert p\Vert}(f(\tilde{c})-f(c))^{2}\leq\frac{p_{2}^{2}\mu_{\max}(p)K^{2}}{\Vert p\Vert}(\tilde{c}_{2}-c_{2})^{2}.
    \end{eqnarray}
    Finally, let $\varepsilon\in(0,1)$ be such that $\varepsilon\Vert c\Vert\leq r$ and $\varepsilon\Vert c\Vert\leq\Vert p\Vert/(p_{2}^{2}\mu_{\max}(p)K^{2})$ (the second requirement is vacuous if $K=0$). Then, $B(c,\varepsilon\Vert c\Vert)\subseteq B(c,r)$ and $p_{2}^{2}\mu_{\max}(p)K^{2}/\Vert p\Vert\leq1/(\varepsilon\Vert c\Vert)$, so that (\ref{eqAuxT7}) yields $\pi_{i}(f(\tilde{c}))\leq(\tilde{c}_{2}-c_{2})^{2}/(\varepsilon\Vert c\Vert)$ for $\tilde{c}\in B(c,\varepsilon\Vert c\Vert)\cap\mathcal{I}(c)$.
\end{proof}

\begin{proof}[Proof of Lemma~{\upshape\ref{lemmaInclined}}]
Suppose there is a sequence $\{\tilde{c}_{n}\}_{n\geq2}$ such that, for $n\geq2$: $\tilde{c}_{n}\in B(c,\Vert c\Vert/n)$; $\tilde{c}_{n1}<c_{1}$; $\tilde{c}_{n2}-c_{2}\geq n(c_{1}-\tilde{c}_{n1})$; $p\cdot(\tilde{c}_{n}-c)>0$; and $u(\tilde{c}_{n})\leq u(c)$ (i.e., the claim fails for $\varepsilon=1/n$). Then, $\lim_{n\rightarrow\infty}\tilde{c}_{n}=c$. Let $y_{n}=(\tilde{c}_{n}-c)/\Vert \tilde{c}_{n}-c\Vert$, $n\geq2$, and assume, without loss of generality (i.e., passing to a subsequence), $\lim_{n\rightarrow\infty}y_{n}=y$, with $y_{1}\leq0$, $y_{2}\geq0$ and $\Vert y \Vert=1$. Notice that
    \begin{eqnarray}\label{eqAuxIneq}
        0\geq \frac{1}{c\cdot \nabla u(c)} \frac{u(\tilde{c}_{n})- u(c)}{\Vert \tilde{c}_{n}-c\Vert}=p\cdot y_{n}+o(1)\implies p\cdot y\leq0.
    \end{eqnarray}
    However, $p\cdot y_{n}>0$, $n\geq2$, so that $p\cdot y\geq0$. Therefore, $p\cdot y=0$. Next, notice that $\Vert y_{n}\Vert=1$, $n\geq2$, allow us to write
    \begin{eqnarray*}
        \frac{1}{\vert y_{n1}\vert}\geq\frac{y_{n2}}{\vert y_{n1}\vert}=\frac{\tilde{c}_{n2}-c_{2}}{c_{1}-\tilde{c}_{n1}}\geq n \implies \lim_{n\rightarrow\infty}y_{n1}=y_{1}=0,
    \end{eqnarray*} 
    and, therefore, $y_{2}=1$. But, then, $p\cdot y=0$ implies $p_{2}=0$, absurd.
\end{proof}

\begin{proof}[Proof of Proposition~{\upshape\ref{propFrontierParametrization}}]
    Let $q\in\mathbb{R}^{L-1}_{++}$ and $p=(q,1)$. The first-order conditions of (\ref{PMP}) imply $p=\lambda\nabla F(s(p))$ for some $\lambda>0$, so that $s(p)\in\mathcal{T}$ and $\nabla F(s(p))/\Vert\nabla F(s(p))\Vert=p/\Vert p\Vert$. Conversely, let $z\in\mathcal{T}$ and $p=(f_{f}(z),1)=\nabla F(z)/\partial_{L}F(z)\in\mathbb{R}^{L}_{++}$. Since $Y$ is convex, for $\tilde{z}\in Y$ the segment $z+s(\tilde{z}-z)$, $s\in[0,1]$, lies in $Y$, so that $F(z+s(\tilde{z}-z))\leq0=F(z)$, $s\in[0,1]$, and differentiating at $s=0$ yields $\nabla F(z)\cdot(\tilde{z}-z)\leq0$, that is, $p\cdot\tilde{z}\leq p\cdot z$. Therefore, $z$ solves (\ref{PMP}) at prices $p$ and, by uniqueness, $z=s(p)=s(f_{f}(z),1)$.
\end{proof}

\begin{proof}[Proof of Theorem~{\upshape\ref{theoProduction}}]
The proof is analogous to the one of Theorem \ref{theoAwesome}. Clearly, $\xi(p)=p\cdot(s(p)-s(p))=0$. Hotelling's lemma implies $\xi(q)=\Pi(p)-p\cdot\nabla\Pi(q)$, so that
    \begin{eqnarray}\label{eqAuxProd1}
        \nabla\xi(q)=-p\mathbf{H}\Pi(q)=-\sum^{L}_{i=1}p_{i}\nabla\partial_{i}\Pi(q),
    \end{eqnarray}
    for $q\in\mathbb{R}^{L}_{++}$, and $p\mathbf{H}\Pi(p)=0$ implies $\nabla\xi(p)=0$. Next, since $q\mathbf{H}\Pi(q)=0$, $q\in\mathbb{R}^{L}_{++}$, we can differentiate this expression to obtain
    \begin{eqnarray*}
        \mathbf{J}(q\mathbf{H}\Pi(q))=\mathbf{H}\Pi(q)+\sum^{L}_{i=1}q_{i}\mathbf{H}\partial_{i}\Pi(q)=0,
    \end{eqnarray*}
    for $q\in\mathbb{R}^{L}_{++}$, and, taking $q=p$,
    \begin{eqnarray}\label{eqAuxProd2}
        \sum^{L}_{i=1}p_{i}\mathbf{H}\partial_{i}\Pi(p)=-\mathbf{H}\Pi(p).
    \end{eqnarray}
    Finally, differentiating (\ref{eqAuxProd1}) yields $\mathbf{H}\xi(q)=-\sum^{L}_{i=1}p_{i}\mathbf{H}\partial_{i}\Pi(q)$, $q\in\mathbb{R}^{L}_{++}$, and (\ref{eqAuxProd2}) allows us to conclude that $\mathbf{H}\xi(p)=\mathbf{H}\Pi(p)=\mathbf{J}s(p)$.
\end{proof}

\begin{proof}[Proof of Corollary~{\upshape\ref{corDepthProjProduction}}]
    Since $p=(f_{f}(y),1)$, Proposition \ref{propFrontierParametrization} implies $y=s(p)$ and $\nabla F(y)\parallel p$. Therefore, with $\xi(\cdot)$ defined for this $p$, we have $\pi_{f}(q)=\xi(q,1)/\Vert p\Vert$ for all $q\in\mathbb{R}^{L-1}_{++}$. Direct differentiation then yields
    \begin{eqnarray*}
        \nabla\pi_{f}(q)=\frac{\nabla\xi(q,1)M^{T}}{\Vert p\Vert}\quad\textrm{and}\quad\mathbf{H}\pi_{f}(q)=\frac{M\mathbf{H}\xi(q,1)M^{T}}{\Vert p\Vert},
    \end{eqnarray*}
    for $q\in\mathbb{R}^{L-1}_{++}$. Since $(f_{f}(y),1)=p$, Theorem \ref{theoProduction} implies $\pi_{f}(f_{f}(y))=\xi(p)/\Vert p\Vert=0$, $\nabla\pi_{f}(f_{f}(y))=\nabla\xi(p)M^{T}/\Vert p\Vert=0$ and $\mathbf{H}\pi_{f}(f_{f}(y))=M\mathbf{H}\Pi(p)M^{T}/\Vert p\Vert$.
\end{proof}

\begin{proof}[Proof of Proposition~{\upshape\ref{propBoundsHessianProduction}}]
\begin{sloppypar}
By Proposition \ref{propFrontierParametrization}, $f_{f}(s(q,1))=q$, $q\in\mathbb{R}^{L-1}_{++}$, so that $\mathbf{J}f_{f}(y)\mathbf{J}s(p)M^{T}=I\in\mathbb{R}^{(L-1)\times(L-1)}$. Since $\mathbf{J}s(p)M^{T}\in\mathbb{R}^{L\times(L-1)}$ has rank $L-1$, $\mathbf{H}\Pi(p)=\mathbf{J}s(p)$ has rank at least $L-1$. Furthermore, $p\mathbf{H}\Pi(p)=0$ and $\mathbf{H}\Pi(p)$ symmetric imply $\mathbf{ker}\,\mathbf{H}\Pi(p)=\mathbf{span}\{p\}$. 
    
    Next, let $v\in\mathbf{ker}\,M\mathbf{H}\Pi(p)M^{T}$ and $z=vM=(v,0)\in\mathbb{R}^{L}$. Then, $z\mathbf{H}\Pi(p)=\alpha e_{L}$, $\alpha\in\mathbb{R}$, and $\alpha p_{L}=z\mathbf{H}\Pi(p)p^{T}=0$ implies $\alpha=0$, since $p_{L}=1$. Therefore, $z\in\mathbf{span}\{p\}$ and $z_{L}=0$ imply $z=0$ and $v=0$. Since $\Pi(\cdot)$ is convex, $\mathbf{H}\Pi(p)$ is positive semi-definite, and we conclude that $M\mathbf{H}\Pi(p)M^{T}$ is symmetric and positive definite.
\end{sloppypar}

    Next, Corollary \ref{corDepthProjProduction} allows us to write the following Taylor expansion of $\pi_{f}(\cdot)$ around $f_{f}(y)$,
    \begin{eqnarray*}
        \pi_{f}(q)=\frac{1}{2\Vert p\Vert}(q-f_{f}(y))M\mathbf{H}\Pi(p)M^{T}(q-f_{f}(y))^{T}+\Vert q-f_{f}(y)\Vert^{2}o(1),
    \end{eqnarray*}
    whose first term lies between $\mu^{f}_{\min}(p)\Vert q-f_{f}(y)\Vert^{2}/(2\Vert p\Vert)$ and $\mu^{f}_{\max}(p)\Vert q-f_{f}(y)\Vert^{2}/(2\Vert p\Vert)$. Taking $\varepsilon>0$ such that the $o(1)$ term is bounded by $\mu^{f}_{\min}(p)/(4\Vert p\Vert)$ in absolute value for $q\in B(f_{f}(y),\varepsilon)$ yields the result.
\end{proof}

\begin{proof}[Proof of Theorem~{\upshape\ref{theoFundamentalIneqProduction}}]
    The proof follows that of Theorem \ref{theoFundamentalIneq}. Let $p=(f_{f}(y),1)$. Since $\pi_{f}(f_{f}(y))=0$, the inequality holds trivially for $\tilde{y}=y$, so let $\tilde{y}\in\mathcal{T}\setminus\{y\}$. Proposition \ref{propBoundsHessianProduction} and the continuity of $f_{f}(\cdot)$ imply that there is $r>0$ such that if $\tilde{y}\in B(y,r)$, then
    \begin{eqnarray}\label{eqAuxT6Prod}
        \pi_{f}(f_{f}(\tilde{y}))\geq\frac{\mu^{f}_{\min}(p)}{4\Vert p\Vert}\Vert f_{f}(\tilde{y})-f_{f}(y)\Vert^{2}.
    \end{eqnarray}
    The computation in the proof of Theorem \ref{theoFundamentalIneq} (with $F(\cdot)$ in place of $u(\cdot)$) gives
    \begin{eqnarray*}
        \mathbf{J}f_{f}(y)=\frac{1}{\nabla F(y)\cdot e_{L}}M\biggr(I-(f_{f}(y),1)^{T}e_{L}\biggr)\mathbf{H}F(y),
    \end{eqnarray*}
    and the Taylor expansion of $\psi(x)=\Vert f_{f}(x)-f_{f}(y)\Vert^{2}$ around $y$ reads
    \begin{eqnarray}\label{eqTaylorRhoProd}
        \frac{\Vert f_{f}(x)-f_{f}(y)\Vert^{2}}{\Vert x-y\Vert^{2}}=\frac{1}{(\nabla F(y)\cdot e_{L})^{2}}\biggr\Vert\frac{(x-y)}{\Vert x-y\Vert}\mathbf{H}F(y)(I-e^{T}_{L}(f_{f}(y),1))\biggr\Vert^{2}+o(1).
    \end{eqnarray}
    Notice that there is $\varepsilon^{\prime}>0$ such that
    \begin{eqnarray}\label{eqClaimProd}
        \biggr\Vert\frac{(\tilde{y}-y)}{\Vert\tilde{y}-y\Vert}\mathbf{H}F(y)(I-e^{T}_{L}(f_{f}(y),1))\biggr\Vert^{2}>\varepsilon^{\prime},
    \end{eqnarray}
    for all $\tilde{y}\in B(y,\varepsilon^{\prime})\cap\mathcal{T}\setminus\{y\}$. Otherwise, as in the proof of Lemma \ref{lemmaLowerBoundHessianUtility}, there is a sequence $\tilde{y}_{n}\in B(y,1/n)\cap\mathcal{T}\setminus\{y\}$, $n\geq1$, such that $\lim_{n\rightarrow\infty}z_{n}=\lim_{n\rightarrow\infty}(\tilde{y}_{n}-y)/\Vert\tilde{y}_{n}-y\Vert=z$, $\Vert z\Vert=1$, $\nabla F(y)\cdot z=\lim_{n\rightarrow\infty}(F(\tilde{y}_{n})-F(y))/\Vert\tilde{y}_{n}-y\Vert=0$, and $z\mathbf{H}F(y)(I-e^{T}_{L}(f_{f}(y),1))=0$. Transposing the last identity and noticing that the last row of $I-(f_{f}(y),1)^{T}e_{L}$ is zero, we obtain $(I-(f_{f}(y),1)^{T}e_{L})\mathbf{H}F(y)z^{T}=0$ and, therefore, $\mathbf{J}f_{f}(y)z^{T}=0$. However, $z$ belongs to the tangent space $\mathbf{span}\{\nabla F(y)\}^{\perp}$ of $\mathcal{T}$ at $y$, and the restriction of $\mathbf{J}f_{f}(y)$ to this tangent space is the differential at $y$ of the restriction of $f_{f}(\cdot)$ to $\mathcal{T}$, which is a diffeomorphism onto $\mathbb{R}^{L-1}_{++}$ by Proposition \ref{propFrontierParametrization} and the smoothness of $s(\cdot)$. Hence, $z=0$, absurd.

    Next, reduce $r>0$ so that $r\leq\varepsilon^{\prime}$ and the $o(1)$ term in (\ref{eqTaylorRhoProd}) is bounded by $\varepsilon^{\prime}/(2(\nabla F(y)\cdot e_{L})^{2})$ in absolute value for $x\in B(y,r)$. Then, (\ref{eqTaylorRhoProd}) and (\ref{eqClaimProd}) give $\Vert f_{f}(\tilde{y})-f_{f}(y)\Vert^{2}\geq\varepsilon^{\prime}\Vert\tilde{y}-y\Vert^{2}/(2(\nabla F(y)\cdot e_{L})^{2})$ for $\tilde{y}\in B(y,r)\cap\mathcal{T}$, and (\ref{eqAuxT6Prod}) yields $\pi_{f}(f_{f}(\tilde{y}))\geq k\Vert\tilde{y}-y\Vert^{2}$ for $\tilde{y}\in B(y,r)\cap\mathcal{T}$, with $k=\mu^{f}_{\min}(p)\varepsilon^{\prime}/(8\Vert p\Vert(\nabla F(y)\cdot e_{L})^{2})$.

    Finally, since $y\neq0$, let $\varepsilon\in(0,1)$ be such that $\varepsilon\Vert y\Vert\leq r$ and $\varepsilon\leq k\Vert y\Vert$. Then, $B(y,\varepsilon\Vert y\Vert)\subseteq B(y,r)$ and $k\geq\varepsilon/\Vert y\Vert$, so that $\pi_{f}(f_{f}(\tilde{y}))\geq\varepsilon\Vert\tilde{y}-y\Vert^{2}/\Vert y\Vert$ for $\tilde{y}\in B(y,\varepsilon\Vert y\Vert)\cap\mathcal{T}$.
\end{proof}

The proofs of Lemma \ref{lemmaSightedEquilibria} and Theorem \ref{theoExistence} rely on the following claim, which converts the market clearing conditions into bounds on the prices faced by each generation.

\begin{claim}\label{claimBoxes}
    Under Assumptions \ref{assTechnology} and \ref{assBorderPrices}, let $T\geq1$ and $p=(p_{1},\ldots,p_{T+1})\in\mathbb{R}^{L_{1}+\cdots+L_{T+1}}_{++}$ satisfy (\ref{eqMarketClearing}) for $1\leq t\leq T$. Then, $\Vert\sum_{h\in G_{t}}x^{h}/H_{t}\Vert_{\infty}\leq b_{t}$ and $(p_{t},p_{t+1})/p_{t1}\in\mathcal{B}_{t}(\sigma_{t})$, $0\leq t\leq T-1$, where $x^{h}(\cdot)$ is evaluated at the prices faced by $h$ (for $t=0$, $(p_{0},p_{1})$ must be read as $p_{1}$).
\end{claim}
\begin{proof}
    Let $y$ be the finite production allocation of $\mathcal{E}$ given by $y^{j}=s^{j}(p_{t},p_{t+1})$ for $j\in G^{f}_{t}$, $0\leq t\leq T$ (with $y^{j}=s^{j}(p_{1})$ for $j\in G^{f}_{0}$), and $y^{j}=0$ otherwise. For $1\leq t\leq T$, (\ref{eqMarketClearing}) and $x^{h}(\cdot)\geq0$ imply $\omega_{t}+\sum_{j}y^{j}_{t}=\sum_{h\in G_{t-1}\cup G_{t}}x^{h}_{t}\geq0$, so that $y$ is feasible up to period $T$ and item iii of Assumption \ref{assTechnology} yields
    \begin{eqnarray*}
        0\leq\sum_{h\in G_{t}}x^{h}_{t}\leq\sum_{h\in G_{t-1}\cup G_{t}}x^{h}_{t}=\omega_{t}+\sum_{j}y^{j}_{t}\leq\omega_{t}+\bar{y}_{t}(1,\ldots,1),
    \end{eqnarray*}
    for $1\leq t\leq T$, and the same bound holds for $\sum_{h\in G_{t-1}}x^{h}_{t}$. Therefore, $\Vert\sum_{h\in G_{t}}x^{h}_{t}\Vert_{\infty}\leq\Vert\omega_{t}\Vert_{\infty}+\bar{y}_{t}$ and $\Vert\sum_{h\in G_{t}}x^{h}_{t+1}\Vert_{\infty}\leq\Vert\omega_{t+1}\Vert_{\infty}+\bar{y}_{t+1}$ whenever $t\leq T$ and $t+1\leq T$, respectively. For $0\leq t\leq T-1$, both bounds are available and the definition of $b_{t}$ gives $\Vert\sum_{h\in G_{t}}x^{h}/H_{t}\Vert_{\infty}\leq b_{t}$ (for $t=0$, only the first bound is needed). The contrapositive of Assumption \ref{assBorderPrices} then gives $(p_{t},p_{t+1})/p_{t1}\in\mathcal{B}_{t}(\sigma_{t})$, $0\leq t\leq T-1$.
\end{proof}

\begin{proof}[Proof of Lemma~{\upshape\ref{lemmaSightedEquilibria}}]
    First, I show that $\mathcal{V}_{J}\subseteq\prod^{J+1}_{t=1}\mathcal{K}_{t}$. Let $p\in\mathcal{V}_{J}$. By Claim \ref{claimBoxes} with $T=J$ and Definition \ref{defSightedEquilibria}, $(p_{t},p_{t+1})/p_{t1}\in\mathcal{B}_{t}(\sigma_{t})$ for $0\leq t\leq J$. Since $p_{11}=1$, $p_{1}\in\mathcal{B}_{0}(\sigma_{0})$ reads $\kappa_{1}=\sigma_{0}\leq p_{1i}\leq\kappa^{-1}_{1}$, $1\leq i\leq L_{1}$, so that $p_{1}\in\mathcal{K}_{1}$. Suppose that $p_{t}\in\mathcal{K}_{t}$ for some $1\leq t\leq J$. Then, $\kappa_{t}\leq p_{t1}\leq\kappa^{-1}_{t}$ and $(p_{t},p_{t+1})/p_{t1}\in\mathcal{B}_{t}(\sigma_{t})$ imply
    \begin{eqnarray*}
        \kappa_{t+1}=\sigma_{t}\kappa_{t}\leq\sigma_{t}p_{t1}\leq p_{(t+1)i}\leq \frac{p_{t1}}{\sigma_{t}}\leq\frac{1}{\sigma_{t}\kappa_{t}}=\kappa^{-1}_{t+1},
    \end{eqnarray*}
    for $1\leq i\leq L_{t+1}$, so that $p_{t+1}\in\mathcal{K}_{t+1}$. By induction, $p_{t}\in\mathcal{K}_{t}$, $1\leq t\leq J+1$.

    Next, $\mathcal{V}_{J}$ is compact. It is bounded by the previous paragraph. Let $\{p^{n}\}_{n\geq1}$ be a sequence in $\mathcal{V}_{J}$ with $\lim_{n\rightarrow\infty}p^{n}=p\in\mathbb{R}^{L_{1}+\cdots+L_{J+1}}$. Since $\mathcal{K}_{t}$ is closed, $p_{t}=\lim_{n\rightarrow\infty}p^{n}_{t}\in\mathcal{K}_{t}\subset\mathbb{R}^{L_{t}}_{++}$, $1\leq t\leq J+1$, and $p_{11}=\lim_{n\rightarrow\infty}p^{n}_{11}=1$. Since $x^{h}(\cdot)$ and $s^{j}(\cdot)$ are continuous on strictly positive prices (Lemma \ref{lemmaBasicIdentities}, the continuity of $w_{h}(\cdot)$, and the smoothness of $s^{j}(\cdot)$ assumed in Section \ref{sec3}), (\ref{eqMarketClearing}) holds at $p$ for $1\leq t\leq J$, and, since $\mathcal{B}_{J}(\sigma_{J})$ is closed and $p_{J1}>0$, $(p_{J},p_{J+1})/p_{J1}=\lim_{n\rightarrow\infty}(p^{n}_{J},p^{n}_{J+1})/p^{n}_{J1}\in\mathcal{B}_{J}(\sigma_{J})$. Therefore, $p\in\mathcal{V}_{J}$ and $\mathcal{V}_{J}$ is closed.

    Next, $\mathcal{V}_{J}$ is nonempty. The insulated economy $\mathcal{E}_{J}$ satisfies all the assumptions of \textcite[Chapters 3 and 4]{ArrowHahn_1971}: for every firm, $0\in\tilde{Y}^{j}$ (since $F^{j}(0)=0$), $\tilde{Y}^{j}$ is closed (since $F^{j}(\cdot)$ is continuous) and convex (since $Y^{j}$ is convex and the constraints on the non-selected coordinates are linear), which are Assumptions A.3.1--A.3.3, while A.3.4 and A.3.5 are items i and ii of Assumption \ref{assTechnology}; for every household, the consumption set is a nonnegative orthant (A.4.1), the survival assumption A.4.2 is satisfied by the null consumption vector, the wealth is the value of the endowment plus the shares $\theta_{jh}$ of profits, with $\sum_{h}\theta_{jh}=1$ (A.4.3), and $U^{h}(\cdot)$ is continuous, semi-strictly quasiconcave and nonsatiated, since $u^{h}(\cdot)$ is strictly increasing (A.4.4).

    Under item i of Assumption \ref{assResourceRelated}, \textcite[Theorem 5, p. 119]{ArrowHahn_1971} and its proof yield a competitive equilibrium $(p^{*},c^{*},y^{*})$ of $\mathcal{E}_{J}$ in which every household has a strictly positive wealth: $p^{*}\in\mathbb{R}^{L_{1}+\cdots+L_{J+1}}_{+}\setminus\{0\}$; $y^{*j}$ maximizes $p^{*}\cdot y$ over $\tilde{Y}^{j}$; $c^{*h}$ maximizes $U^{h}(\cdot)$ subject to $p^{*}\cdot c\leq w^{*}_{h}=p^{*}\cdot\omega^{h}+\sum_{j}\theta_{jh}p^{*}\cdot y^{*j}$, with $w^{*}_{h}>0$; and $\sum_{h}c^{*h}_{t}\leq\omega_{t}+\sum_{j}y^{*j}_{t}$, $1\leq t\leq J+1$, where the sums run over the agents of $\mathcal{E}_{J}$.

    I claim that $p^{*}\in\mathbb{R}^{L_{1}+\cdots+L_{J+1}}_{++}$. Otherwise, $p^{*}_{ti}=0$ for some $1\leq t\leq J+1$ and $1\leq i\leq L_{t}$, and item ii of Assumption \ref{assResourceRelated} gives $h\in G_{t-1}\cup G_{t}$ that desires the $i$-th commodity of period $t$. Since $u^{h}(\cdot)$ is strictly increasing, adding any positive amount of this commodity to $c^{*h}$ increases $U^{h}(\cdot)$ at no cost, contradicting the utility maximization of $h$.

    Since $p^{*}\gg0$ and $u^{h}(\cdot)$ is strictly increasing, $c^{*h}$ places zero in every non-desired coordinate (otherwise, the corresponding expenditure could be transferred to a desired commodity, increasing $U^{h}(\cdot)$) and, by the strict quasiconcavity of $u^{h}(\cdot)$, $c^{*h}$ is the unique Walrasian demand of $h$ at the prices he faces and wealth $w^{*}_{h}$. Also, by the strict convexity of $Y^{j}$, $y^{*j}=s^{j}(\cdot)$ and $p^{*}\cdot y^{*j}=\Pi^{j}(\cdot)$ at the prices faced by $j$, so that $w^{*}_{h}=w_{h}(\cdot)$ and $c^{*h}=x^{h}(\cdot)$, $h\in\bigcup^{J}_{t=0}G_{t}$.

    Finally, the non-satiation of $U^{h}(\cdot)$ implies $p^{*}\cdot c^{*h}=w^{*}_{h}$, and summing over $h$ (using $\sum_{h}\theta_{jh}=1$) gives $p^{*}\cdot(\sum_{h}c^{*h}-\sum_{h}\omega^{h}-\sum_{j}y^{*j})=0$. Since the vector in parentheses is nonpositive and $p^{*}\gg0$, it must be zero, that is, (\ref{eqMarketClearing}) holds at $p^{*}$ for $1\leq t\leq J$, and the market of period $J+1$ of $\mathcal{E}_{J}$ clears:
    \begin{eqnarray}\label{eqLastPeriodClearing}
        \sum_{h\in G_{J}}x^{h}_{J+1}(p^{*}_{J},p^{*}_{J+1})=\sum_{h\in G_{J}}\omega^{h}_{J+1}+\sum_{j\in G^{f}_{J}}s^{j}_{J+1}(p^{*}_{J},p^{*}_{J+1}).
    \end{eqnarray}
    Since $x^{h}(\cdot)$ and $s^{j}(\cdot)$ are homogeneous of degree zero, all these conditions still hold after the normalization $p^{*}_{11}=1$, which I adopt.

    It remains to verify that $(p^{*}_{J},p^{*}_{J+1})/p^{*}_{J1}\in\mathcal{B}_{J}(\sigma_{J})$. Let $y$ be the finite production allocation of $\mathcal{E}$ that coincides with $y^{*}$ on the firms of $\mathcal{E}_{J}$ and is zero elsewhere. By (\ref{eqMarketClearing}) for $1\leq t\leq J$ and (\ref{eqLastPeriodClearing}), $\omega_{t}+\sum_{j}y^{j}_{t}\geq0$ for $1\leq t\leq J+1$ (for $t=J+1$, $\omega_{J+1}+\sum_{j}y^{j}_{J+1}\geq\sum_{h\in G_{J}}\omega^{h}_{J+1}+\sum_{j\in G^{f}_{J}}y^{*j}_{J+1}=\sum_{h\in G_{J}}x^{h}_{J+1}\geq0$), so that $y$ is feasible up to period $J+1$ and item iii of Assumption \ref{assTechnology} bounds its net output in periods $J$ and $J+1$.

    As in the proof of Claim \ref{claimBoxes}, this gives $\Vert\sum_{h\in G_{J}}x^{h}_{J}\Vert_{\infty}\leq\Vert\omega_{J}\Vert_{\infty}+\bar{y}_{J}$ and, by (\ref{eqLastPeriodClearing}), $\Vert\sum_{h\in G_{J}}x^{h}_{J+1}\Vert_{\infty}\leq\Vert\omega_{J+1}\Vert_{\infty}+\bar{y}_{J+1}$, so that $\Vert\sum_{h\in G_{J}}x^{h}/H_{J}\Vert_{\infty}\leq b_{J}$ and Assumption \ref{assBorderPrices} yields $(p^{*}_{J},p^{*}_{J+1})/p^{*}_{J1}\in\mathcal{B}_{J}(\sigma_{J})$. We conclude that $p^{*}\in\mathcal{V}_{J}$.

    Finally, let $p\in\mathcal{V}^{\infty}_{J+1}$. Then, $(p_{1},\ldots,p_{J+2})\in\mathcal{V}_{J+1}$, so that $p_{11}=1$ and (\ref{eqMarketClearing}) holds for $1\leq t\leq J+1$. Claim \ref{claimBoxes} with $T=J+1$ gives $(p_{J},p_{J+1})/p_{J1}\in\mathcal{B}_{J}(\sigma_{J})$, and, therefore, $(p_{1},\ldots,p_{J+1})\in\mathcal{V}_{J}$ and $p\in\mathcal{V}^{\infty}_{J}$. We conclude that $\mathcal{V}^{\infty}_{J+1}\subseteq\mathcal{V}^{\infty}_{J}$.
\end{proof}

\begin{proof}[Proof of Theorem~{\upshape\ref{theoExistence}}]
    By Lemma \ref{lemmaSightedEquilibria}, $\mathcal{V}_{J}\subseteq\prod^{J+1}_{t=1}\mathcal{K}_{t}$, so that
    \begin{eqnarray*}
        \mathcal{V}^{\infty}_{J}\cap\mathcal{K}=\mathcal{V}_{J}\times\prod^{\infty}_{t=J+2}\mathcal{K}_{t},
    \end{eqnarray*}
    for $J\geq1$, which is nonempty and compact by Lemma \ref{lemmaSightedEquilibria} and Tychonoff's Theorem. Also, $\{\mathcal{V}^{\infty}_{J}\cap\mathcal{K}\}_{J\geq1}$ is a nested sequence of sets, since $\{\mathcal{V}^{\infty}_{J}\}_{J\geq1}$ is.

    First, I show that $\mathcal{M}\subseteq\lim_{J\rightarrow\infty}(\mathcal{V}^{\infty}_{J}\cap\mathcal{K})$. Let $p\in\mathcal{M}$. Since (\ref{eqMarketClearing}) holds for all $t\geq1$, Claim \ref{claimBoxes} (applied for every $T\geq1$) gives $(p_{t},p_{t+1})/p_{t1}\in\mathcal{B}_{t}(\sigma_{t})$, $t\geq0$, and the induction in the proof of Lemma \ref{lemmaSightedEquilibria} gives $p_{t}\in\mathcal{K}_{t}$, $t\geq1$, that is, $p\in\mathcal{K}$. Moreover, $(p_{1},\ldots,p_{J+1})\in\mathcal{V}_{J}$, $J\geq1$, so that $p\in\mathcal{V}^{\infty}_{J}\cap\mathcal{K}$, $J\geq1$.

    Next, let $p\in\lim_{J\rightarrow\infty}(\mathcal{V}^{\infty}_{J}\cap\mathcal{K})$. Then, $p\in\mathcal{K}\subset\mathbb{R}^{\infty}_{++}$, $p_{11}=1$, and, for every $t\geq1$, $(p_{1},\ldots,p_{t+2})\in\mathcal{V}_{t+1}$ implies that (\ref{eqMarketClearing}) holds at $t$. Therefore, $p\in\mathcal{M}$ and $\mathcal{M}=\lim_{J\rightarrow\infty}(\mathcal{V}^{\infty}_{J}\cap\mathcal{K})$. Since $\mathbb{R}^{\infty}$ is a Hausdorff space, Cantor's Intersection Theorem implies that the intersection of the nested sequence of nonempty compact sets $\{\mathcal{V}^{\infty}_{J}\cap\mathcal{K}\}_{J\geq1}$ is nonempty and compact, so that $\mathcal{M}$ is nonempty and compact. Finally, if $p\in\lim_{J\rightarrow\infty}\mathcal{V}^{\infty}_{J}$, then $(p_{1},\ldots,p_{J+1})\in\mathcal{V}_{J}\subseteq\prod^{J+1}_{t=1}\mathcal{K}_{t}$ for every $J\geq1$, so that $p\in\mathcal{K}$. Hence,
    \begin{eqnarray*}
        \lim_{J\rightarrow\infty}\mathcal{V}^{\infty}_{J}=\lim_{J\rightarrow\infty}(\mathcal{V}^{\infty}_{J}\cap\mathcal{K})=\mathcal{M}\subseteq\mathcal{K}.
    \end{eqnarray*}
\end{proof}

%-------------------------------------------
% References
%-------------------------------------------

% Print bibliography
\printbibliography
\end{document}